\documentclass[letter,11pt]{article}
\usepackage{microtype}
\usepackage{a4wide}
\usepackage{color}
\usepackage{xcolor}
\usepackage{tikz}
\usetikzlibrary{arrows.meta, automata, positioning}
\usepackage{hyperref}
\hypersetup{colorlinks=true,citecolor=blue,linkcolor=blue,urlcolor=blue}
\fussy
\usepackage{amsmath}
\usepackage{amssymb}
\usepackage{amsthm} 

\usepackage{todonotes}
\theoremstyle{plain}

\newtheorem*{thm*}{Theorem}

\newtheorem*{lem*}{Lemma}

\newtheorem*{prop*}{Proposition}

\newtheorem*{cor*}{Corollary}

\newtheorem*{conj*}{Conjecture}

\DeclareMathOperator{\PCSP}{PCSP}
\DeclareMathOperator{\CSP}{CSP}

\newcommand{\A}{\mathbf{A}}
\newcommand{\B}{\mathbf{B}}

\newcommand{\X}{\mathbf{X}}

\usepackage[english]{babel}

\usepackage[utf8]{inputenc}

\usepackage{mathtools}
\usepackage[T1]{fontenc}
\usepackage{amsmath}
\usepackage{amsthm,amssymb}
\usepackage{graphicx}
\usepackage{thmtools}
\usepackage{thm-restate}
\usepackage{cleveref}
\usepackage{wrapfig}
\usepackage{tikz}
\usepackage{cite}
\usepackage{xcolor}

\newcommand{\fiPCSP}{\operatorname{fiPCSP}}
\newcommand{\AND}{\operatorname{AND}}

\newcommand{\E}{\mathbb{E}}
\newcommand{\one}{\mathbf{1}}
\newcommand{\R}{\mathbb R}

\newcommand{\eps}{\epsilon}

\newcommand{\bx}{\mathbf x}
\newcommand{\bu}{\mathbf u}
\newcommand{\bv}{\mathbf v}
\newcommand{\bV}{\mathbf V}
\newcommand{\bW}{\mathbf W}
\newcommand{\bw}{\mathbf w}
\newcommand{\ba}{\mathbf a}
\newcommand{\bb}{\mathbf b}

\newcommand{\be}{\mathbf e}

\newcommand{\N}{\mathbb N}

\newcommand{\cN}{\mathcal{N}}
\newcommand{\cP}{\mathcal{P}}

\newcommand{\cV}{\mathcal{V}}
\newcommand{\cW}{\mathcal{W}}
\newcommand{\cX}{\mathcal{X}}

\newcommand{\EQ}{\operatorname{EQ}}
\newcommand{\cI}{\mathcal{I}}
\newcommand{\bA}{\mathbf A}
\newcommand{\bB}{\mathbf B}
\newcommand{\bC}{\mathbf C}
\newcommand{\bD}{\mathbf D}

\newcommand{\Z}{\mathbb Z}
\renewcommand{\S}{\mathbb S}
\newcommand{\COR}{\operatorname{COR}}

\newcommand{\bX}{\mathbf X}
\newcommand{\bI}{\mathbf I}

\newcommand{\cR}{\mathcal R}
\newcommand{\comp}{\operatorname{comp}}
\newcommand{\sound}{\operatorname{sound}}

\newcommand{\REQ}{\operatorname{REQ}}

\newcommand{\BCR}{\operatorname{BCR}}

\newcommand{\MAJ}{\operatorname{MAJ}}

\newcommand{\AT}{\operatorname{AT}}
\newcommand{\PAR}{\operatorname{Parity}}
\renewcommand{\PAR}{\operatorname{PAR}}

\newcommand{\PLUR}{\operatorname{PLUR}}
\newcommand{\plur}{\operatorname{plur}}

\newcommand{\NAE}{\operatorname{NAE}}
\newcommand{\OR}{\operatorname{OR}}
\newcommand{\Pol}{\operatorname{Pol}}

\newcommand{\bU}{\mathbf U}
\newcommand{\cU}{\mathcal U}

\newcommand{\val}{\operatorname{val}}
\newcommand{\Mat}{\operatorname{Mat}}
\newcommand{\V}{\mathbb V}
\newcommand{\W}{\mathbb W}

\newcommand{\ang}[2]{\langle #1, #2 \rangle}

\newtheorem{theorem}{Theorem}[section]
\newtheorem{claim}[theorem]{Claim}
\newtheorem{lemma}[theorem]{Lemma}
\newtheorem{definition}[theorem]{Definition}

\newtheorem{corollary}[theorem]{Corollary}

\newtheorem{proposition}[theorem]{Proposition}
\newtheorem{remark}[theorem]{Remark}
\newtheorem{question}[theorem]{Question}

\newcommand{\One}{\textsc{1-in-3-SAT}}
\newcommand{\NAESAT}{\textsc{NAE-SAT}}
\newcommand{\OneNAE}{(\One,\NAESAT)}

\usepackage{amsmath}
\usepackage{amssymb}
\usepackage{amsthm} 
\usepackage{bm}
\usepackage{bbm}
\usepackage{mathtools}
\usepackage{mathrsfs}
\usepackage{csquotes}
\usepackage{todonotes}
\usepackage[title]{appendix}
\usepackage{bbm}
\usepackage{xcolor}
\usepackage{blindtext}
\usepackage{enumerate}
\usepackage{manfnt}
\usepackage{multicol}
\usepackage{microtype}

\MakeRobust{\ref}%

\makeatletter
\newcommand{\labeltext}[2]{%
  \@bsphack
  \csname phantomsection\endcsname %
  \def\@currentlabel{#1}{\label{#2}}%
  \@esphack
}
\makeatother

\usepackage[noframe]{showframe}
\usepackage{framed}
 {\endMakeFramed}
\definecolor{shadecolor}{gray}{0.88}

\usepackage{caption}
\usepackage{subcaption}
\usepackage[export]{adjustbox}
\usepackage{hyperref} 
\hypersetup{colorlinks=true,citecolor=blue,linkcolor=blue,urlcolor=blue}

\newcommand{\Test}[2]{
\def\temp{#2}\ifx\temp\empty
  \operatorname{Test}_{#1}
\else
  \operatorname{Test}_{#1}^{#2}
\fi
}

\renewcommand{\vec}[1]{\mathbf{#1}}
\newcommand{\bc}{\vec{c}}

\newcommand{\bp}{\vec{p}}
\newcommand{\bq}{\vec{q}}
\newcommand{\br}{\vec{r}}

\newcommand{\by}{\vec{y}}

\newcommand{\bz}{\vec{z}}

\DeclareMathOperator{\Tr}{Tr}

\DeclareMathOperator{\argmax}{argmax}

\DeclareMathOperator{\supp}{supp}

\newcommand{\bone}{\mathbf{1}}  
\newcommand{\bzero}{\mathbf{0}} 
\newcommand\Frob[2]{{\ensuremath\langle #1,#2\rangle}_{\operatorname{F}}}

\theoremstyle{plain}
\theoremstyle{definition}

\newtheorem*{defn*}{Definition}

\begin{document}

\title{New Algorithms and Hardness Results for\\ Robust Satisfiability of (Promise) CSPs\thanks{An extended abstract of this work appeared in the Proceedings of SODA 2026. Research supported in part by a Simons Investigator award, NSF CCF-2211972, NSF DMS-2503280, and UKRI EP/X024431/1.}}

\author{
Joshua Brakensiek\\
UC Berkeley
\and
Lorenzo Ciardo\\
TU Graz
\and
Venkatesan Guruswami\\
Simons Institute and UC Berkeley
\and
Aaron Potechin\\
University of Chicago
\and
Stanislav \v{Z}ivn\'y\\
University of Oxford
}

\date{\today}

\maketitle

\begin{abstract}
In this paper, we continue the study of robust satisfiability of promise CSPs (PCSPs), initiated in (Brakensiek, Guruswami, Sandeep, STOC 2023 / Discrete Analysis 2025), and obtain the following results:
 \begin{itemize}
\item For the PCSP $\textsc{1-in-3-SAT}$ vs ${\textsc{NAE-SAT}}$ with negations, we prove that it is hard, under the Unique Games conjecture (UGC), to satisfy $1-\Omega(1/\log (1/\epsilon))$ constraints in a $(1-\epsilon)$-satisfiable instance. This shows that the exponential loss incurred by the BGS algorithm for the case of Alternating-Threshold polymorphisms is necessary, in contrast to the polynomial loss achievable for Majority polymorphisms.
\item For any Boolean PCSP that admits Majority polymorphisms,  we give an algorithm satisfying $1-O(\sqrt{\epsilon})$ fraction of the weaker constraints when promised the existence of an assignment satisfying $1-\epsilon$ fraction of the stronger constraints. This significantly generalizes the Charikar--Makarychev--Makarychev algorithm for 2-SAT, and matches the optimal trade-off possible under the UGC. The algorithm also extends, with the loss of an extra $\log (1/\epsilon)$ factor, to PCSPs on larger domains with a certain structural condition, which is implied by, e.g., a family of Plurality polymorphisms. 
\item We prove that assuming the UGC, robust satisfiability is preserved under the addition of equality constraints. As a consequence, we can extend the rich algebraic techniques for decision/search PCSPs to robust PCSPs. The methods involve the development of a correlated and robust version of the general SDP rounding algorithm for CSPs due to (Brown-Cohen, Raghavendra, ICALP 2016), which might be of independent interest.
\end{itemize}

\end{abstract}

\newpage

\setcounter{tocdepth}{2}
\tableofcontents

\newpage

\section{Introduction}

The CSP dichotomy theorem has precisely identified which problems in the rich class of constraint satisfaction problems (CSPs) are polynomial-time solvable, with the rest being NP-complete~\cite{Zhuk20:jacm,Bulatov17:focs}. 
Strikingly, this landmark result shows that simple gadget reductions from 3-SAT are the only obstructions to the existence of an efficient algorithm for a CSP, and conversely the existence of a single ``non-trivial polymorphism'' suffices for a polynomial time satisfiability algorithm. Informally, a polymorphism is an operator that combines multiple satisfying assignments to the predicates defining the CSP into another satisfying assignment.

On the algorithmic side, in essence there are only two broad approaches: local consistency (also captured by a few levels of the Sherali--Adams hierarchy of linear programs)~\cite{BartoKozik09} and generalizations of Gaussian elimination~\cite{Bulatov06:sicomp,Idziak10:siam}. However, the overall algorithm in the CSP dichotomy theorem is highly non-trivial due to the intricate ways in which these two basic algorithmic paradigms might 
have to be combined to solve an arbitrary tractable CSP.

The picture is considerably simpler and clearer %
(but is still non-trivial) when focusing on polynomial time \emph{robust satisfiability algorithms} for tractable CSPs, a concept first considered in the beautiful work of Zwick~\cite{Zwick98:stoc}.
In addition to finding perfectly satisfying assignments when they exist, such algorithms are also robust in the sense that, when given as input almost-satisfiable instances (namely those that admit an assignment failing to satisfy only an $\epsilon$ fraction of the constraints), they find an assignment satisfying all but $g(\epsilon)$ fraction of the constraints, for some loss $g(\epsilon)$ that vanishes as $\epsilon \to 0$.

For ease of terminology, let us refer to CSPs that admit such efficient robust satisfiability algorithms as \emph{robust CSPs}. We further call them $(\eps, g(\eps))$-robust CSPs to indicate the loss incurred by the robust algorithm. Every robust CSP is also tractable---this follows from two works of Barto and Kozik~\cite{BartoKozik09,BK16:sicomp}.
However, the converse is not true and there are tractable CSPs that lack robust satisfiability algorithms.
The quintessential such CSPs are defined by linear relations over an Abelian group. Satisfiability of such CSPs can be efficiently ascertained via Gaussian elimination, but by the celebrated inapproximability results of H\aa stad~\cite{Hastad01}, they are not robust.\footnote{In fact, even for almost-satisfiable instances, it is NP-hard to beat the approximation ratio achieved by the trivial algorithm that simply outputs a random assignment.} On the other hand, tractable CSPs that are solved by local consistency (called bounded-width CSPs in the literature) are in fact robust. For Boolean CSPs, this was effectively implicit in Zwick's original work that pioneered robust satisfiability~\cite{Zwick98:stoc}. For CSPs over any fixed finite domains, this was shown by Barto and Kozik~\cite{BK16:sicomp}. An earlier breakthrough result of Barto and Kozik~\cite{BartoKozik09} had shown that any CSP that cannot express linear equations (in a certain formal sense) is solved by local consistency algorithms. Thus, we have the pleasing picture that CSPs solved by local consistency---one of the two basic algorithmic strategies---are precisely those that are robust (this statement was explicitly conjectured in~\cite{Guruswami12:toc}). 
Further, for all such CSPs, there is a robust satisfiability algorithm  based on semidefinite programming. We therefore have a single unified approach for robust satisfiability compared to the highly complex situation for exact satisfiability of CSPs. In a way, the robust CSP dichotomy offers a crisper and more comprehensible complexity criterion.

Given this backdrop concerning robust CSPs, and motivated by the quest for understanding robustness of SDP-based algorithms more broadly, Brakensiek, Guruswami, and Sandeep~\cite{BGS23:stoc} initiated the study of robust algorithms for \emph{promise CSPs}, which we introduce next.
Promise CSPs (also PCSPs for short) are a vast generalization of CSPs that have received significant attention in recent years~\cite{GS20:icalp,BrandtsWZ21,Barto21:stacs,BrakensiekG21,BartoK22,bz22:ic,NZ22,CZ23:soda-minions,ciardo2023approximate,KOWZ23:sicomp,bgs23:theoretics,CZ23:sicomp,CZ23:stoc,NZ24:talg,fnotw24:stacs,CZ24:stoc,HastadMNZ24,BanakhK24,AvvakumovFOT025,LZ25,NVWZ25,CKKNZ25,Mottet25}.
A promise CSP is defined by a fixed collection of relation
pairs $(P_i,Q_i)$ over some domain pair $(D,E)$, with  $P_i \subseteq Q_i$.\footnote{More generally, there must be a map $h: D\to E$ that is a homomorphism from each $P_i$ to $Q_i$.} 
 Given a CSP instance based
on the relations, the goal is to find an assignment satisfying the (weak)
constraints given by the $Q_i$ relations if promised that an assignment satisfying the
(strong) constraints given by the $P_i$ relations exists (but is not known). A
classic example of a promise CSP is the
\textit{approximate graph coloring} problem~\cite{GJ76}: given a $k$-colorable graph, find an
$\ell$-coloring of it, where $3\leq k\leq\ell$. In our terminology, this is
just the PCSP with a single pair of relations $(P,Q)$, where $P$ is the
disequality relation on a $k$-element set, and $Q$ is the disequality relation
on an $\ell$-element set.

Another, more recent, example of a PCSP is the $(2+\eps)$-SAT problem, introduced and
studied by Austrin, Guruswami, and H{\aa}stad~\cite{AGH17:sicomp} (who also coined the expression promise CSP). They actually studied a more general of a problem: given an instance of
$k$-SAT with the promise that an assignment exists that satisfies at least $g$
literals in each clause, where $1\leq g\leq k$, find a standard satisfying
assignment (satisfying at least $1$ literal in each clause). In this case, the
$P$ relations encode Boolean clause assignments with Hamming weight at least $g$,
where the $Q$ relations encode Boolean clause assignments with Hamming weight
at least $1$. Another example is
the $\One$ vs $\NAESAT$ problem, identified in the influential paper of
Brakensiek and Guruswami~\cite{BG21:sicomp} that initiated a systematic study of
Boolean PCSPs. Here one is given a satisfiable instance of $\One$ and the
goal is to find an assignment that satisfies $1$ or $2$ variables per clause.
Astoundingly, this problem is solvable in polynomial time (via an algorithm not previously considered in the context of CSPs)~\cite{BG21:sicomp}. Further, such an algorithm cannot be obtained via a reduction to (finite domain) CSPs~\cite{BBKO21:jacm}!

The study of PCSPs calls for significant new algorithmic and hardness techniques. Studying such techniques in the broader context of PCSPs has also led to new results for (standard, non-promise) CSPs, e.g., a single algorithm blending together linear programming with linear Diophantine equations~\cite{BGWZ20:sicomp} that solves all tractable Boolean CSPs.
Despite a lot of attention and recent progress on PCSPs, the complexity landscape is vast and mostly not understood. In fact, the complexity of Boolean PCSPs is itself a major challenge, in contrast to the CSP world where Schaefer proved a dichotomy for Boolean CSPs already in the 1970s~\cite{Schaefer78:stoc}. Following a classification of Boolean symmetric  PCSPs allowing negations from~\cite{BG21:sicomp}, Ficak, Kozik, Ol\v{s}\'ak, and Stankiewicz obtained a classification of Boolean PCSPs with symmetric relations~\cite{Ficak19:icalp}. Moreover,  Brakensiek, Guruswami, and Sandeep obtained a (conditional) classification of monotone Boolean PCSPs~\cite{bgs23:theoretics}.\footnote{The classification assumes the Rich 2-to-1 conjecture of Braverman, Khot, and Minzer~\cite{Braverman21:itcs}.} 

Returning to robust satisfiability,
given the crisp picture of robust CSPs---namely, either the natural SDP gives an efficient robust algorithm or none exists---the study of robust PCSPs is a natural goal, as proposed by 
  Brakensiek, Guruswami, and Sandeep (BGS)~\cite{BGS23:stoc}.  A robust satisfiability algorithm for a PCSP defined by relation pairs $(P_i,Q_i)$ means the following: given an instance such that $(1-\eps)$ fraction of the constraints are promised to be satisfiable according to the stronger relations $P_i$, there is an algorithm to weakly satisfy (according to the relations $Q_i$) $(1-g(\eps))$ fraction of the constraints, where $g(\eps) \to 0$ as $\eps \to 0$. As with CSPs, in this case we say that the PCSP is robust or $(\eps,g(\eps))$-robust.

BGS focused on Boolean PCSPs where the known satisfiability algorithms\footnote{Recall that robust PCSPs must first of all be tractable.} can be attributed to the existence of polymorphisms. A polymorphism for a relation pair $(P,Q)$ is a homomorphism from a (categorical) power $P^m$ of $P$ to $Q$, with $\Pol(P,Q)$ representing the set of all polymorphisms. See Section~\ref{subsec:polymorphism-defn} for a formal definition. More precisely, BGS focused on three families of polymorphisms: %
Majority (MAJ), Alternating Threshold (AT), and Parity~\cite{BG21:sicomp}. For any odd $L \in \N$, we let $\MAJ_L, \AT_L, \PAR_L : \{-1,1\}^L \to \{-1,1\}$ be defined as
\begin{align*}
\MAJ_L(x_1, \hdots, x_L) &:= \one\left[\sum_{i=1}^L x_i \ge 0\right],\\
\AT_L(x_1, \hdots, x_L) &:= \one\left[\sum_{i=1}^L (-1)^{i-1}x_i \ge 0\right],\\
\PAR_L(x_1, \hdots, x_L) &:= \one\left[\sum_{i=1}^L x_i \equiv L\!\! \mod 4\right].
\end{align*}

We let $\MAJ := \{\MAJ_L : L \in \N\text{ odd}\}$, $\AT := \{\AT_L : L \in \N\text{ odd}\}$, $\PAR := \{\PAR_L : L \in \N\text{ odd}\}$ be the respective sets of polymorphisms. 
In two of these cases, Majority and AT, BGS showed that the associated PCSPs are robust via an algorithm based on semidefinite programming~\cite{BGS23:stoc}. They also showed that a partial converse holds: for PCSPs defined by a single pair $(P,Q)$ of symmetric Boolean relations (plus allowing negations), if the relation pair $(P,Q)$ lacks some odd-arity MAJ and some odd-arity AT as a polymorphism, then the PCSP is not robust (assuming the Unique Games conjecture).

The quantitative aspects of the robust satisfaction algorithms in \cite{BGS23:stoc} for the two cases, Majority and AT, however, diverged significantly. For Majority, the BGS algorithm, which is really the same as the Charikar--Makarychev--Makarychev algorithm for 2-SAT but analyzed in greater generality assuming only a Majority polymorphism, guaranteed that at most %
$\widetilde{O}(\epsilon^{1/3})$ fraction of the constraints are violated. This is weaker than the $(\eps,O(\sqrt{\epsilon}))$ robustness  guarantee for 2-SAT~\cite{CMM09:talg}---which is tight~\cite{KKMO07,MOO10} under the Unique Games conjecture~\cite{Khot02,KKMO07,MOO10}.\footnote{Historically, showing evidence for the near-optimality of the Goemans--Williamson robust algorithm for Max-Cut, which also achieves $O(\sqrt{\epsilon})$ loss, was the original motivation for the formulation of the Unique Games conjecture in \cite{Khot02}.} A natural question then is whether the BGS loss guarantee for Majority polymorphisms can be improved to $O(\sqrt{\eps})$, which would then give the right polymorphic generalization of the CMM robust algorithm for 2-SAT. 

The situation for robust algorithms for Boolean PCSPs with AT polymorphisms is worse, as the BGS algorithm only showed $(\eps,O(\log\log(1/\eps)/\log (1/\eps)))$-robustness. A natural question then is whether this exponential loss is necessary, or whether one can achieve polynomial loss also for the AT case similar to the Majority case. We address and resolve both of these questions in this work.

\subsection{Our Results}\label{subsec:results}

Our contributions in this paper fall into three parts. %
First, we show that the robust algorithm for Alternating Threshold due to Brakensiek--Guruswami--Sandeep~\cite{BGS23:stoc} has a near-matching hardness result under the UGC. This is based on a novel integrality gap for $\OneNAE$. Second, we show that any promise template with the Majority polymorphism has a robust algorithm with loss $g(\eps) = O(\sqrt{\eps})$, which is asymptotically tight, and improves over BGS's analysis of $O(\eps^{1/3})$. We further extend this analysis to show that similar algorithms achieve a robustness of $O(\sqrt{\eps}\log(1/\eps))$ for Plurality and related polymorphisms. Finally, we show that the robustness of PCSPs is (approximately) preserved under a large family of gadget reductions under the UGC. We do this by solving %
a seemingly elementary but technically complex problem: given a promise template with a robust algorithm, show that adding the equality relation to the template (approximately) preserves the robustness of the problem.

\medskip \noindent \textbf{Hardness for Alternating Threshold.} A key result of Brakensiek--Guruswami--Sandeep~\cite{BGS23:stoc} is that for any promise template $(P,Q)$ with $\AT \subseteq \Pol(P,Q)$,
$\PCSP(P,Q)$ is robust with $g(\eps) = O(\frac{\log \log (1/\eps)}{\log (1/\eps)})$. Interestingly, for CSPs, similar asymptotics appear with the $\OR$ and $\AND$ families\footnote{For any $L \in \N$, $\OR_L(x_1, \hdots, x_L) = 1$ if $x_i = 1$ for some $i \in [L]$ and $\AND_L(x_1, \hdots, x_L) = 1$ if $x_i = 1$ for all $i \in [L]$.}
of polymorphisms, and these are known to be tight \cite{Zwick98:stoc,Guruswami12:toc}. We show that $\AT$ exhibits a similar behavior by proving UGC hardness.

\begin{theorem}[AT hardness, informal]\label{thm:AT-informal}
Assuming UGC, $\fiPCSP\OneNAE$ is not $(\eps, \Omega(1/\log(1/\eps)))$-robust. %
\end{theorem}

Here $\fiPCSP$ (folded, idempotent PCSP) refers to PCSPs that allow for variables to be negated and set as constants (see \cite{BGS23:stoc}). The use of negations is necessary, as Brakensiek--Guruswami--Sandeep~\cite{BGS23:stoc} observed that $\OneNAE$ without negations is robust with polynomial loss.  \Cref{thm:AT-informal} is proved in \Cref{sec:AT}. Using Raghavendra's theorem~\cite{Raghavendra08:everycsp}, we prove \Cref{thm:AT-informal} by constructing an explicit integrality gap---the same high-level strategy as used by Guruswami--Zhou~\cite{Guruswami12:toc} for HORN-3-SAT, although the execution and analysis in our setting are significantly more complex. 
We explain further details in \Cref{subsec:tech}.

\medskip \noindent \textbf{Improved Analysis for Majority and Beyond.}
Our next result is a robust algorithm for PCSPs with Majority polymorphisms with an improved loss function.
\begin{theorem}[MAJ robustness, informal]\label{thm:MAJ-informal}
For any promise template $(P,Q)$ with $\MAJ \subseteq \Pol(P, Q)$, $\PCSP(P, Q)$ is $(\eps, O(\eps^{1/2}))$-robust.
\end{theorem}

This result is proved in \Cref{sec_majority}. We note that the algorithm used in \Cref{thm:MAJ-informal} is identical to the one used in \cite{BGS23:stoc}. The main improvement in the analysis comes from a more refined analysis of multivariate normal distributions. See \Cref{subsec:tech} for further details.

Note that \Cref{thm:MAJ-informal} only applies to Boolean PCSPs. A commonly studied non-Boolean variant of the Majority polymorphisms is the \emph{Plurality} polymorphisms. %
Typical examples of (P)CSPs admitting such polymorphisms are \textit{unique games}~\cite{Khot02} as well as the so-called SetSAT problem---a non-Boolean generalization of $(2+\eps)$-SAT~\cite{AGH17:sicomp} introduced in~\cite{BrandtsWZ21}.

\begin{theorem}[PLUR robustness, informal]\label{thm:PLUR-informal}
For any promise template $(P,Q)$ with $\PLUR \subseteq \Pol(P,Q)$, $\PCSP(P,Q)$ is $(\eps, O(\eps^{1/2}\log(1/\eps)))$-robust.
\end{theorem}

This result is proved in \Cref{sec_separablesPCSPs}. We note that \Cref{thm:PLUR-informal} is merely a special case of our main result in \Cref{sec_separablesPCSPs}, which applies to any \emph{separable} PCSP, see \Cref{thm_robust_SDP_separable_PCSPs}. We describe this broader family more precisely in \Cref{subsec:tech}.

\paragraph{Robust Gadget Reductions (adding EQUALITY).} In virtually all classifications of (variants of) CSPs, an essential tool is \emph{gadget reductions} between CSP templates. For example, in the CSP dichotomy, the hardness side of the CSP is done using gadget reductions from 3-SAT \cite{BulatovJK05}. Robust (P)CSPs are no exception, and gadget reductions are frequently used to study the relationship between templates \cite{DalmauK13,BGS23:stoc}. However, there is a significant distinction between the ordinary CSP dichotomy~\cite{Bulatov17:focs,Zhuk20:jacm} and the one for robust (P)CSPs~\cite{BK16:sicomp}: the allowance of \emph{equality constraints}, which we denote by $\EQ$. For exact satisfiability of (P)CSPs, if we specify that some variables are to be equal, we can efficiently compute the connected components of the equality relation and distill the problem down to a smaller number of variables (and without any equality constraints). However, for robust (P)CSPs, equality is a rather subtle concept. If only $1-\eps$ constraints are satisfied in the optimal assignment, we do not know whether each equality constraint should be trusted or ignored. That said, it still seems quite reasonable to assume that adding equality constraints should only mildly change the robustness of the resulting PCSP.

\begin{question}[Barto--Kozik~\cite{BK16:sicomp}]\label{ques:eq}
Let $\Gamma$ be a CSP template, i.e., a set of relations. Assume that $\CSP(\Gamma)$ is $(\eps, f(\eps))$-robust. Is $\CSP(\Gamma \cup \{\EQ\})$ then $(\eps, O(f(\eps)))$-robust?
\end{question}

Despite Barto--Kozik~\cite{BK16:sicomp} giving a complete classification of all robust CSPs,\footnote{We observe that Barto--Kozik were able to establish this classification by using gadget reductions that do \emph{not} allow equality. Using \Cref{thm:gadget-informal}, one could (in theory) simplify parts of Barto--Kozik's proof, although the proof of \Cref{thm:gadget-informal} is much more complicated than the workarounds needed by Barto--Kozik. %
} they do not answer \Cref{ques:eq} except in the very weak sense that $\CSP(\Gamma \cup \{\EQ\})$ is $(\eps, O(\log \log(1/\eps)$ $/ \log(1/\eps)))$-robust, independent of $f$. In this paper, we establish \Cref{ques:eq} is \emph{nearly} true for both CSPs and PCSPs:

\begin{theorem}[Robustness of Equality, informal]\label{thm:EQ-informal}
Assume UGC. For any promise template $\Gamma$, if $\PCSP(\Gamma)$ is $(\eps, f(\eps))$-robust, then $\PCSP(\Gamma \cup \{\EQ\})$ is $(\eps, O(f(\eps^{1/6})))$-robust.
\end{theorem}

As an immediate corollary, we can now use the most general gadget reductions available for studying (P)CSPs to study robust (P)CSPs, modulo a polynomial loss in robustness.

\begin{theorem}[Gadget Reductions, informal]\label{thm:gadget-informal}
Assume UGC. Let $\Gamma$ and $\Gamma'$ be promise templates such that there is a gadget reduction\footnote{More precisely, there is a \emph{minion homomorphism} from $\Pol(\Gamma)$ to $\Pol(\Gamma')$.} from $\PCSP(\Gamma')$ to $\PCSP(\Gamma)$.
If $\PCSP(\Gamma)$ is $(\eps, f(\eps))$-robust, then $\PCSP(\Gamma')$ is $(\eps, O(f(\eps^{1/6})))$-robust.
\end{theorem}

These results are proved in \Cref{sec:eq-robust}. In the technical overview (\Cref{subsec:tech}), we describe %
the techniques we use %
to establish our results, including an adaptation of the algorithm of Brown-Cohen and Raghavendra~\cite{BR16:correlation} for approximate (P)CSPs.

\subsection{Technical Overview}\label{subsec:tech}

\paragraph{Hardness for Alternating Threshold.}

Our integrality gap instance for $\One$ vs $\NAESAT$ has a similar high-level idea to the following LP integrality gap instance for Horn-SAT~\cite{Guruswami12:toc}:
\begin{enumerate}
\item We have variables $\{x_{j1},x_{j2}: j \in [k]\}$ where $k = \lceil{1/\log_2(\epsilon)}\rceil + 1$.
\item We have the unary constraints $x_{11}$, $x_{12}$, and $\neg{x}_{k1}$ where $1$ is True and $-1$ is False.
\item For all $j \in [k-1]$, we have the constraints $\neg{x}_{j1} \vee \neg{x}_{j2} \vee x_{(j+1)1}$ and $\neg{x}_{j1} \vee \neg{x}_{j2} \vee x_{(j+1)2}$.
\end{enumerate}
Clearly this instance is unsatisfiable, as we insist that $x_{11},x_{12}$ are True, and then $x_{k1},x_{k2}$ should also be True due to the chain of implication constraints, but we insist that $x_{k1}$ is False. As there are only $O(k) = O(\log (1/\eps))$ constraints, the integral value of this instance is at most $1 - \Omega(1/\log(1/\eps))$.

On the other hand, the following LP solution gives value at least $1-\epsilon$ to all of the constraints:
\begin{enumerate}
\item For each $j \in [k]$, we give $x_{j1}$ and $x_{j2}$ bias $1-2^{j+1-k}$.
\item For the constraint $x_{11}$, we can set $x_{11} = 1$ with probability $1-2^{1-k}$ and $-1$ with probability $2^{1-k}$ and we will have that $\E[x_{11}] = 1-2^{1-k} - 2^{1-k} = 1 - 2^{2-k}$. Thus, the LP gives a value of $1-2^{1-k} \geq 1 - \epsilon$ for this constraint. By symmetry, the LP also gives a value of $1-2^{1-k} \geq 1 - \epsilon$ for the constraint $x_{12}$.
\item For each $j \in [k-1]$, for the constraint $\neg{x}_{j1} \vee \neg{x}_{j2} \vee x_{(j+1)1}$, we can take the distribution where 
\begin{enumerate}[a.]
\itemsep=0ex
\item[a.] With probability $1 - 2^{j+1-k}$, we set $x_{j1} = x_{j2} = x_{(j+1)1} = 1$.
\item[b.] With probability $2^{j-k}$, we set $x_{j1} = 1$, $x_{j2} = -1$, and $x_{(j+1)1} = 1$.
\item[c.] With probability $2^{j-k}$, we set $x_{j1} = -1$, $x_{j2} = 1$, and $x_{(j+1)1} = 1$.
\end{enumerate}
With this distribution, $\E[x_{j1}] = \E[x_{j2}] = (1 - 2^{j-k}) - 2^{j-k} = 1 - 2^{j+1-k}$ and $\E[x_{(j+1)1}] = (1 - 2^{j+1-k}) - 2^{j+1-k} = 1 - 2^{j+2-k}$. Thus, the LP gives this constraint a value of $1$. By symmetry, the LP also gives a value of $1$ to the constraint $\neg{x}_{j1} \vee \neg{x}_{j2} \vee x_{(j+1)2}$.
\item The bias for $x_{k1}$ is $1 - 2 = -1$ so $x_{k1}$ is always set to $-1$ which satisfies the constraint $\neg{x}_{k1}$.
\end{enumerate}
One way to think about this integrality gap instance is as follows. In order to avoid violating a constraint, the following must hold.
\begin{enumerate}
\itemsep=0ex
\item Variables with bias $1-2^{j+1-k}$ must be rounded to $1$.
\item For all $j \in [k-1]$, if the variables with bias $1-2^{j+1-k}$ are rounded to $1$ then the variables with bias $1-2^{j+2-k}$ are rounded to $1$.
\item Variables with bias $-1$ must be rounded to $-1$.
\end{enumerate}
Since we can obtain a contradiction in $O(\log(1/\eps))$ steps, at least $\Omega(1/\log(1/\eps))$ of the constraints must be violated.

We will use a similar idea for our integrality gap instance. We will construct our instance so that while the SDP value is at least $1 - \epsilon$, if we want to avoid violating a significant number of constraints,
\begin{enumerate}
\itemsep=0ex
\item Almost all of the variables with bias $1-2^{-k}$ must be set to $1$ and almost all of the variables with bias $2^{-k}-1$ must be set to $-1$.
\item For all $j \in [k]$, if almost all of the variables with bias $1 - 2^{-j}$ are set to $1$ then almost all of the variables with bias $2^{1-j} - 1$ are set to $-1$. Similarly, for all $j \in [k]$, if almost all of the variables with bias $2^{-j} - 1$ are set to $-1$ then almost all of the variables with bias $1 - 2^{1-j}$ are set to $1$.
\end{enumerate}
We then observe that these conditions imply that almost all of the variables with bias $0$ are set to $1$ and almost all of the variables with bias $0$ are set to $-1$, which is impossible. Since we can obtain a contradiction in $O(\log(1/\eps))$ steps, at least $\Omega(1/\log(1/\epsilon))$ of the constraints must be violated.

In order to have $\One$ constraints, it turns out that we need general vectors of the form $\{x\vec{v}_{0} + \sqrt{1-x^2}\vec{w}\}$ where $\vec{w}$ is orthogonal to $\vec{v}_0$. A natural choice for this is to use all $\vec{w} \in S^{d-1}$ for some large $d$ which depends on $\epsilon$. While this gives an integrality gap instance, it has infinite size and 
is not that easy to analyze since %
it involves functions on the sphere rather than functions with multivariate Gaussian inputs. Thus, we modify this integrality gap instance as follows:
\begin{enumerate}
\item Instead of using vectors $\vec{w} \in S^{d-1}$, we will use vectors $\vec{w} \sim \cN(0,1/d)^d$.
\item For all but a negligible portion of the constraints, our vectors $\vec{w} \sim \cN(0,1/d)^d$ are very close to unit vectors so we can discard the negligible number of constraints where the vectors are badly behaved.
\item We discretize our instance by splitting our space into regions and mapping all vectors in each region to a representative vector in that region.
\end{enumerate}
Through a careful analysis, we show that for this modified instance, $\Omega(1/\log(1/\epsilon))$ fraction of the constraints must be violated and even after these modifications, the SDP value for our instance is at least $1 - \epsilon$.

\paragraph{Improved Analysis for Majority.} We now switch to %
designing robust algorithms for families of PCSPs. To begin, we discuss the algorithm used by Brakensiek, Guruswami, Sandeep~\cite{BGS23:stoc} for promise templates with Majority polymorphisms. This algorithm was inspired by the algorithm used by Charikar, Makarychev, Makarychev~\cite{CMM09:talg} for robust MAX 2-SAT.

For convenience, we relabel the Boolean domain as $\{-1,+1\}$. Fix a template $(P,Q)$ with $\MAJ \subseteq \Pol(P,Q)$. Consider an instance of $\PCSP(P,Q)$ on variables $x_1, \hdots, x_n$ and clauses $C_1, \hdots, C_m$. The algorithm begins by solving the Basic SDP for this instance by finding unit vectors $\bv_0, \bv_1, \hdots, \bv_n \in \R^{n+1}$ (with $\bv_0$ representing the ``truth'' vector) such that the average value of the vector assignment to the clauses is $1-\eps$, where $\eps > 0$ is the specified robustness parameter. Next, we sample a random multivariate normal vector $\br \in \cN(0^{n+1}, I_{n+1})$. Then, for all $i \in [n]$, we round $x_i$ to $+1$ if $\langle \bv_i, \bv_0 + \br \cdot \eps^{2/3}\rangle \ge 0$ and $-1$ otherwise. Here, we improve $\eps^{2/3}$ to $\sqrt{\eps}$ via a new analysis.

We briefly explain the key ideas in Brakensiek, Guruswami, Sandeep~\cite{BGS23:stoc} in the analysis of this algorithm. Using a reduction in their paper, we may also assume without loss of generality that $Q = \{-1,+1\}^k \setminus \{(-1)^k\}$. In other words, assume that the majority vote of any list of assignments to $P$ is never all $-1$'s. For simplicity, fix a clause $C_i$ on variables $x_1, \hdots, x_k$ such that the SDP vectors $\bv_0, \bv_1, \hdots, \bv_k$ give a value of $1-\eps$. It may be the case that for all $i \in [n]$, $\langle \bv_i, \bv_0 \rangle \approx -\Theta(\eps)$, so purely rounding $\langle \bv_i, \bv_0 \rangle$ will fail to satisfy any of the clauses. A key observation by BGS is that as long as the vectors have completeness $1-\eps$, there exists a probability distribution $(w_1, \hdots, w_k)$ such that $\sum_{i=1}^k w_i\langle \bv_i, \bv_0 \rangle \geq -\eps$ (see \Cref{lem:robust-MAJ-pol}). %
Let $\bu := \sum_{i=1}^k w_i\bv_i$. By concentration, we can assume with probability $1- \eps^{O(1)}$ that $|\langle \bu, \br\rangle| = O(\log(1/\eps))$. As a key observation, note that since $\langle \bu, \bv_0 + \br \cdot \eps^{2/3}\rangle = \sum_{i=1}^k w_i\langle \bv_i, \bv_0 + \br \cdot \eps^{2/3}\rangle$, if all $k$ variables round to $-1$, then $\langle \bu, \bv_0 + \br \cdot \eps^{2/3}\rangle$ is negative. However, $\langle \bu, \bv_0\rangle  \ge -\eps$ and the standard deviation of $\langle \bu, \br \cdot \eps^{2/3}\rangle$  is at most $\eps^{1/3}$. %
Thus, if $\langle \bu, \bv_0 + \br \cdot \eps^{2/3}\rangle$ is negative, it is barely negative. Hence, this tightly constrains the value of each $\langle \bv_i, \bv_0 + \br \cdot \eps^{2/3}\rangle$, which is unlikely due to anti-concentration of the normal distribution. %

With a more careful analysis of these rounding probabilities, we can change the rounding threshold to $\langle \bv_i, \bv_0 + \br \cdot \sqrt{\eps}\rangle$ and get a $1 - O(\sqrt{\eps})$ success probability (soundness). The key idea is, instead of directly comparing each $\langle \bu, \br\rangle$ to the individual distributions $\langle \bv_i, \br\rangle$, we use a more careful decomposition of the vectors. %
In particular, define $\bv_i^{=}$ to be the component of $\bv_i$ parallel to $\bu$ and let $\bv_i^{\perp}$ be the component of $\bv_i$ perpendicular to $\bu$. Since $\bv_i^{=}$ and $\bu$ are related by a scalar, $\langle \bv_i^{=}, \br\rangle$ and $\langle \bu, \br\rangle$ are also related by a scalar. However, $\langle \bv_i^{\perp}, \br\rangle$ is independent of $\langle \bu, \br\rangle$. Using this observation, we can split our argument into three high-level cases.

First, if $\|\bu\|_2^2 = \Omega(\eps \log(1/\eps))$, %
then $\langle \bu, \bv_0\rangle$ will dominate $\langle \bu, \br \sqrt{\eps}\rangle$, so the chances that $\langle \bu, \bv_0 + \br \sqrt{\eps}\rangle \le 0$ are quite small.

Second, if a perpendicular component is large, that is $\|w_i \bv_i^{\perp}\|_2 = \Omega(1)$ for some $i \in [k]$, then even if we condition on $\langle \bu, \br\rangle$, the value of $\langle w_i\bv_i, \bv_0 + \sqrt{\eps} \br\rangle $ still has considerable variance. In particular, most likely $\langle w_i\bv_i, \bv_0 + \sqrt{\eps} \br\rangle $ will either be (1) too positive, in which case $i$ is rounded correctly, or (2) too negative, in which case the average of $\sum_{j \neq i} w_j\langle \bv_j, \bv_0 + \sqrt{\eps} \br\rangle$ is positive, so some other $j$ is rounded correctly.

These two cases themselves are enough to get a $\sqrt{\eps}\log(1/\eps)$ loss. To shave the $\log$, in the third and final case, we finely partition the space of potential ``bad'' outcomes and show that these in total contribute at most $O(\sqrt{\eps})$ loss to the rounding. This is %
the most technical part of the argument.

\paragraph{New algorithms for Plurality and Separable Families.} %
We extend these rounding techniques for Majority to non-Boolean domains and more general rounding functions. To do this, we abstract out the essential feature of the analysis of Majority: the existence of a hyperplane separation between the strong form of the constraint $P$ and SDP-configurations whose rounding lies outside of the weak form $Q$.    
We define templates $(P,Q)$ with a generalization of this property that we call \emph{separable} families. The precise description is given in~\Cref{defn_separable_rho} but, at a high level, here is the idea.
In the Boolean Majority case, our analysis relies upon a linear function separating $\langle \bv_i, \bv_0\rangle$ from an absent tuple of $Q$, which in turn can be expressed via the inner product with a weight vector $w$. %
In non-Boolean domains $D$, we have a separate vector $\bv_{i,d}$ for each variable $x_i$ and each domain element $d \in D$. Hence, we encode both $P$ and the SDP-configurations whose rounding lies outside of $Q$ as two convex bodies living in the matrix space $\R^{k \times D}$, where $k$ is the arity of the constraint. The first is the convex hull of the \emph{one-hot} encodings $\Pi_{\bp}$ of tuples $\bp\in P$ (where the $(i,d)$-th entry equals $1$ if $p_i = d$ and $0$ otherwise). The second is the preimage under the given rounding function $\rho$ of ``bad'' tuples---those lying outside of $Q$. If the polymorphisms of $(P,Q)$ are rich enough, one can show that these two convex sets are disjoint---in which case, they must admit a hyperplane separation. The latter is naturally expressed via a linear functional over $R^{k\times D}$ determined by the Frobenius inner product times a suitable weight matrix $W$---the non-Boolean analogue of the weight vector $w$.
If, for some rounding function $\rho$, the template $(P,Q)$ admits such a separation, we call it \textit{$\rho$-separable}.\footnote{Such predicates have some resemblance to the ``regional polymorphisms'' defined by Brakensiek--Guruswami~\cite{BG19:soda}.} Provided that $\rho$ satisfies an extra \textit{conservativity} condition (roughly speaking, the winning element of a given distribution must have weight bounded away from zero), we are able to use the hyperplane separation property to show that $(P,Q)$ is robustly solved by SDP with loss $O(\sqrt{\eps}\log(1/\eps))$, by generalizing part of the analysis performed in the Boolean Majority case. Note that the loss we achieve in this setting is slightly worse than the one for Majority, by a $\log(1/\eps)$ factor (although, even for Majority, better than the loss achieved in \cite{BGS23:stoc}). This is due to the fact that the trick of splitting the SDP vectors into parallel and orthogonal components does not carry over in the non-Boolean domain. 

A notable example of a separable family is the Unique Games problem. Unique Games has been known to have a robust algorithm for a long time due to the algorithm of Charikar, Makarychev, and Makarychev \cite{CMM06:stoc} (although this algorithm is rather different from the one used by the same authors for 2-SAT). The fact that their 2-SAT algorithm can be extended to Unique Games may be of independent interest.
The underlying polymorphism driving this is \emph{Plurality}, which selects the most commonly occurring element in a list of domain elements, even if its frequency is much less than $1/2$.
More interestingly, our polymorphism-based result captures \textit{all} PCSPs admitting Plurality---in particular,  the family of so-called SetSAT PCSPs identified in~\cite{BrandtsWZ21} as a natural non-Boolean generalization of $(2+\eps)$-SAT~\cite{AGH17:sicomp}. Unlike Unique Games, these problems were previously not known to be robustly solvable with any loss.

\paragraph{Robust Gadget Reductions (Adding Equality).}

We now discuss the proofs of \Cref{thm:EQ-informal} and \Cref{thm:gadget-informal}. Assuming \Cref{thm:EQ-informal}, \Cref{thm:gadget-informal} is straightforward to establish by combining existing gadget reductions for robust (P)CSPs \cite{DalmauK13,BK16:sicomp,BGS23:stoc} with state-of-the-art gadget reductions for (P)CSPs \cite{BBKO21:jacm}. The precise details are worked out in \Cref{cor:gadget-robust}.

As such, we focus on sketching the proof of \Cref{thm:EQ-informal}. We crucially build off the algorithm of Brown-Cohen and Raghavendra~\cite{BR16:correlation} (``BCR algorithm'') for solving approximate MAX (P)CSPs.\footnote{Technically, Raghavendra's theorem and the result of Brown-Cohen--Raghavendra are only stated for CSPs, but as noted by Brakensiek, Guruswami, Sandeep~\cite{BGS23:stoc}, their arguments extend to PCSPs with minimal modification.} For the purposes of this high-level overview, we assume that the our promise template is a single pair of Boolean relations $(P, Q)$ with $P \subseteq Q \subseteq \{-1,1\}^k$. %

We first describe the essential features of the BCR algorithm. For an instance of $\PCSP(P,Q)$ on variable set $x_1, \hdots, x_n$ and clauses $C_1, \hdots, C_m$, we can think of an SDP solution as a collection of unit vectors $\bv_0, \bv_1, \hdots, \bv_n \in \R^{n+1}$. After finding an SDP solution with near-optimal value, the BCR algorithm proceeds by sampling two random objects: (i) a list of random vectors $\br_1, \hdots, \br_D \in \R^{n+1}$ sampled from a multivariate normal distribution, and (ii) a rounding function\footnote{The choice of rounding function is based on the existence of certain \emph{approximate polymorphisms}, see \Cref{subsec:BCR} for a precise definition.} $H : \R^{D} \to [-1, 1]$. For each $i \in [n]$, we then define a fractional assignment $z_i \in [-1,1]$ via
\[
    z_i := H(\bv_0 \cdot \bv_1 + \bv^{\perp}_i \cdot \br_1, \hdots, \bv_i \cdot \bv_0 + \bv^{\perp}_i \cdot \br_D),
\]
where we set $\bv^{\perp}_i = \bv_i - (\bv_0 \cdot \bv_i) \bv_0.$ We then get an integral solution to the PCSP by independently rounding $x_i$ to $+1$ with probability $\frac{1+z_i}{2}$ and $-1$ otherwise. 

Let $\cR$ be the probability distribution over the choices of $(\br_1, \hdots, \br_D)$ and $H$. We can thus think of the BCR rounding scheme as a map $\BCR : \R^{n+1} \times \cR \to [-1, 1]$ for which \emph{global shared randomness} $R \in \cR$ is picked at the start of the algorithm, and then for each $i \in [n]$, we set $z_i := \BCR(\bv_i, R)$. 
The assumed robustness of the algorithm then translates into the following guarantee.

\smallskip\textbf{Key Property.} For \emph{every} SDP solution $\bv_0, \bv_1, \hdots, \bv_n$ with value at least $1 - \eps$, the rounded assignment $\BCR(\bv_1, R),$ $\hdots,$ $\BCR(\bv_n, R)$ will satisfy $1 - f(\eps)$ of the constraints of our instance in expectation over the choice of $R \in \cR$.

As is, the existing scheme may not be robust for $\PCSP(\EQ)$ for the following reason: given two vectors $\bv_i$ and $\bv_j$ that are $\delta$ apart in Euclidean distance, the corresponding $\BCR(\bv_i, R)$ and $\BCR(\bv_j, R)$ might be very different for a typical $R \sim \cR$.\footnote{In general $H$ is very slightly smooth, so one can directly combine the BCR algorithm with a correlated rounding trick to get a robust algorithm for equality. The ``catch'' is that soundness of the robust algorithm depends on the arity of the approximate polymorphisms considered by BCR. However, no effective bound is given by BCR on the size of these approximate polymorphisms, resulting in guarantees much worse than \Cref{thm:EQ-informal}.} 
In order to make the BCR rounding scheme also robust for equality, we exploit the Key Property in the following way. Consider a map $M_{\delta} : \R^{n+1} \to \R^{n+1}$ (not necessarily linear) such that, for every unit vector $\bv$, the distance between $\bv$ and $M_{\delta}(\bv)$ is at most $\delta$. We can then consider the following scheme: $S_{\delta}(\bv, R) := \BCR(M_{\delta}(\bv), R)$. A key observation is that $S_{\delta}$ is still a robust rounding scheme with slightly worse parameters.

To see why, using $S_{\delta}$ to round a solution $\bv_1, \hdots, \bv_n$ is effectively the same as using $\BCR$ to round $M_{\delta}(\bv_1), \hdots, M_{\delta}(\bv_n)$. Since each $\bv_i$ is close in Euclidean distance to $M_{\delta}(\bv_i)$, for $\delta$ sufficiently small, the SDP value of the solution\footnote{It may be the case that $M_{\delta}(\bv_1), \hdots, M_{\delta}(\bv_n)$ is no longer a valid SDP solution due to violating triangle inequalities. Circumventing this issue is highly technical and involves adapting a smoothing trick due to Raghavendra and Steurer~\cite{raghavendra2009round}.  We ignore this important issue for the purposes of this overview.}  $M_{\delta}(\bv_1), \hdots, M_{\delta}(\bv_n)$ is still approximately $1-\eps$, and as such we still satisfy roughly $1 - f(\eps)$ constraints on average.

More generally, $M_{\delta} : \R^{n+1} \to \R^{n+1}$ does not need to be deterministic, rather it can be any randomized map such that the input vector can never be more than $\delta$ far from the output vector. We call such a randomized map a \emph{$\delta$-spread} if for each unit vector $\bv$, if the probability distribution $M_{\delta}(\bv)$ is supported within the ball $B(\bv, \delta)$ and the probability degrades smoothly with distance. %
See \Cref{def:spread} for a precise definition.

For a given $\delta$-spread (a distribution of $M_{\delta}$'s), the corresponding rounding scheme $S_{\delta}(\bv, R) := {\E}_{M_{\delta}}[\BCR(M_{\delta}(\bv), R)]$ is called a \emph{$\delta$-smoothing} of BCR. 
By the aforementioned logic, any $\delta$-smoothing of $\BCR$ is still approximately $(\eps, f(\eps))$-robust. In particular, we now have a large collection of rounding schemes that are all robust for $\PCSP(P, Q)$.

Our next step is to pick one of these $\delta$-smoothings that is also robust for equality. For each unit vector $\bv \in \R^{n+1}$, we look at the following $L^2$ norm:
\[
    \|S_{\delta}(\bv)\|_2 := \sqrt{\E_{R \sim \cR}[S_{\delta}(\bv, R)^2]}.
\]
Recall that the range of $S_{\delta}$ is $[-1,1]$, so $S_{\delta}(\bv, R)^2$ roughly measures the certainty the rounding scheme has for this value. Rather unintuitively, we select the $\delta$-smoothing of $\BCR$ such that $\|S_{\delta}(\bv)\|_2$ is \emph{minimized} for all unit vectors $\bv \in \R^{n+1}$.  We call this scheme $\REQ_{\delta}$, as we shall soon see it is Robust for EQuality. %
Roughly speaking, $\REQ_{\delta}$ is the $\delta$-smoothing of $\BCR$ with maximal entropy. 

The key lemma we seek to show is that there are (small) constants $c_1, c_2 \ge 1$ such that if two vectors $\bv$ and $\bw$ are within distance $\delta^{c_1}$ of each other, then \[\sqrt{\E_{R \sim \cR}[(\REQ_{\delta}(\bv, R) - \REQ_{\delta}(\bw, R)^2]} \le \delta^{c_2}.\] See \Cref{lem:req-eq-robust} for a precise statement. The proof of this lemma uses the fact that the possible $\delta$-smoothings around $\bv$ are quite similar to the $\delta$-smoothings around $\bw$ (see \Cref{prop:smooth-close}). If we think of these spaces of $\delta$-smoothings as convex bodies, the $L^2$ minimizer of one of these convex bodies must then be in close proximity to the $L^2$ minimizer of the other convex body. This is enough to prove \Cref{lem:req-eq-robust}.

However, $\REQ_{\delta}$ by itself is not a robust rounding scheme for $(P,Q)$ with equality. The reason why is that so far we have only established that if $\bv_i$ and $\bv_j$ are close, then on average over the choice of $R \in \cR$, the outputs $z_i = \REQ_{\delta}(\bv, R)$ and $z_j = \REQ_{\delta}(\bw, R)$ are close. However, remember that the $z_i$'s are rounded into an integral assignment with \emph{independent} coin flips for each $i \in [n]$. In particular, if $z_i = z_j = 0$, then the equality constraint is satisfied only $1/2$ of the time.

To correct this issue, as our final step we modify the independent rounding into \emph{correlated} rounding. In the Boolean setting, this involves picking a random \emph{global} threshold $t \in [-1/2, 1/2]$ and rounding each $x_i$ to $1$ if $z_i > t$ and $-1$ otherwise. Now, if $z_i \approx z_j$, the corresponding equality constraint will almost always be satisfied. The tradeoff is that the robustness $f(\eps)$ for $\PCSP(P,Q)$ gets worse by a constant factor. See \Cref{subsec:cor-round} for precise details.

Altogether, our rounding scheme is $\REQ_{\delta}$ for $\delta = \eps^{O(1)}$ with correlated rounding. Due to various polynomial losses in the course of the proof, the new scheme is (approximately) $(\eps, f(\eps^{1/6}))$-robust.

\subsection{Paper Outline}

In \Cref{subsec:prelim}, we give necessary background material and notation for understanding our main results. \Cref{sec:AT} gives an integrality gap for robustly solving $\fiPCSP\OneNAE$, proving \Cref{thm:AT-informal}. \Cref{sec_majority} gives an asymptotically optimal analysis of Brakensiek--Guruswami--Sandeep's algorithm for the Majority polymorphism, proving \Cref{thm:MAJ-informal}. \Cref{sec_separablesPCSPs} generalizes \Cref{thm:AT-informal} to separable PCSPs, proving in particular \Cref{thm:PLUR-informal}. \Cref{sec:eq-robust} proves \Cref{thm:EQ-informal} and \Cref{thm:gadget-informal} by adapting the rounding scheme of Brown-Cohen and Raghavendra~\cite{BR16:correlation}. In \Cref{sec:concl}, we give some concluding thoughts and list some open questions.

\section{Notation and Preliminaries}\label{subsec:prelim}

\subsection{Promise CSPs}\label{subsec:promise-csp}
Let $D,E$ be finite domains. A \textit{$\PCSP$ template} is a collection $\Gamma=\{(P_1,Q_1),\dots,(P_t,Q_t)\}$ of pairs of relations of arities $k_1,\dots,k_t\in\N$, where for all $i\in [t]$ $P_i\subseteq D^{k_i}$ and $Q_i\subseteq E^{k_i}$. We further require that there exists a \textit{homomorphism} $h$ that, for all $i\in [t]$, sends tuples in $P_i$ to tuples in $Q_i$; i.e., a map $h:D\to E$ such that $h(\bx)\in Q_i$ for each $i\in [t]$ and each $\bx\in P_i$.\footnote{At times, it will be more convenient to denote a PCSP template $\Gamma$ as a pair $(\A,\B)$ of homomorphic \textit{relational structures}, where $\A$ has domain $A$ and relations $P_1,\dots,P_t$ over $A$, and $\B$ has domain $B$ and relations $Q_1,\dots,Q_t$ over $B$. We shall adopt both notations interchangeably.}
When $t=1$ (i.e., when $\Gamma$ contains a unique promise relation $(P,Q)$), we shall also write $\PCSP(P,Q)$ for $\PCSP(\Gamma)$.
Also,
when $P_i=Q_i$ for all $i\in [t]$, we say that $\Gamma$ is a \textit{$\CSP$ template}.

Let $\Gamma$ be a PCSP template. An \textit{instance} of $\PCSP(\Gamma)$ consists of a set of variables $V=\{x_1,\dots,x_n\}$ and a set of constraints $\mathcal{C}=\{C_1,\dots,C_m\}$, where each constraint $C_i$ is described by a tuple of variables $(x_{i,1},\dots,x_{i,k_i})$ and a relation pair $(P_i,Q_i)\in\Gamma$ of arity $k_i$.
We say that a map $g:V\to D$ \textit{strongly} satisfies $C_i$ if $(g(x_{i,1}),\dots,g(x_{i,k_i}))\in P_i$, while we say that a map $g':V\to E$ \textit{weakly} satisfies $C_i$ if $(g'(x_{i,1}),\dots,g'(x_{i,k_i}))\in Q_i$. 
Observe that composing any map $g:V\to D$ with the required homomorphism $h$ between the $P_i$'s and the $Q_i$'s results in a map $g':V\to E$ that weakly satisfies all constraints that are strongly satisfied by $g$.

Consider a function $f:(0,1)\to (0,1)$ with $f(\eps)\to 0$ as $\eps\to 0$. We say that $\PCSP(\Gamma)$ is \textit{robustly solvable with loss $f$} (in short, \textit{$f$-robust}) if there exists an algorithm $\operatorname{Alg}$ with the following properties:
\begin{itemize}
    \item The inputs to $\operatorname{Alg}$ are an instance of $\PCSP(\Gamma)$ and a threshold parameter $\epsilon\in (0,1)$, with the promise that there is a map $g:V\to D$ that strongly satisfies at least $1-\eps$ fraction of the constraints;
    \item $\operatorname{Alg}$ outputs the description of a map $g':V\to E$ that weakly satisfies at least $1-f(\eps)$ fraction of the constraints;
    \item The runtime of $\operatorname{Alg}$ is polynomial in the size of the instance.
\end{itemize}

We shall sometimes consider \textit{weighted} instances of $\PCSP(\Gamma)$, where the constraint set $\mathcal{C}$ is equipped with a probability distribution $w:\mathcal{C}\to\R^+$. In this case, the definition of robust solvability of $\PCSP(\Gamma)$ is entirely analogous, except that it considers the total weight---rather than the fraction---of satisfied constraints. %

\subsection{Basic SDP}
\label{subsec_basic_SDP}

Let $D$ be a finite domain, let $P \subseteq D^k$, and let $\eps \in [0,1]$. Let $N \in \N$ be a positive integer. We say that a tuple of $|D|k+1$ vectors in $\R^N$, $(\bv_0, \{\bv_{1,d} : d \in D\}, \hdots, \{\bv_{k,d} : d \in D\})$ is a \emph{$1-\eps$ approximate vector assignment} to $P$ if there exists a probability distribution $\mu$ on $D^k$ with the following properties.
\begin{align}
    \mu(P) &\ge 1-\eps, \label{eq:meas}\\
    \forall i \in [k], d \in D, \langle \bv_0, \bv_{i,d}\rangle &= \sum_{\ba \in D^k} \one[a_i = d] \mu(\ba),\label{eq:1st-moment}\\
    \forall i,i' \in [k], d,d' \in D, \langle \bv_{i,d}, \bv_{i',d'}\rangle &= \sum_{\ba \in D^k} \one[a_i = d]\one[a_{i'}= d'] \mu(\ba).\label{eq:2nd-moment}
\end{align}

If $D = \{-1,1\}$, this is equivalent to looking at $\bv_i := \bv_{i,1} - \bv_{i,-1}$ subject to the following constraints (see, e.g., \cite{BGS23:stoc} for a justification)
\begin{align}
    \mu(P) &\ge 1-\eps, \label{eq:meas-b}\\
    \forall i \in [k], \langle \bv_0, \bv_{i}\rangle &= \sum_{\ba \in D^k} a_i \mu(\ba),\label{eq:1st-moment-b}\\
    \forall i,i' \in [k], \langle \bv_{i}, \bv_{i'}\rangle &= \sum_{\ba \in D^k} a_ia_{i'}\mu(\ba).\label{eq:2nd-moment-b}
\end{align}

Given an instance of $\CSP(P)$ on variables $x_1, \hdots, x_n$ and clauses $C_1, \hdots, C_m$, we say that a  tuple of $|D|n+1$ vectors $(\bv_0, \{\bv_{1,d} : d \in D\}, \hdots, \{\bv_{n,d} : d \in D\})$ is a basic SDP solution if for every $j \in [m]$, there is a probability distribution $\mu_j$ over $D^k$ such that (\ref{eq:1st-moment}) and (\ref{eq:2nd-moment}) (or equivalently, (\ref{eq:1st-moment-b}) and (\ref{eq:2nd-moment-b}) in the Boolean setting) hold for the $|D|k+1$ tuples of vectors corresponding to the variables $x_{j,1}, \hdots, x_{j,k_j}$ in the clause $C_j$. The value (or completeness) of the SDP solution is equal to $\frac{1}{m}\sum_{j=1}^m \mu_j(P)$. Note that for any $\eps > 0$, one can compute the optimal value of the SDP up to an additive $\eps$ in $(n/\eps)^{O(1)}$ time~\cite{grotschel1993Geometric,gartner2012approximation}.\footnote{Brakensiek, Guruswami, and Sandeep~\cite{BGS23:stoc} note that finding such a solution is somewhat subtle as the constraints (\ref{eq:1st-moment-b}) and (\ref{eq:2nd-moment-b}) need to be exactly satisfied. See Appendix A of the full version of their paper for a more detailed discussion.}

When we refer to the SDP relaxation of an instance of $\PCSP(P, Q)$, we mean the SDP relaxation of the instance as an instance of $\CSP(P)$. The predicate $Q$ only appears in the \emph{rounding} of the SDP relaxation, where the continuous vectors are converted into a discrete assignment.

Given a PCSP template $\Gamma$ and a weighted instance of $\PCSP(\Gamma)$, we define the SDP relaxation of an instance of $\PCSP(\Gamma)$ analogously, where each constraint $C_j$ has the constraints (\ref{eq:1st-moment}) and (\ref{eq:2nd-moment}) corresponding to $P_j$ (and the variables of $C_j$). Furthermore, the value of the CSP is $\sum_{j=1}^m w(C_j)\mu_j(P_j)$.

\subsection{Polymorphisms}
\label{subsec:polymorphism-defn}
Given a pair of predicates $P \subseteq Q \subseteq D^k$, we say that an operator $f : D^L \to D$ is a \emph{polymorphism} of $(P,Q)$ if for any $t^1, \hdots, t^L \in P$, we have that $t \in Q$, where $t_i = f(t^1_i, \hdots, t^L_i)$ for all $i \in [k]$. We let $\Pol(P,Q)$ denote the set of all polymorphisms of $(P,Q)$, with $\Pol^L(P,Q)$ being specifically the set of polymorphisms of arity $L$.

In this paper, we focus on a couple commonly-studied families of polymorphisms. For any odd $L \in \N$, let $\MAJ_L, \AT_L : \{-1,1\}^L \to \{-1,1\}$ be defined as
\begin{align*}
\MAJ_L(x_1, \hdots, x_L) &:= \one\left[\sum_{i=1}^L x_i \ge 0\right],\\
\AT_L(x_1, \hdots, x_L) &:= \one\left[\sum_{i=1}^L (-1)^{i-1}x_i \ge 0\right].
\end{align*}

We let $\MAJ := \{\MAJ_L : L \in \N\text{ odd}\}$ and $\AT := \{\AT_L : L \in \N\text{ odd}\}$. For example, $\MAJ \subseteq \Pol(P,Q)$ means that $\Pol(P,Q)$ has majority polymorphisms of all odd arities and in fact has all folded (i.e., unbiased) weighted threshold functions as polymorphisms~\cite{bgs23:theoretics}.

\subsection{Analytical facts}
\label{prelimns_analytical_facts}

We let $\cN(a,b)$ denote the normal distribution with mean $a$ and variance of $b$ (i.e., standard deviation of $\sqrt{b})$. Likewise, given $\bm{\mu} \in \R^n$ and $\Sigma \in \R^{n \times n}$, we let $\cN(\bm{\mu}, \Sigma)$ denote the $n$-variate normal distribution with mean $\bm{\mu}$ and covariance $\Sigma$. We cite the following standard facts about the normal distribution.

\begin{proposition}[e.g., \cite{vershynin2018high}, \cite{BGS23:stoc}]\label{thm:gauss}
Let $X$ be distributed as $\cN(\mu,\sigma^2)$. Then,
\begin{itemize}
\itemsep=0ex
\item[(a)] For any interval $[a,b] \subseteq \R$, $\Pr[X \in [a,b]] \le \frac{b-a}{2\sigma}$.
\item[(b)] For any $t \ge 0$, $\Pr[X \ge t + \mu] \le e^{-\frac{t^2}{2\sigma^2}}$. 
\end{itemize}
\end{proposition}

\begin{proposition}[e.g., \cite{vershynin2018high}, \cite{BGS23:stoc}]\label{prop:gauss}
Let $x$ be a Gaussian variable with mean zero and variance one. Then for all $t \geq 0$
\begin{enumerate}
\itemsep=0ex
\item $P(x \geq t) \leq \frac{1}{2}e^{-\frac{t^2}{2}}$.
\item $P(t \leq x \leq 2t) \leq \frac{t}{\sqrt{2\pi}}e^{-\frac{t^2}{2}}$.
\item For all $a \in \mathbb{R}$, $P(|x - a| \leq t) \leq \frac{2t}{\sqrt{2\pi}} \le t$.
\item Given $I \subseteq [t, \infty]$, $P(x \in I) \le |I|e^{-\frac{t^2}{2}}$.
\end{enumerate}
\end{proposition}
\begin{proof}
Observe that $P(x \geq t) = \int_{t}^{\infty}{\frac{e^{-\frac{x^2}{2}}}{\sqrt{2\pi}}dx} = \int_{0}^{\infty}{\frac{e^{-\frac{(y+t)^2}{2}}}{\sqrt{2\pi}}dy} \leq e^{-\frac{t^2}{2}}\int_{0}^{\infty}{\frac{e^{-\frac{y^2}{2}}}{\sqrt{2\pi}}dy} = \frac{1}{2}e^{-\frac{t^2}{2}}$, $P(t \leq x \leq 2t) = \int_{t}^{2t}{\frac{e^{-\frac{x^2}{2}}}{\sqrt{2\pi}}dx} \leq \frac{t}{\sqrt{2\pi}}e^{-\frac{t^2}{2}}$, and $P(|x-a| \leq t) = \int_{a-t}^{a+t}{\frac{e^{-\frac{x^2}{2}}}{\sqrt{2\pi}}dx} \leq \frac{2t}{\sqrt{2\pi}}$.
\end{proof}

We shall also use the following concentration bound for the squared norm of matrices $X\sim\cN(\bzero, I_n)$, which follows from Laurent--Massart bounds on the concentration of the chi-squared distribution.
\begin{proposition}[\cite{laurent2000adaptive}]
\label{prop_concentration_chi_squared}
    Let $X$ be distributed as $X\sim\cN(0_n, I_n)$. Then, for each $\alpha\in\R$, it holds that 
    \[\Pr(\|X\|^2\in[n-2n^{1/2+\alpha},n+2n^{1/2+\alpha}+2n^\alpha])\geq1-e^{-n^\alpha}.
    \]
\end{proposition}

\section{Integrality Gap for $\AT$ Polymorphism}\label{sec:AT}

\subsection{Attempted Integrality Gap Instance}
Before giving our full analysis, we first describe a simple attempted integrality gap instance which fails but captures the intuition for our integrality gap instance.

Let $\vec{w}_1,\vec{w}_2,\vec{w}_3$ be three unit vectors that are orthogonal to $\vec{v}_0$ such that $\vec{w}_1 + \vec{w}_2 + \vec{w}_3 = 0$. For example, we could have $\vec{v}_0 = (1,0,0)$, $\vec{w}_1 = (0,1,0)$, $\vec{w}_2 = (0,-\frac{1}{2},\frac{\sqrt{3}}{2})$, and $\vec{w}_3 = (0,-\frac{1}{2},-\frac{\sqrt{3}}{2})$.

Let $k_0 = \lceil{\log_2(n)}\rceil$. Let us pretend that we have the vectors 
\[
\{(1 - 2^{-k})\vec{v}_0 + \vec{w}_i: k \in \{0,1,\ldots,k_0\}, i \in [3]\} \
\bigcup \ \{-(1 - 2^{-k})\vec{v}_0 + \vec{w}_i: k \in \{1,2,\ldots,k_0\}, i \in [3]\}
\]
and moreover, one of the following two cases holds
\begin{enumerate}
\item $(1 - 2^{-k_0})\vec{v}_0 + \vec{w}_i$ is rounded to $1$ for all $i \in [3]$ and $-(1 - 2^{-k_0})\vec{v}_0 + \vec{w}_i$ is rounded to $-1$ for all $i \in [3]$.
\item $(1 - 2^{-k_0})\vec{v}_0 + \vec{w}_i$ is rounded to $-1$ for all $i \in [3]$ and $-(1 - 2^{-k_0})\vec{v}_0 + \vec{w}_i$ is rounded to $1$ for all $i \in [3]$.
\end{enumerate}
See the figure below for an illustration of this set of vectors when $k_0 = 3$.\
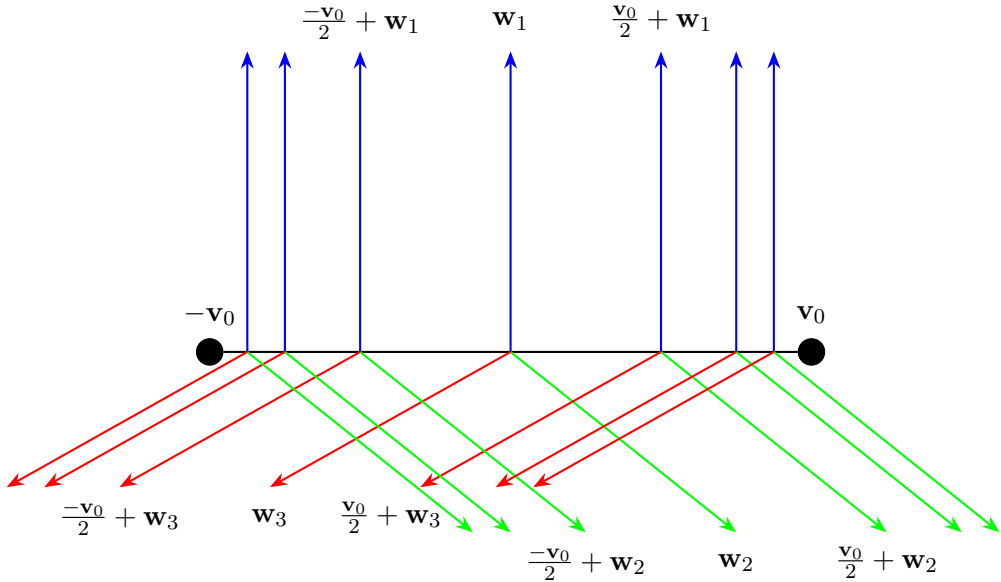
    \begin{center}
    \begin{tikzpicture}[
      mycircle/.style={
         circle,
         draw=black,
         fill=black,
         fill opacity = 1,
         text opacity=1,
         inner sep=0pt,
         minimum size=10pt,
         font=\small},
      myarrow/.style={-Stealth, thick},
      node distance=0.6cm and 1.2cm
      ] 
      \node[mycircle] at (-4,0) (l) {};
      \node[mycircle] at (4,0) (r) {};
      \draw[thick] (l) -- (r);
      \draw[blue] [myarrow] (0,0) -- (0,4);
      \draw[green] [myarrow] (0,0) -- (3,-2.4);
      \draw[red] [myarrow] (0,0) -- (-3.2,-1.8);
      \draw[blue] [myarrow] (2,0) -- (2,4);
      \draw[green] [myarrow] (2,0) -- (5,-2.4);
      \draw[red] [myarrow] (2,0) -- (-1.2,-1.8);
      \draw[blue] [myarrow] (-2,0) -- (-2,4);
      \draw[green] [myarrow] (-2,0) -- (1,-2.4);
      \draw[red] [myarrow] (-2,0) -- (-5.2,-1.8);
      \draw[blue] [myarrow] (3,0) -- (3,4);
      \draw[green] [myarrow] (3,0) -- (6,-2.4);
      \draw[red] [myarrow] (3,0) -- (-0.2,-1.8);
      \draw[blue] [myarrow] (-3,0) -- (-3,4);
      \draw[green] [myarrow] (-3,0) -- (0,-2.4);
      \draw[red] [myarrow] (-3,0) -- (-6.2,-1.8);
      \draw[blue] [myarrow] (3.5,0) -- (3.5,4);
      \draw[green] [myarrow] (3.5,0) -- (6.5,-2.4);
      \draw[red] [myarrow] (3.5,0) -- (0.3,-1.8);
      \draw[blue] [myarrow] (-3.5,0) -- (-3.5,4);
      \draw[green] [myarrow] (-3.5,0) -- (-0.5,-2.4);
      \draw[red] [myarrow] (-3.5,0) -- (-6.7,-1.8);
      \node[] at (-4,0.5) (tl) {$-\vec{v}_0$};
      \node[] at (4,0.5) (tr) {$\vec{v}_0$};
      \node[] at (0,4.4) (tw1a) {$\vec{w}_1$};
      \node[] at (2,4.4) (tw1b) {$\frac{\vec{v}_0}{2} + \vec{w}_1$};
      \node[] at (-2,4.4) (tw1c) {$\frac{-\vec{v}_0}{2} + \vec{w}_1$};
      \node[] at (3,-2.8) (tw2a) {$\vec{w}_2$};
      \node[] at (5,-2.8) (tw2b) {$\frac{\vec{v}_0}{2} + \vec{w}_2$};
      \node[] at (1,-2.8) (tw2c) {$\frac{-\vec{v}_0}{2} + \vec{w}_2$};
      \node[] at (-3.2,-2.2) (tw3a) {$\vec{w}_3$};
      \node[] at (-1.6,-2.2) (tw3b) {$\frac{\vec{v}_0}{2} + \vec{w}_3$};
      \node[] at (-5.2,-2.2) (tw3c) {$\frac{-\vec{v}_0}{2} + \vec{w}_3$};
    \end{tikzpicture}
    \end{center}

If so, then consider the triples 
\[
\{(b(1 - 2^{-k})\vec{v}_0 + \vec{w}_i, b(1 - 2^{-k})\vec{v}_0 + \vec{w}_{i'},-b(1-2^{1-k})\vec{v}_0 + \vec{w}_{i''})\}
\]
where $k \in [k_0]$, $b \in \{-1,1\}$ and $i,i',i''$ are distinct elements of $[3]$. The vectors in each such triple sum to either $\vec{v}_0$ or $-\vec{v}_0$ so for each such triple, we can add a $\One$ 
constraint on either the vectors or their negations. However, it is not hard to show that at least one of the resulting $\NAESAT$ constraints must be violated.

To see this, assume that all of the $\NAE$ constraints are satisfied. We can assume without loss of generality that $(1 - 2^{-k_0})\vec{v}_0 + \vec{w}_i$ is rounded to $1$ for all $i \in [3]$ and $-(1 - 2^{-k_0})\vec{v}_0 + \vec{w}_i$ is rounded to $-1$ for all $i \in [3]$.  We now make the following observations:
\begin{enumerate}
\item $(1 - 2^{-k_0})\vec{v}_0 + \vec{w}_i$ is rounded to $1$ for all $i \in [3]$. 
\item Since we have $\NAESAT$ constraints on the triples  
\[
((-1)^{k_0-k}(1 - 2^{-k})\vec{v}_0 + \vec{w}_i, (-1)^{k_0-k}(1 - 2^{-k})\vec{v}_0 + \vec{w}_{i'},(-1)^{k_0-k+1}(1-2^{1-k})\vec{v}_0 + \vec{w}_{i''})
\]
such that $i,i',i''$ are distinct elements of $[3]$, 
if $(-1)^{k_0-k}(1 - 2^{-k})\vec{v}_0 + \vec{w}_i$ is rounded to $(-1)^{k_0-k}$ for all $i \in [3]$ then $(-1)^{k_0-k+1}(1-2^{1-k})\vec{v}_0 + \vec{w}_{i}$ must be rounded to $(-1)^{k_0-k+1}$ for all $i \in [3]$.
\item $-(1 - 2^{-k_0})\vec{v}_0 + \vec{w}_i$ is rounded to $-1$ for all $i \in [3]$.
\item Since we have $\NAESAT$ constraints on the triples  
\[
((-1)^{k_0-k+1}(1 - 2^{-k})\vec{v}_0 + \vec{w}_i, (-1)^{k_0-k+1}(1 - 2^{-k})\vec{v}_0 + \vec{w}_{i'},(-1)^{k_0-k}(1-2^{1-k})\vec{v}_0 + \vec{w}_{i''})
\]
such that $i,i',i''$ are distinct elements of $[3]$, 
if $(-1)^{k_0-k+1}(1 - 2^{-k})\vec{v}_0 + \vec{w}_i$ is rounded to $(-1)^{k_0-k+1}$ for all $i \in [3]$ then $(-1)^{k_0-k}(1-2^{1-k})\vec{v}_0 + \vec{w}_{i}$ must be rounded to $(-1)^{k_0-k}$ for all $i \in [3]$.
\end{enumerate}
Putting these observations together, we have that for all $k \in [k_0 + 1]$, $(-1)^{k_0-k+1}(1-2^{1-k})\vec{v}_0 + \vec{w}_{i}$ must be rounded to $(-1)^{k_0-k+1}$ for all $i \in [3]$ and $(-1)^{k_0-k}(1-2^{1-k})\vec{v}_0 + \vec{w}_{i}$ must be rounded to $(-1)^{k_0-k}$ for all $i \in [3]$. Plugging in $k = 1$, we have that $\vec{w}_1$, $\vec{w}_2$, and $\vec{w}_3$ must all be rounded to both $1$ and $-1$ which gives a contradiction. This argument is illustrated in the figure below where blue represents vectors which should be rounded to $1$ and red represents vectors which should be rounded to $-1$.
    
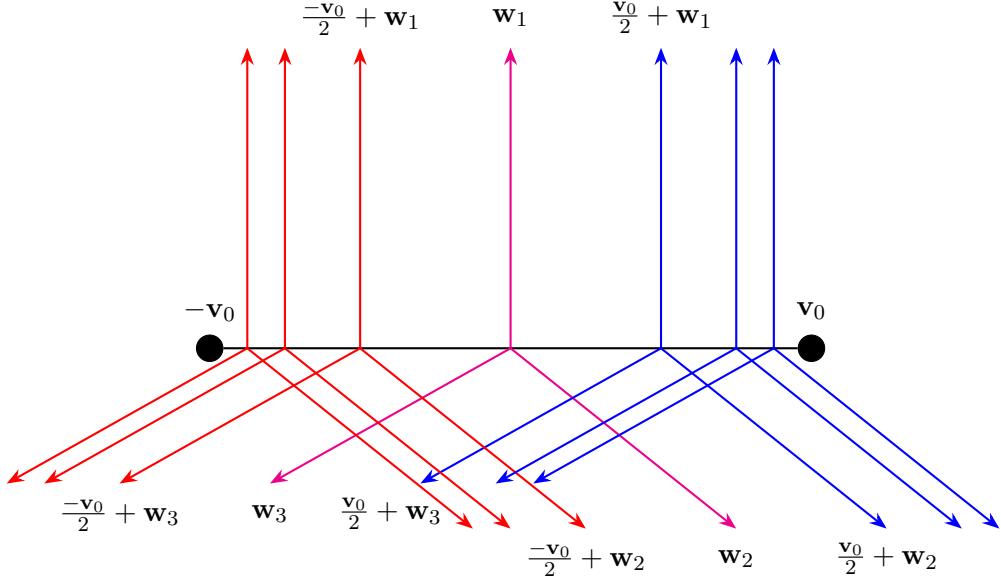
\begin{center}
    \begin{tikzpicture}[
      mycircle/.style={
         circle,
         draw=black,
         fill=black,
         fill opacity = 1,
         text opacity=1,
         inner sep=0pt,
         minimum size=10pt,
         font=\small},
      myarrow/.style={-Stealth, thick},
      node distance=0.6cm and 1.2cm
      ] 
      \node[mycircle] at (-4,0) (l) {};
      \node[mycircle] at (4,0) (r) {};
      \draw[thick] (l) -- (r);
      \draw[magenta] [myarrow] (0,0) -- (0,4);
      \draw[magenta] [myarrow] (0,0) -- (3,-2.4);
      \draw[magenta] [myarrow] (0,0) -- (-3.2,-1.8);
      \draw[blue] [myarrow] (2,0) -- (2,4);
      \draw[blue] [myarrow] (2,0) -- (5,-2.4);
      \draw[blue] [myarrow] (2,0) -- (-1.2,-1.8);
      \draw[red] [myarrow] (-2,0) -- (-2,4);
      \draw[red] [myarrow] (-2,0) -- (1,-2.4);
      \draw[red] [myarrow] (-2,0) -- (-5.2,-1.8);
      \draw[blue] [myarrow] (3,0) -- (3,4);
      \draw[blue] [myarrow] (3,0) -- (6,-2.4);
      \draw[blue] [myarrow] (3,0) -- (-0.2,-1.8);
      \draw[red] [myarrow] (-3,0) -- (-3,4);
      \draw[red] [myarrow] (-3,0) -- (0,-2.4);
      \draw[red] [myarrow] (-3,0) -- (-6.2,-1.8);
      \draw[blue] [myarrow] (3.5,0) -- (3.5,4);
      \draw[blue] [myarrow] (3.5,0) -- (6.5,-2.4);
      \draw[blue] [myarrow] (3.5,0) -- (0.3,-1.8);
      \draw[red] [myarrow] (-3.5,0) -- (-3.5,4);
      \draw[red] [myarrow] (-3.5,0) -- (-0.5,-2.4);
      \draw[red] [myarrow] (-3.5,0) -- (-6.7,-1.8);
      \node[] at (-4,0.5) (tl) {$-\vec{v}_0$};
      \node[] at (4,0.5) (tr) {$\vec{v}_0$};
      \node[] at (0,4.4) (tw1a) {$\vec{w}_1$};
      \node[] at (2,4.4) (tw1b) {$\frac{\vec{v}_0}{2} + \vec{w}_1$};
      \node[] at (-2,4.4) (tw1c) {$\frac{-\vec{v}_0}{2} + \vec{w}_1$};
      \node[] at (3,-2.8) (tw2a) {$\vec{w}_2$};
      \node[] at (5,-2.8) (tw2b) {$\frac{\vec{v}_0}{2} + \vec{w}_2$};
      \node[] at (1,-2.8) (tw2c) {$\frac{-\vec{v}_0}{2} + \vec{w}_2$};
      \node[] at (-3.2,-2.2) (tw3a) {$\vec{w}_3$};
      \node[] at (-1.6,-2.2) (tw3b) {$\frac{\vec{v}_0}{2} + \vec{w}_3$};
      \node[] at (-5.2,-2.2) (tw3c) {$\frac{-\vec{v}_0}{2} + \vec{w}_3$};
    \end{tikzpicture}
    \end{center}
The problem with this attempted integrality gap instance is that these vectors are not unit vectors. In order to have unit vectors, we would need to have the vectors 
\[
\{b(1 - 2^{-k})\vec{v}_0 + \sqrt{2^{1-k} - 2^{-2k}}\vec{w}_i: k \in \{0,1,\ldots,k_0\}, b \in \{-1,1\}, i \in [3]\},
\]
which are illustrated in the figure below. However, we can no longer have $\One$ constraints on these vectors. To handle this, we will consider vectors of the form $b(1 - 2^{-k})\vec{v}_0 + \sqrt{2^{1-k} - 2^{-2k}}\vec{w}$ where $b \in \{-1,1\}$ and $||\vec{w}|| \approx 1$ and show that we can use a similar argument. %

    
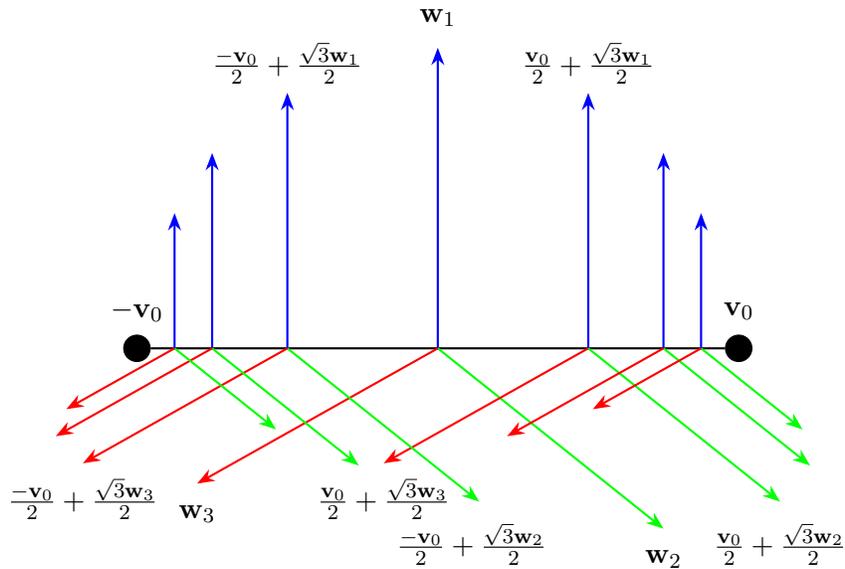
\begin{center}
    \begin{tikzpicture}[
      mycircle/.style={
         circle,
         draw=black,
         fill=black,
         fill opacity = 1,
         text opacity=1,
         inner sep=0pt,
         minimum size=10pt,
         font=\small},
      myarrow/.style={-Stealth, thick},
      node distance=0.6cm and 1.2cm
      ] 
      \node[mycircle] at (-4,0) (l) {};
      \node[mycircle] at (4,0) (r) {};
      \draw[thick] (l) -- (r);
      \draw[blue] [myarrow] (0,0) -- (0,4);
      \draw[green] [myarrow] (0,0) -- (3,-2.4);
      \draw[red] [myarrow] (0,0) -- (-3.2,-1.8);
      \draw[blue] [myarrow] (2,0) -- (2,3.4);
      \draw[green] [myarrow] (2,0) -- (4.55,-2.04);
      \draw[red] [myarrow] (2,0) -- (-0.72,-1.53);
      \draw[blue] [myarrow] (-2,0) -- (-2,3.4);
      \draw[green] [myarrow] (-2,0) -- (0.55,-2.04);
      \draw[red] [myarrow] (-2,0) -- (-4.72,-1.53);
      \draw[blue] [myarrow] (3,0) -- (3,2.6);
      \draw[green] [myarrow] (3,0) -- (4.95,-1.56);
      \draw[red] [myarrow] (3,0) -- (0.92,-1.17);
      \draw[blue] [myarrow] (-3,0) -- (-3,2.6);
      \draw[green] [myarrow] (-3,0) -- (-1.05,-1.56);
      \draw[red] [myarrow] (-3,0) -- (-5.08,-1.17);
      \draw[blue] [myarrow] (3.5,0) -- (3.5,1.8);
      \draw[green] [myarrow] (3.5,0) -- (4.85,-1.08);
      \draw[red] [myarrow] (3.5,0) -- (2.06,-.81);
      \draw[blue] [myarrow] (-3.5,0) -- (-3.5,1.8);
      \draw[green] [myarrow] (-3.5,0) -- (-2.15,-1.08);
      \draw[red] [myarrow] (-3.5,0) -- (-4.94,-.81);
      \node[] at (-4,0.5) (tl) {$-\vec{v}_0$};
      \node[] at (4,0.5) (tr) {$\vec{v}_0$};
      \node[] at (0,4.4) (tw1a) {$\vec{w}_1$};
      \node[] at (2,3.8) (tw1b) {$\frac{\vec{v}_0}{2} + \frac{\sqrt{3}\vec{w}_1}{2}$};
      \node[] at (-2,3.8) (tw1c) {$\frac{-\vec{v}_0}{2} + \frac{\sqrt{3}\vec{w}_1}{2}$};
      \node[] at (3,-2.8) (tw2a) {$\vec{w}_2$};
      \node[] at (4.55,-2.6) (tw2b) {$\frac{\vec{v}_0}{2} + \frac{\sqrt{3}\vec{w}_2}{2}$};
      \node[] at (0.45,-2.6) (tw2c) {$\frac{-\vec{v}_0}{2} + \frac{\sqrt{3}\vec{w}_2}{2}$};
      \node[] at (-3.2,-2.2) (tw3a) {$\vec{w}_3$};
      \node[] at (-0.72,-2) (tw3b) {$\frac{\vec{v}_0}{2} + \frac{\sqrt{3}\vec{w}_3}{2}$};
      \node[] at (-4.72,-2) (tw3c) {$\frac{-\vec{v}_0}{2} + \frac{\sqrt{3}\vec{w}_3}{2}$};
    \end{tikzpicture}
    \end{center}
\subsection{Continuous Integrality Gap Instance}
We now describe a natural modification of the failed integrality gap instance described above which gives a candidate integrality gaps instance with infinitely many variables and constraints. We will then further modify this integrality gap instance to discretize it and shift to multivariate Gaussian vectors rather than unit vectors as multivariate Gaussian vectors are easier to analyze and (with the appropriate scaling)  are close to unit vectors when the dimension is sufficiently high.

Let $k_0 = \lceil{\log(\frac{1}{\epsilon})}\rceil + 20$. We will work in $\mathbb{R}^{d+1}$ where $\vec{v}_0 = \vec{e}_{d+1}$ is orthogonal to $\mathbb{R}^d$ and $d$ will be chosen later.

\begin{definition}
Let $A_{x}$ be the set of points $\{x\vec{v}_{0} + \sqrt{1-x^2}\vec{w}: \vec{w} \in S^{d-1}\}$.%
\end{definition}
\begin{definition}
For each $x$ of the form $1 - 2^{-k}$ where $k \in [k_0]$, 
\begin{enumerate}
\item We define $T_x$ to be the set of triples of points $(\vec{p}_1,\vec{p}_2,\vec{p}_3)$ such that 
$\vec{p}_1 + \vec{p}_2 + \vec{p}_3 = 0$, $\vec{p}_1, \vec{p}_2 \in A_x$, and $\vec{p}_3 \in A_{1-2x}$.
\item Similarly, we define $T_{-x}$ to be the set of triples of points $(\vec{p}_1,\vec{p}_2,\vec{p}_3)$ such that 
$\vec{p}_1 + \vec{p}_2 + \vec{p}_3 = 0$, $\vec{p}_1, \vec{p}_2 \in A_{-x}$, and $\vec{p}_3 \in A_{2x-1}$.
\end{enumerate}  
\end{definition}
Our candidate integrality gap instance is as follows:
\begin{enumerate}
\item For each $k \in [k_0]$, letting $x = 1 - 2^{-k}$, we choose a random triple from $T_x$ with probability $\frac{1}{2(k_0+1)}$ and we choose a random triple from $T_{-x}$ with probability $\frac{1}{2(k_0+1)}$.
\item Let $x_0 = 1 - 2^{-k_0}$. With probability $\frac{1}{2(k_0 + 1)}$, we choose a random $\vec{p} \in A_{x_0}$ and take the triple $(-v_0,-v_0,\vec{p})$. Similarly, with probability $\frac{1}{2(k_0 + 1)}$, we choose a random $\vec{p} \in A_{-x_0}$ and take the triple $(v_0,v_0,\vec{p})$.
\end{enumerate}
Intuitively, this is a gap instance because the SDP thinks that almost all of the constraints can be satisfied but in any actual assignment, one of the following must occur which leads to at least $\frac{\delta}{16(k_0 + 1)}$ of the $\NAESAT$ constraints being violated where $\delta > 0$ is an absolute constant which we will choose later. Without loss of generality, we can assume that $\vec{v}_0$ is rounded to $1$ and $-\vec{v}_0$ is rounded to $-1$.
\begin{enumerate}
\item A $\delta$ proportion of $\vec{p} \in A_{x_0}$ are assigned $-1$ which means that $\delta$ of the triples $(-v_0,-v_0,\vec{p})$ are all $-1$.
\item A $\delta$ proportion of $\vec{p} \in A_{-x_0}$ are assigned $1$ which means that $\delta$ of the triples $(v_0,v_0,\vec{p})$ are all $1$.
\item For some $k \in [k_0]$, letting $x = 1 - 2^{-k}$, less than $\delta$ of the vectors in $A_{x}$  are assigned $-1$ and at least $\delta$ of the vectors in $A_{1-2x}$ are assigned $1$ which implies that at least $\frac{\delta}{8}$ of the triples $(\vec{p}_1,\vec{p}_2,\vec{p}_3) \in T_x$ are all $1$.
\item For some $k \in [k_0]$, letting $x = 1 - 2^{-k}$, less than $\delta$ of the vectors in $A_{-x}$  are assigned $1$ and at least $\delta$ of the vectors in $A_{1-2x}$ are assigned $-1$ which implies that at least $\frac{\delta}{8}$ of the triples $(\vec{p}_1,\vec{p}_2,\vec{p}_3) \in T_{-x}$ are all $-1$.
\end{enumerate}
Before we modify our integrality gap instance and make this intuition rigorous, we give a more precise description of the triples of points in $T_x$.
\begin{proposition}
If $\vec{p}_1, \vec{p}_2 \in A_x$ and $\vec{p}_3 \in A_{1-2x}$ then $\vec{p}_1 + \vec{p}_2 + \vec{p}_3 = \vec{v}_0$ if and only if \\
$\vec{p}_3 = (1-2x)\vec{v}_0 + \sqrt{1 - (2|x|-1)^2}\vec{w}$, $\vec{p}_1 = x\vec{v}_0 + \sqrt{1 - x^2}\left(-\frac{\sqrt{1 - (2|x|-1)^2}}{2\sqrt{1 - x^2}}\vec{w} + \sqrt{1 - \frac{1 - (2|x|-1)^2}{4(1 - x^2)}}\vec{u}\right)$, and $\vec{p}_2 = x\vec{v}_0 + \sqrt{1 - x^2}\left(-\frac{\sqrt{1 - (2|x|-1)^2}}{2\sqrt{1 - x^2}}\vec{w} - \sqrt{1 - \frac{1 - (2|x|-1)^2}{4(1 - x^2)}}\vec{u}\right)$ for some  $\vec{w},\vec{u} \in S^{d-1}$ such that $\vec{w} \cdot \vec{u} = 0$.
\end{proposition}
\begin{proof}
Assume that $\vec{p}_1, \vec{p}_2 \in A_x$, $\vec{p}_3 \in A_{1-2x}$, and $\vec{p}_1 + \vec{p}_2 + \vec{p}_3 = \vec{v}_0$. Since $\vec{p}_3 \in A_{1-2x}$, $\vec{p}_3 = (1-2x)\vec{v}_0 + \sqrt{1 - (2|x|-1)^2}\vec{w}$ for some $\vec{w} \in S^{d-1}$. Now write $\vec{p}_1 = x\vec{v}_0 + a\vec{w} + b\vec{u}$ where $\vec{u} \in S^{d-1}$, $\vec{w} \cdot \vec{u} = 0$, and $b \geq 0$. We now make the following observations:
\begin{enumerate}
\item Since $\vec{p}_1 + \vec{p}_2 + \vec{p}_3 = \vec{v}_0$, 
$\vec{p}_2 = x\vec{v}_0 - \left(a + \sqrt{1 - (2|x|-1)^2}\right)\vec{w} - b\vec{u}$.
\item Since $||\vec{p}_1||^2 = ||\vec{p}_2||^2$, $a^2 = \left(a + \sqrt{1 - (2|x|-1)^2}\right)^2$ which implies that $a = -\frac{1}{2}\sqrt{1 - (2|x|-1)^2}$.
\item Since $||\vec{p}_1||^2 = 1$, $b = \sqrt{1 - x^2 - a^2} = \sqrt{1 - x^2 - \frac{1 - (2|x|-1)^2}{4}}$.
\end{enumerate}
The other direction is easy to verify directly.
\end{proof}
Based on this, we define the following noise coefficient $\rho_x$.
\begin{definition}
We define $\rho_x = \frac{\sqrt{1 - (2|x|-1)^2}}{2\sqrt{1 - x^2}}$.
\end{definition}
Observe that if $|x| = 1 - 2^{-k}$ and $k$ is large then $\rho_x = \frac{\sqrt{1 - (2|x|-1)^2}}{2\sqrt{1 - x^2}} \approx \frac{\sqrt{2^{2-k}}}{2\sqrt{2^{1-k}}} = \frac{1}{\sqrt{2}}$.
\subsection{Discretized Integrality Gap Instance}
\label{subsec_discretized_IGI}
We now modify our candidate integrality gap instance by having it use multivariate Gaussian vectors rather than unit vectors and discretizing it.
\begin{definition}
Given $x \in [-1,1]$, we define $A'_x$ to be the distribution of vectors obtained by sampling $\vec{w} \sim \cN(0,1/d)^d$ and taking the vector $x\vec{v}_0 + \sqrt{1-x^2}\vec{w}$.
\end{definition}
\begin{definition}
Given $x \in [-1,1]$, we define $T'_x$ to be the distribution of triples of points $(\vec{p}_1,\vec{p}_2,\vec{p}_3)$ obtained by sampling $\vec{u},\vec{w} \sim \cN(0,1/d)^d$ and then taking 
$\vec{p}_3 = (1-2x)\vec{v}_0 + \sqrt{1 - (2|x|-1)^2}\vec{w}$, \\
$\vec{p}_1 = x\vec{v}_0 + \sqrt{1 - x^2}\left(-\frac{\sqrt{1 - (2|x|-1)^2}}{2\sqrt{1 - x^2}}\vec{w} + \sqrt{1 - \frac{1 - (2|x|-1)^2}{4(1 - x^2)}}\vec{u}\right)$, and \\
$\vec{p}_2 = x\vec{v}_0 + \sqrt{1 - x^2}\left(-\frac{\sqrt{1 - (2|x|-1)^2}}{2\sqrt{1 - x^2}}\vec{w} - \sqrt{1 - \frac{1 - (2|x|-1)^2}{4(1 - x^2)}}\vec{u}\right)$.

Similarly, we define $T'_{-x}$ to be the distribution of triples of points $(\vec{p}_1,\vec{p}_2,\vec{p}_3)$ obtained by sampling $\vec{u},\vec{w} \sim \cN(0,1/d)^d$ and then taking 
$\vec{p}_3 = (2x-1)\vec{v}_0 + \sqrt{1 - (2|x|-1)^2}\vec{w}$, \\
$\vec{p}_1 = -x\vec{v}_0 + \sqrt{1 - x^2}\left(-\frac{\sqrt{1 - (2|x|-1)^2}}{2\sqrt{1 - x^2}}\vec{w} + \sqrt{1 - \frac{1 - (2|x|-1)^2}{4(1 - x^2)}}\vec{u}\right)$, and \\
$\vec{p}_2 = -x\vec{v}_0 + \sqrt{1 - x^2}\left(-\frac{\sqrt{1 - (2|x|-1)^2}}{2\sqrt{1 - x^2}}\vec{w} - \sqrt{1 - \frac{1 - (2|x|-1)^2}{4(1 - x^2)}}\vec{u}\right)$.
\end{definition}
Our modified integrality gap instance is as follows:
\begin{enumerate}
\item For each $k \in [k_0]$, letting $x = 1 - 2^{-k}$, we sample a triple from $T'_x$ with probability $\frac{1}{2(k_0+1)}$ and we sample a triple from $T'_{-x}$ with probability $\frac{1}{2(k_0+1)}$.
\item Let $x_0 = 1 - 2^{-k_0}$. With probability $\frac{1}{2(k_0 + 1)}$, we sample a vector $\vec{p}$ from $A'_{x_0}$ and take the triple $(-v_0,-v_0,\vec{p})$. Similarly, with probability $\frac{1}{2(k_0 + 1)}$, we sample a vector $\vec{p}$ from $A'_{-x_0}$ and take the triple $(v_0,v_0,\vec{p})$.
\item We discard any triples such that one of the following is true:
\begin{enumerate}
\item When we sampled $\vec{w}$ from $\cN(0,1/d)^d$, $|(||\vec{w}||^2 - 1)| > \frac{\epsilon}{1000}$.
\item When we sampled $\vec{u}$ from $\cN(0,1/d)^d$, either $|(||\vec{u}||^2 - 1)| > \frac{\epsilon}{1000}$ or $|\vec{w} \cdot \vec{u}| > \frac{\epsilon}{1000}$.
\end{enumerate}
\item We discretize the instance by partitioning the set of vectors in $\mathbb{R}^{d+1}$ of length at most $2$ into regions of radius at most $\frac{\epsilon}{1000}$ and choosing a center for each region which is within $\frac{\epsilon}{1000}$ of all points in the region. We also choose a center  for the region with vectors of length greater than $2$ (we can choose any point for this center as only discarded triples of points will have a vector in this region). We then replace all of the vectors with the center of the region they are contained in.
\end{enumerate}
\begin{restatable}{theorem}{thmoneinthreeversusNAE}
\label{thm:1in3versusNAE}
For all $\epsilon > 0$, if $d \geq \frac{80000000}{\epsilon^3}$ than for the integrality gap instance described above (\Cref{subsec_discretized_IGI}),
\begin{enumerate}
\item There is a solution to the SDP that gives a value of at least $1 - \epsilon$ to each constraint that is not discarded.
\item At most $6e^{-\frac{10}{\epsilon}}$ of the constraints that were sampled are discarded.
\item All assignments violate at least $\frac{1}{10^{60}\lceil{log\left(\frac{1}{\epsilon}\right)}\rceil}$ of the sampled $\NAESAT$ constraints (this includes constraints that were discarded).
\end{enumerate}
\end{restatable}

By Raghavendra's theorem~\cite{Raghavendra08:everycsp}, this integrality gap immediately translates into UGC hardness.

\begin{corollary}
Assuming UGC, for all $\epsilon$ such that $0 < \epsilon \leq \frac{1}{100}$, $\fiPCSP\OneNAE$ is not $(\eps, \frac{1}{10^{61}\log(1/\eps)})$-robust.
\end{corollary}

\begin{proof}[Proof sketch for the first statement of Theorem \ref{thm:1in3versusNAE}:]
For each vector $\vec{v}$ that is a center of a region that is not the outer region, we can take the vector $\vec{v}' = \sqrt{1 - \frac{\epsilon}{4}}\frac{\vec{v}}{||\vec{v}||} + \frac{\sqrt{\epsilon}}{2}\vec{z}_v$ for some vector $\vec{z}_v$ which is orthogonal to everything else. We assign $\vec{v}_0$ to itself.

To show that this gives a valid SDP solution which has a value of at least $1 - \epsilon$ for each constraint which is not discarded, we first observe that for each triple of unit vectors $\vec{v}_1, \vec{v}_2, \vec{v}_3$ such that $\vec{v}_1 + \vec{v}_2 + \vec{v}_3 = \pm\vec{v}_0$, there is a probability distribution of $D$ of satisfying assignments to the corresponding constraint such that 
\begin{enumerate}
\item For all $i \in [3]$, $\E_{D}[x_i] = \vec{v}_i \cdot \vec{v}_0$.
\item For all distinct $i,j \in [3]$, $\E_{D}[{x_i}{x_j}] = \vec{v}_i \cdot \vec{v}_j$.
\end{enumerate}
We then observe that if $\vec{v}'_1$, $\vec{v}'_2$, and $\vec{v}'_3$ are unit vectors which are close to $\vec{v}_1$, $\vec{v}_2$, and $\vec{v}_3$ then there is a pseudo-distribution $D'$ (which may give negative probabilities to some assignments $(x_1,x_2,x_3) \in \{-1,1\}^3$) which is close to $D$ such that 
\begin{enumerate}
\item For all $i \in [3]$, $\E_{D'}[x_i] = \vec{v}'_i \cdot \vec{v}_0$.
\item For all distinct $i,j \in [3]$, $\E_{D'}[{x_i}{x_j}] = \vec{v}'_i \cdot \vec{v}'_j$.
\end{enumerate}
Finally, we observe that replacing each vector $\vec{v}$ with $\sqrt{1 - \frac{\epsilon}{4}}\frac{\vec{v}}{||\vec{v}||} + \frac{\sqrt{\epsilon}}{2}\vec{z}_v$ modifies each pseudo-distribution $D'$ by replacing each probability $p$ with $(1-\frac{\epsilon}{4})p + \frac{\epsilon}{32}$ which is sufficient to make all of the probabilities non-negative.
\end{proof}

The full details of the proof of the first statement of~\Cref{thm:1in3versusNAE} are given in~\Cref{app_omitted_proofs}. 

\begin{proof}[Proof of the second statement of Theorem \ref{thm:1in3versusNAE}:]
We can prove the second statement of Theorem \ref{thm:1in3versusNAE} using Chernoff bounds. For completeness, we include a proof of these bounds in~\Cref{app_omitted_proofs}. %
\begin{restatable}{lemma}{lemChernoffbounds}
\label{lem:Chernoffbounds} 
The following tail bounds hold:
\begin{enumerate}
\item For all $t \in [0,1]$, if $\vec{w} \sim \cN(0,1/d)^d$ then $P\left(|(||\vec{w}||^2 - 1)| \geq t\right) \leq 2e^{-\frac{dt^2}{8}}$.
\item Given a vector $\vec{w}$, for all $t \geq 0$, if $\vec{u} \sim \cN(0,1/d)^d$ is independent of $\vec{w}$ then 
$P(|\vec{w} \cdot \vec{u}| \geq t) \leq 2e^{-\frac{d{t^2}}{2||\vec{w}||^2}}$.
\end{enumerate}
\end{restatable}
By Lemma \ref{lem:Chernoffbounds}, since $d \geq \frac{80000000}{\epsilon^3}$, if $\vec{w}$ and $\vec{u}$ are drawn independently from $\cN(0,1/d)^d$ then 
\begin{enumerate}
\item $P(|(||\vec{w}||^2 - 1)| \geq \frac{\epsilon}{1000}) = P(|(||\vec{u}||^2 - 1)| \geq \frac{\epsilon}{1000})\leq 2e^{-\frac{10}{\epsilon}}$
\item Whenever $||\vec{w}|| \leq 2$, $P(|\vec{w} \cdot \vec{u}| \geq \frac{-\epsilon}{1000}) \leq 2e^{-\frac{10}{\epsilon}}$.
\end{enumerate}
so the proportion of discarded constraints is at most $6e^{-\frac{10}{\epsilon}}$.
\end{proof}

\subsection{Soundness Analysis}
We now prove the third statement of Theorem \ref{thm:1in3versusNAE}. To do this, we view an assignment to the points in our integrality gap instance as a map $f: \mathbb{R}^{d+1} \to \{-1,1\}$ where $f(\vec{v})$ is the value given to the point corresponding to the center of the region $\vec{v}$ is contained in.
\begin{definition}
Given a map $f: \mathbb{R}^{d+1} \to \{-1,1\}$, for each vector $\vec{w} \in \mathbb{R}^{d}$, we define $f_x(\vec{w})$ to be the value given to $\vec{p} = x\vec{v_0} + \sqrt{1-x^2}\vec{w}$.
\end{definition}
\begin{definition}
Given $\rho \in [-1,1]$ and a function $f: \mathbb{R}^{d} \to \mathbb{R}$, we define 
\[(N_{\rho}f)(\vec{x}) = E_{\vec{y} \sim \cN(0,1/d)^d}[f({\rho}\vec{x} + \sqrt{1 - {\rho}^2}\vec{y})].
\]
\end{definition}
The third statement of Theorem \ref{thm:1in3versusNAE} follows from the following lemma.
\begin{lemma}
For all $\delta \in [0,10^{-55}]$, if there is an $x = 1 - 2^{k}$ such that $E_{\vec{w} \sim \cN(0,1/d)^d}[f_x(\vec{w})] > 1 - 2\delta$ but $E_{\vec{w} \sim \cN(0,1/d)^d}[f_{1-2x}(\vec{w})] \geq 2\delta - 1$ then the proportion of monochromatic triples $(\vec{p}_1, \vec{p}_2, \vec{p}_3)$ in $T'_{x}$ is at least $\frac{\delta}{8}$.
\end{lemma}
\begin{proof}
Assume that this is not true. Observe that whenever $f_{1-2x}(\vec{w}) = 1$, the proportion of triples $(\vec{p}_1,\vec{p}_2,(1-2x)\vec{v}_0 + \sqrt{1 - (2|x|-1)^2}\vec{w})$ which are monochromatic is at least 
\[
1 - 2P_{\vec{w}' \sim \cN(0,1/d)^d}(f_x({\rho_x}\vec{w} + \sqrt{1 - \rho_x^2}\vec{w}') = -1) = (N_{\rho_x}f_x)(-\vec{w}).
\]

Since $\E_{\vec{w} \sim \cN(0,1/d)^d}[f_{1-2x}(\vec{w})] \geq 2\delta - 1$, $f_{1-2x}(\vec{w}) = 1$ for at least $\delta$ proportion of $\vec{w}$. This implies that in order to avoid having $\frac{\delta}{8}$ proportion of monochromatic triples, we must have that $(N_{\rho_x}f_x)(-\vec{w}) < \frac{1}{4}$ for at least $\frac{\delta}{2}$ proportion of $\vec{w} \sim \cN(0,1/d)^d$. We now show that this cannot occur. To do this, we write $f_x = \E[f_x] + f'_x$ and show that $||(N_{\rho_x}f'_x)||^2$ is small.
\begin{lemma}\label{lem:noisynormbound}
For all $t \in \mathbb{N}$, %
if $\E_{\vec{w} \sim \cN(0,1/d)^d}[f_x(\vec{w})] > 1 - 2\delta$ then %
\[
\E_{\vec{w} \sim \cN(0,1/d)^d}[(N_{\rho_x}f'_x)(\vec{w})^2] \leq 10(3^{t})\delta^{\frac{3}{2}} + 4{\rho_{x}^{2t}}\delta
\]
\end{lemma}
\begin{proof}
Let $S = \{\vec{w} \in \mathbb{R}^{d}: f_x(\vec{w}) = -1\}$ and let $\delta' = \E_{\vec{w} \sim \cN(0,1/d)^d}[1_{\vec{w} \in S}] < \delta$ be the probability that $f_x(\vec{w}) = -1$ for a random $\vec{w} \sim \cN(0,1/d)^d$. {Letting $\{h_i: i \in \mathbb{N} \cup \{0\}\}$ be the Hermite polynomials, we} can write $f_x = (1 - 2\delta') + f'_{x,low} + f'_{x,high}$ where $f'_{x,low}$ is {a linear combination of the polynomials $\{\prod_{j=1}^{d}{h_{i_j}(\sqrt{d}w_j)}: 1 \leq \sum_{j=1}^{d}{i_j} \leq t\}$ (i.e., the generalized multivariate Hermite polynomials of degree between $1$ and $t$ which are orthogonal under the distribution $\cN(0,1/d)^d$) and $f'_{x,high} = f_x - (1 - 2\delta') - f'_{x,low}$ is a linear combination of the polynomials $\{\prod_{j=1}^{d}{h_{i_j}(\sqrt{d}w_j)}: t+1 \leq \sum_{j=1}^{d}{i_j} \leq deg(f_x)\}$.}%

We observe that if we have a bound $\E_{\vec{w} \sim \cN(0,1/d)^d}[f'_{x,low}(\vec{w})^2] \leq B$ then we can make the following deductions.
\begin{enumerate}
\item As described in Chapter 9 of Ryan O'Donnell's textbook ``Analysis of Boolean Functions'' \cite{o2014analysis}, the Bonami Lemma (which is a special case of the Hypercontractivity Theorem) says that if $p$ is a polynomial of degree at most $t$ then $\E_{\vec{x} \in \{-1,1\}^n}[p(\vec{x})^4] \leq {9^t}\left(\E_{\vec{x} \in \{-1,1\}^n}[p(\vec{x})^2]\right)^2$. As observed by Gross \cite{gross1975logarithmicsobolev}, since Gaussian inputs can be approximated by a linear combination of a large number of Boolean variables, the Hypercontractivity Theorem (and in particular the Bonami Lemma) apply for Gaussian inputs as well as Boolean inputs. In other words,
$\E_{\vec{w} \sim \cN(0,1/d)^d}[p(\vec{w})^4] \leq {9^t}\left(\E_{\vec{w} \sim \cN(0,1/d)^d}[p(\vec{w})^2]\right)^2$. Thus, $\E_{\vec{w} \sim \cN(0,1/d)^d}[(f'_{x,low}(\vec{w}))^4] \leq {9^t}B^2$ which implies that $\E_{\vec{w} \sim \cN(0,1/d)^d: f_{x}(\vec{w}) = -1}[|f'_{x,low}(\vec{w})|] \leq \sqrt[4]{\frac{{9^t}B^2}{\delta'}}$.

For a discussion of the history of Bonami's Lemma and the Hypercontractivity Theorem, see Chapter 9.7 of O'Donnell's book \cite{o2014analysis}.

\item For all $\vec{w} \in S^{d-1}$ such that $f_x(\vec{w}) = -1$, $f'_{x,high}(\vec{w}) = 2\delta' - 1 - f'_{x,low}(\vec{w})$. This implies that 
\begin{align*}
\E_{\vec{w} \sim \cN(0,1/d)^d: f_x(\vec{w}) = -1}[f'_{x,low}(\vec{w})f'_{x,high}(\vec{w})] &= \E_{\vec{w} \sim \cN(0,1/d)^d: f_x(\vec{w}) = -1}[(2\delta' - 1)f'_{x,low}(\vec{w}) - f'_{x,low}(\vec{w})^2] \\
&\leq \sqrt[4]{\frac{{9^t}B^2}{\delta'}} - \E_{\vec{w} \sim \cN(0,1/d)^d: f_x(\vec{w}) = -1}[f'_{x,low}(\vec{w})^2]
\end{align*}
\item For all $\vec{w} \in S^{d-1}$ such that $f_x(\vec{w}) = 1$, $f'_{x,high}(\vec{w}) = 2\delta' - f'_{x,low}(\vec{w})$. This implies that 
\begin{align*}
\E_{\vec{w} \sim \cN(0,1/d)^d: f_x(\vec{w}) = 1}[f'_{x,low}(\vec{w})f'_{x,high}(\vec{w})] &= \E_{\vec{w} \sim \cN(0,1/d)^d: f_x(\vec{w}) = 1}[2{\delta'}f'_{x,low}(\vec{w}) - f'_{x,low}(\vec{w})^2]\\
&\leq 2(\delta')^2 - \frac{1}{2}\E_{\vec{w} \sim \cN(0,1/d)^d: f_x(\vec{w}) = 1}[f'_{x,low}(\vec{w})^2]
\end{align*}
where the inequality uses the observation that 
\[
2{\delta'}f'_{x,low}(\vec{w}) - f'_{x,low}(\vec{w})^2 = 2{\delta'}^2 - \frac{1}{2}\left(2{\delta'}-f'_{x,low}(\vec{w})\right)^2 - \frac{1}{2}f'_{x,low}(\vec{w})^2 \leq 2{\delta'}^2 - \frac{1}{2}f'_{x,low}(\vec{w})^2
\]
\item Putting these observations together, 
\begin{align*}
0 = \E_{\vec{w} \sim \cN(0,1/d)^d}[f'_{x,low}(\vec{w})f'_{x,high}(\vec{w})] &= {\delta'}\E_{\vec{w} \sim \cN(0,1/d)^d: f_x(\vec{w}) = -1}[f'_{x,low}(\vec{w})f'_{x,high}(\vec{w})] \\
&+ (1-\delta')\E_{\vec{w} \sim \cN(0,1/d)^d: f_x(\vec{w}) = -1}[f'_{x,low}(\vec{w})f'_{x,high}(\vec{w})] \\
&\leq {\delta'}\sqrt[4]{\frac{{9^t}B^2}{\delta'}} + 2(\delta')^2 - \frac{1}{2}\E_{\vec{w} \sim \cN(0,1/d)^d}[f'_{x,low}(\vec{w})^2]
\end{align*}
which implies that 
\[
\E_{\vec{w} \sim \cN(0,1/d)^d}[f'_{x,low}(\vec{w})^2] \leq 2\left(\sqrt[4]{{9^t}{(\delta')^3}B^2} + 2(\delta')^2\right)
\]
\end{enumerate}
Thus, if $B > 2\sqrt[4]{{9^t}{(\delta')^3}B^2} + 4(\delta')^2$ then we can improve our bound on $\E_{\vec{w} \in S^{d-1}}[f'_{x,low}(\vec{w})^2]$. Letting $B^{*}$ be the infimum of the bounds we can obtain, we must have that $B^{*} = 2\sqrt[4]{{9^t}{(\delta')^3}(B^{*})^2} + 4(\delta')^2$. Solving for $B^{*}$ gives 
\[
\sqrt{B^{*}} = \sqrt[4]{C{(\delta')^3}} + \sqrt{\sqrt{{9^t}(\delta')^3} + 4(\delta')^2} \leq \max{\{3\sqrt[4]{{9^t}{(\delta')^3}}, 3{\delta}'\}} \leq 3\sqrt[4]{{9^t}{(\delta')^3}}
\]
where the first inequality follows by considering whether $4(\delta')^2 \leq 3\sqrt{{9^t}(\delta')^3}$ or $4(\delta')^2 > 3\sqrt{{9^t}(\delta')^3}$. Thus,  $\E_{\vec{w} \sim \cN(0,1/d)^d}[f'_{x,low}(\vec{w})^2] \leq 10\sqrt{{9^t}{\delta}^3}$ which implies that $\E_{\vec{w} \sim \cN(0,1/d)^d}[(N_{\rho_x}f'_{x,low})(\vec{w})^2] \leq 10\sqrt{{9^t}{\delta}^3}$.

Since $\E_{\vec{w} \sim \cN(0,1/d)^d}[(f(\vec{w}) - (1 - 2\delta'))^2] = 4(1-\delta'){\delta'}^2 + {\delta'}(-2+2\delta')^2 \leq 4{\delta'}$, we have that $\E_{\vec{w} \sim \cN(0,1/d)^d}[f'_{x,high}(\vec{w})^2] \leq 4{\delta'}$ which implies that $\E_{\vec{w} \sim \cN(0,1/d)^d}[(N_{\rho_x}f'_{x,high})(\vec{w})^2] \leq 4{\rho_{x}^{2t}}{\delta}$. The result now follows as
\begin{align*}
\E_{\vec{w} \sim \cN(0,1/d)^d}[(N_{\rho_x}f'_{x})(\vec{w})^2] &= \E_{\vec{w} \sim \cN(0,1/d)^d}[(N_{\rho_x}f'_{x,low})(\vec{w})^2] + \E_{\vec{w} \sim \cN(0,1/d)^d}[(N_{\rho_x}f'_{x,high})(\vec{w})^2]\\
&{\leq 10(3^{t})\delta^{\frac{3}{2}} + 4{\rho_{x}^{2t}}\delta.}
\end{align*}
\end{proof}
We now show that {the proportion of monochromatic triples must be at least $\frac{\delta}{8}$.} %
As we observed earlier, in order to avoid having $\frac{\delta}{8}$ proportion of monochromatic triples, we must have that $(N_{\rho_x}f_x)(-\vec{w}) < \frac{1}{4}$ for at least $\frac{\delta}{2}$ proportion of $\vec{w} \sim \cN(0,1/d)^d$. Since $N_{\rho_x}f_x = \E_{\vec{w} \sim \cN(0,1/d)^d}[f_x(\vec{w})] + N_{\rho_x}f'_x \geq 1 - 2\delta + N_{\rho_x}f'_x$, we have that $|(N_{\rho_x}f'_x)(\vec{w})| \geq \frac{3}{4} - 4\delta$ for at least $\frac{\delta}{2}$ proportion of $\vec{w} \in S^{d-1}$. Since $\delta \leq \frac{1}{16}$, this implies that $\E_{\vec{w} \sim \cN(0,1/d)^d}[(N_{\rho_x}f'_x)(\vec{w})^2] \geq \frac{\delta}{2} \cdot (\frac{1}{2})^2 = \frac{\delta}{8}$.

However, Lemma \ref{lem:noisynormbound} implies that if $\delta > 0$ is sufficiently small then $\E_{\vec{w} \sim \cN(0,1/d)^d}[(N_{\rho_x}f'_x)] \leq 
\frac{\delta}{10}$. To see this, take $t = 50$ and observe that since each $\rho$ that appears is at most $.9$, $\rho^{2t} \leq \frac{1}{80}$. Since $\delta < \frac{1}{40000(9^{t})}$, $10(3^{t})\delta^{\frac{3}{2}} \leq \frac{\delta}{20}$. By Lemma \ref{lem:noisynormbound},
\[
\E_{\vec{w} \sim \cN(0,1/d)^d}[(N_{\rho_x}f'_x)(\vec{w})^2] \leq 10(3^{t})\delta^{\frac{3}{2}} + 4{\rho^{2t}}\delta \leq \frac{\delta}{20} + \frac{\delta}{20} = \frac{\delta}{10}. %
\]
{This gives a contradiction so the proportion of monochromatic triples must be at least $\frac{\delta}{8}$, as needed.}
\end{proof}

\section{Improved Analysis for $\MAJ$ Polymorphism}
\label{sec_majority}

In this section, we prove that the robust approximate algorithm for 2-SAT of Charikar--Makarychev--Makarychev\cite{CMM09:talg} has the same $f(\eps) = O(\sqrt{\eps})$ loss guarantee for any PCSP template with the $\MAJ$ polymorphism family. %
This improves over the analysis of Brakensiek--Guruswami--Sandeep~\cite{BGS23:stoc}, which only proved $f(\eps) = \widetilde{O}(\sqrt[3]{\eps})$ loss. Furthermore, in the next section, we shall extend the result (with a slightly worse loss guarantee) to PCSP templates on non-Boolean domains with more general types of polymorphisms.

\begin{theorem}\label{thm:robust-MAJ}
Let $\Gamma$ be a promise template with $\MAJ \subseteq \Pol(\Gamma)$. Then, $\PCSP(\Gamma)$ has a uniform robust algorithm with loss function $f(\eps) = O_{\Gamma}(\sqrt{\eps}).$
\end{theorem}

\subsection{CMM Algorithm}
\label{sec_CMM_algorithm}

We state the algorithm of Charikar--Makarychev--Makarychev (henceforth ``CMM'') as follows.

\begin{itemize}
\item Input: $\eps > 0$, weighted instance of $\PCSP(\Gamma)$ on variables $x_1, \hdots, x_n$ and clauses $C_1, \hdots, C_m$ as in \Cref{subsec:promise-csp}.
\item Solve the basic SDP to find a vector solution $\bv_0, \bv_1, \hdots, \bv_n \in \R^{n+1}$ with objective value at least $1-\eps$.
\item Sample a random Gaussian $\br \in \cN(0_{n+1},I_{n+1})$. 
\item For all $i \in [n]$, round $x_i$ to $+1$ if $\langle \bv_i, \bv_0 + \br \sqrt{\eps}\rangle \ge 0$ and $-1$ otherwise.
\end{itemize}

We call this the $\eps$-CMM algorithm to make the choice of $\eps$ explicit. We make use of the following structural lemma proved by \cite{BGS23:stoc}.
\begin{lemma}[Lemma~4.5 of \cite{BGS23:stoc}, adapted]\label{lem:MAJ-pol}
Let $P \subseteq Q \subseteq \{-1,1\}^k$ be such that $\MAJ \subseteq \Pol(P,Q)$. For any $\bb \not\in Q$, there exists $\bw \in \R^k$ with the following properties. 
\begin{itemize}
\item $\|\bw\|_1 = 1$. 
\item For all $i \in [k]$, $b_iw_i \le 0$.
\item For all $\ba \in P$, $\langle \ba, \bw\rangle \ge 0$.
\end{itemize}
\end{lemma}
We note that this lemma is proven in much more generality in \Cref{prop:plur}.
As a corollary, we have that an approximate version of Lemma~\ref{lem:MAJ-pol} holds for approximate solutions to the SDP.
\begin{lemma}[Implicit in \cite{BGS23:stoc}]\label{lem:robust-MAJ-pol}
Let $P \subseteq Q \subseteq \{-1,1\}^k$ be such that $\MAJ \subseteq \Pol(P,Q)$. Fix $\bb \in \{-1,1\}^k \setminus Q$ and let $\bw \in \R^k$ be as guaranteed by Lemma~\ref{lem:MAJ-pol}. Then, for any $1-\gamma$ approximate vector assignment $(\bv_0, \bv_1, \hdots, \bv_k)$ to $P$ we have that
\begin{align}
    \sum_{i=1}^k w_i \langle\bv_i, \bv_0\rangle \ge -\gamma.\label{eq:robust-MAJ}
\end{align}
\end{lemma}
\begin{proof}
Since $(\bv_0, \hdots, \bv_k)$ is a $1-\gamma$ approximate vector assignment to $P$, there exists a probability distribution $\mu$ over $\{-1,1\}^k$ such that $\mu(P) \ge 1-\gamma$ and for all $i \in [k]$,
\[
    \langle\bv_i, \bv_0\rangle = \sum_{\ba \in \{-1,1\}^k} a_i \mu(\ba).
\]
Thus, we have that
\begin{align*}
    \sum_{i=1}^k w_i \langle\bv_i, \bv_0\rangle &= \sum_{i=1}^k w_i  \sum_{\ba \in \{-1,1\}^k} a_i \mu(\ba)\\
    &= \sum_{\ba \in \{-1,1\}^k} \langle \ba, \bw\rangle \mu(\ba).
\end{align*}
By Lemma~\ref{lem:MAJ-pol}, $\langle \ba, \bw\rangle \ge 0$ for all $\ba \in P$. Furthermore, $\langle \ba, \bw\rangle \ge -1$ for all $\ba \in \{-1,1\}^k$ since $\|\bw\|_1 = 1$. Thus, (\ref{eq:robust-MAJ}) follows since $\mu(\{-1,1\}^k \setminus P) \le \gamma$. 
\end{proof}

We now state our main technical result.

\begin{theorem}\label{thm:MAJ-analysis}
Let $P \subseteq Q \subseteq \{-1,1\}^k$ with $\MAJ \subseteq \Pol(P,Q)$. Let $(\bv_0, \bv_1, \hdots, \bv_k)$ be a $1-\gamma$ approximate vector assignment to $P$. For any $\bb \in \{-1,1\}^k \setminus Q$, the probability that the $\eps$-CMM algorithm when run on $(\bv_0, \hdots, \bv_k)$ outputs $\bb$ is at most
\begin{align}
600 k \frac{\max(\eps, \gamma)}{\sqrt{\eps}} \label{eq:MAJ-bound} 
\end{align}
\end{theorem}

From this, Theorem~\ref{thm:robust-MAJ} easily follows.

\begin{proof}[Proof of Theorem~\ref{thm:robust-MAJ}.]
Our Basic SDP solution $(\bv_0, \bv_1, \hdots, \bv_n)$ has the property that each clause $C_j$ is $1-\gamma_j$ approximately satisfied where $\E_j[\gamma_j] \le \eps$. 
Let $k_j$ denote the arity of the clause $C_j$, and define $k_{\max}=\max_j(k_j)$.
By Theorem~\ref{thm:MAJ-analysis}, the probability that $C_j$ is assigned some specific 
assignment which fails weakly satisfy $C_j$ 
is at most 
\begin{align*}
    600k_j\frac{\max(\eps,\gamma_j)}{\sqrt{\eps}}.
\end{align*}
Therefore, by the union bound, the probability that $C_j$ is weakly satisfied is at least
\[
    1 - 2^{k_j}\left(600 k_j \frac{\max(\eps, \gamma_j)}{\sqrt{\eps}}\right) \ge 1 - 2^{k_j}\left(600 k_j \frac{\eps + \gamma_j}{\sqrt{\eps}}\right) 
    \ge 1 - 2^{k_{\max}}\left(600 k_{\max} \frac{\eps + \gamma_j}{\sqrt{\eps}}\right)
    .
\]
Thus, the expected fraction of clauses that are weakly satisfied is at least
\[
        1 - 1200\cdot 2^{k_{\max}} k_{\max} \sqrt{\eps},
\]
as desired. %
\end{proof}

\subsection{Proof of Theorem~\ref{thm:MAJ-analysis}}

 We seek to bound the probability that 
 \begin{align}
 b_i \langle \bv_i, \br + \bv_0/\sqrt{\eps}\rangle \ge 0, \quad \forall i \in [k]. \label{eq:MAJ-bad}
\end{align}

Let $\bw$ correspond to $\bb$ as in Lemma~\ref{lem:MAJ-pol}. Let $\bu := w_1 \bv_1 + \cdots + w_k \bv_k$. %
By Lemma~\ref{lem:robust-MAJ-pol}, we have that $\langle \bu, \bv_0\rangle \ge -\gamma$. Furthermore, since $b_iw_i \le 0$ for all $i \in [k]$, we have that (\ref{eq:MAJ-bad}) holding implies that
\begin{align}
 \langle w_i\bv_i, \br + \bv_0/\sqrt{\eps}\rangle \le 0, \quad \forall i \in [k]. \label{eq:MAJ-bad-2a}
\end{align}
Summing over all $i \in [k]$, this implies that.
\begin{align}
\langle \bu, \br + \bv_0/\sqrt{\eps}\rangle \le 0.\label{eq:MAJ-bad-2b}
\end{align}

To work toward this, we first show that $\langle \bu, \bu\rangle$ cannot stray too much from $\langle \bu, \bv_0\rangle$.

\begin{lemma}\label{lem:bound-u}
$\langle \bu, \bu\rangle \le \langle \bu, \bv_0\rangle + 2\gamma$.
\end{lemma}
\begin{proof}
Let $\mu$ be the probability distribution over $\{-1,1\}^k$ with $\mu(P) \ge 1-\gamma$ that corresponds to $(\bv_0, \hdots, \bv_k)$. Recall from the proof of Lemma~\ref{lem:robust-MAJ-pol} that
\[
    \langle \bu, \bv_0\rangle = \sum_{\ba \in \{-1,1\}^k} \langle \ba, \bw\rangle \mu(\ba).
\]
Likewise, we can compute that
\begin{align*}
    \langle \bu, \bu\rangle &= \sum_{i=1}^k\sum_{j=1}^k w_iw_j \langle \bv_i, \bv_j\rangle\\
    &= \sum_{i=1}^k\sum_{j=1}^k w_iw_j \sum_{\ba \in \{-1,1\}^k} a_ia_j\mu(\ba)\\
    &= \sum_{\ba \in \{-1,1\}^k} \langle \ba, \bw\rangle^2 \mu(\ba).
\end{align*}
Thus, 
\[
    \langle \bu, \bu\rangle - \langle \bu, \bv_0\rangle = \sum_{\ba \in \{-1,1\}^k} \langle \ba, \bw\rangle (\langle \ba, \bw\rangle - 1) \mu(\ba).
\]
Recall that $\|\bw\|_1=1$.
For any $\ba \in P$, we have that $\langle \ba, \bw\rangle \in [0,1]$, so $\langle \ba, \bw\rangle (\langle \ba, \bw\rangle - 1) \le 0$. Furthermore, for any $\ba\in\{-1,1\}^k$, $\langle \ba, \bw\rangle \in [-1,1]$, which implies that $\langle \ba, \bw\rangle (\langle \ba, \bw\rangle - 1) \le 2$. Thus, $\langle \bu, \bu\rangle - \langle \bu, \bv_0\rangle \le 2\gamma$.
\end{proof}

For each $i \in [k]$, let $\bv_i^{=}$ be the component of $\bv_i$ parallel to $\bu$, and let $\bv_i^{\perp}$ be the component of $\bv_i$ perpendicular to $\bu$. That is,
\begin{align*}
    \bv_i^{=} &:= \frac{\langle \bv_i, \bu\rangle}{\langle \bu, \bu\rangle} \bu,\\
    \bv_i^{\perp} &:= \bv_i - \frac{\langle \bv_i, \bu\rangle}{\langle \bu, \bu\rangle} \bu.
\end{align*}

Let $\eta := \max(\gamma, \epsilon)$. We now split into cases.

\paragraph{Case 1, $\|\bu\|^2 \ge 8 \eta \ln (1/\eps)$.} In that case, we have that $\|\bu\|^2 \ge 4\gamma$. Therefore, by Lemma~\ref{lem:bound-u},
\[
    \langle \bu, \frac{1}{\sqrt{\eps}}\bv_0 \rangle \ge \frac{1}{\sqrt{\eps}}(\|\bu\|^2 - 2\gamma)  \ge \frac{\|\bu\|^2}{2\sqrt{\eps}}.
\]
Therefore, (\ref{eq:MAJ-bad-2b}) holds only if
\[
    \langle \br, -\frac{\bu}{\|\bu\|}\rangle > \frac{\|\bu\|}{2\sqrt{\eps}}.
\]
Note that the LHS is normally distributed with mean 0 and variance 1. Therefore, by Proposition~\ref{prop:gauss}, we deduce that (\ref{eq:MAJ-bad}) holds with probability at most
\[
    \frac{1}{2}\exp\left(-\frac{1}{2} \left(\frac{\|\bu\|}{2\sqrt{\eps}}\right)^2 \right) \le \frac{1}{2}\exp\left(-\frac{\eta}{\eps}\log(1/\eps)\right) \le \frac{\eps}{2} \le \frac{\eta}{\sqrt{\eps}}.\]

\paragraph{Case 2, there exists $i \in [k]$ such that $\|w_i\bv_i^{\perp}\| \ge 1/(10k)$.}  

First, note that by Proposition~\ref{prop:gauss}, for all $\delta \geq 0$ we have that
\begin{align}
\Pr[\delta \leq \left|\langle \vec{r}, \bu \rangle\right| \leq 2\delta] = \Pr\left[\frac{\delta}{||\bu||} \leq \left|\langle \br, \frac{\bu}{\|\bu\|} \rangle\right| \leq \frac{2\delta}{||\bu||}\right] \leq
2\frac{\delta}{\|u\|\sqrt{2\pi}}e^{-\frac{{\delta}^2}{2\|\bu\|^2}}
\leq
\frac{\delta}{\|\bu\|}e^{-\frac{{\delta}^2}{2\|\bu\|^2}}.\label{eq:delta-prob}
\end{align}
Assume that (\ref{eq:MAJ-bad}) holds. Then, by (\ref{eq:MAJ-bad-2a}), (\ref{eq:MAJ-bad-2b}), and Lemma~\ref{lem:bound-u}, we have that
\begin{align}
0 \ge \left\langle w_i\bv_i, \br + \frac{\bv_0}{\sqrt{\epsilon}}\right\rangle \ge \langle \bu, \br\rangle + \frac{1}{\sqrt{\eps}}\langle \bu, \bv_0\rangle \ge \langle \bu, \br\rangle + \frac{1}{\sqrt{\eps}}(\|\bu\|^2 - 2\gamma) \ge \langle \bu, \br\rangle - \frac{2}{\sqrt{\eps}}\gamma.\label{eq:wivi}
\end{align}

Now consider the probability distribution of the random variable $\langle w_i\bv_i, \br + \frac{\bv_0}{\sqrt{\epsilon}} \rangle$ conditioned on the value of $\langle \bu,\br\rangle$. Observe that this is a Gaussian distribution with standard deviation $\|{w_i}\bv_i^{\perp}\| \geq \frac{1}{10k}$ (and arbitrary mean).
By Proposition~\ref{prop:gauss}, the probability that (\ref{eq:wivi}) happens is at most \[10k(|\langle \bu, \br\rangle| + \frac{2}{\sqrt{\eps}}\gamma) \le 10k\left(2\delta + \frac{2}{\sqrt{\eps}}\gamma\right).\] %

Thus, the probability that both $\delta \leq \left|\langle \vec{r}, \bu \rangle\right| \leq 2\delta$ and (\ref{eq:MAJ-bad}) hold is at most 
\[
\frac{\delta}{\|\bu\|}e^{-\frac{{\delta}^2}{2\|\bu\|^2}} \cdot 10k\left(2\delta + \frac{2}{\sqrt{\eps}}\gamma)\right) \le \frac{40k{\delta}e^{-\frac{{\delta}^2}{2||\bu||^2}}}{||\bu||}\max\left\{\delta, \frac{\gamma}{\sqrt{\epsilon}}\right\}.
\]
Call the RHS of the above expression $f(\delta)$. We now split our analysis %
into subcases based on the size of $\|\bu\|^2$.

\paragraph{Case 2a, $\|\bu\|^2 \le 4\eta$.}  We %
seek to bound $\sum_{i=-\infty}^{\infty} f(2^{i}\|\bu\|)$. If $i \le 1$, note that
\[
    f(2^i\|\bu\|) \le \frac{40k\cdot 2^i\|\bu\|}{\|\bu\|} \max\left\{2^i\|\bu\|, \frac{\gamma}{\sqrt{\epsilon}}\right\} \le 40k \cdot 2^i \cdot \frac{\eta}{\sqrt{\eps}}.
\]
Summing over all $i \le 1$, we obtain a contribution of at most $160k \eta / \sqrt{\eps}$. Otherwise, if $i \ge 2$, we have that
\[
    f(2^i\|\bu\|) \le 40 k \cdot 2^i \cdot e^{-2^{2i-1}} \cdot \max \left\{2^{i+1} \sqrt{\eta}, \frac{\eta}{\sqrt{\eps}}\right\}
    \leq 
    \frac{160k\eta}{\sqrt{\eps}}\cdot 2^{2i-1}e^{-2^{2i-1}}.
\]
Using that the function $x\mapsto x\cdot e^{-x}$ is nonincreasing in $[1,\infty)$, we find
\begin{align*}
\sum_{i=2}^{\infty}2^{2i-1}e^{-2^{2i-1}}
    \leq 
\sum_{i=2}^{\infty}ie^{-i}\leq \int_{0}^\infty x\cdot e^{-x}\,dx=1
\end{align*}
and, thus, 
\begin{align*}
    \sum_{i=2}^{\infty} f(2^{i}\|\bu\|)
    \leq
    \frac{160k\eta}{\sqrt{\eps}}.
\end{align*}

As a consequence, the total probability bound is at most $\frac{320k\eta}{\sqrt{\eps}}$.

\paragraph{Case 2b, $\|\bu\|^2 > 4\eta$.} By logic similar to Case 1, we have that (\ref{eq:MAJ-bad}) holds only if $\langle \br, -\frac{\bu}{\|\bu\|}\rangle > \frac{\|\bu\|}{2\sqrt{\eps}}.$ Thus, we only need to consider $\delta \ge \frac{\|\bu\|^2}{2\sqrt{\eps}}$. Consider any such $\delta = \lambda \|\bu\| \ge\|\bu\|^2 / (2\sqrt{\eps})$ so $2\lambda \sqrt{\eps} \ge \|\bu\|$. We have that
\[
    f(\lambda \|\bu\|) = 40k\lambda e^{-\frac{\lambda^2}{2}} \max{\left\{\lambda \|\bu\|,\frac{\gamma}{\sqrt{\epsilon}}\right\}} \le 160 k \lambda^3 e^{-\frac{\lambda^2}{2}} \sqrt{\eps}.
\]
Since $2\lambda \sqrt{\eps} \ge \|\bu\| > \sqrt{4\eta}$, we have that $\lambda > 1$. Thus, the relevant probability of (\ref{eq:MAJ-bad}) occurring is at most
\[
    \sum_{i=0}^{\infty} 160 k \cdot 2^{3i} \cdot e^{-2^{2i-1}} \sqrt{\eps} \le 600 k \sqrt{\eps} \le 600 k \eta / \sqrt{\eps},
\]
as desired. %

\paragraph{Case 3, $\|\bu\|^2 < 8\eta\ln\left(\frac{1}{\epsilon}\right)$ and there is no $i$ such that $\|{w_i}\bv_i^{\perp}\| \geq \frac{1}{10k}$.} In this case, since $\bu = \sum_{i=1}^k w_i\bv_i$ and $\sum_{i=1}^{k}{\|{w_i}{\bv_i}\|} = \sum_{i=1}^{k}{|{w_i}|\|{\bv_i}\|} =\sum_{i=1}^{k}{|{w_i}|} = 1$, we have that 
\[1 = \sum_{i=1}^{k}{||{w_i}{\bv_i}||} \leq \sum_{i=1}^{k}{||{w_i}{\bv^{=}_i}||} + \sum_{i=1}^{k}{||{w_i}{\bv^{\perp}_i}||} \leq \sum_{i=1}^{k}{||{w_i}{\bv^{=}_i}||} + \frac{1}{10},\]
and
\[\left\|\sum_{i=1}^{k}{{w_i}{\bv^{=}_i}}\right\|^2 + \left\|\sum_{i=1}^{k}{{w_i}{\bv^{\perp}_i}}\right\|^2= \|\bu\|^2 \leq 8{\eta}\ln\left(\frac{1}{\epsilon}\right).\]
Putting these observations together, there must be an $i$ such that ${w_i}{\bv_i} = c\frac{\bu}{\|\bu\|} + {w_i}{\bv^{\perp}_i}$ where $|c| \geq \frac{1}{2k}$. We now define $t,x,y$ as follows
\begin{align*}
    t &:= -\frac{\langle{w_i}{\bv_i},\bv_0\rangle}{c\sqrt{\epsilon}},\\
    x &:= \langle \frac{\bu}{||\bu||}, \br \rangle,\\
    y &:= -\sum_{i}{\langle {w_i}\bv_i, \br + \bv_0/\sqrt{\eps} \rangle}.
\end{align*}
Note that $t$ is a constant, but $x$ and $y$ are random variables depending on $\br$. In particular, $x$ is a Gaussian random variable with mean zero and variable one. Furthermore, note that
\[
    \langle {w_i}\bv_i, \br + \bv_0/\sqrt{\eps} \rangle = c(x - t) + \langle {w_i}{\bv_i}^{\perp}, \br \rangle.
\]
We next observe the following. First, in order for (\ref{eq:MAJ-bad}) to hold, we must have that $-y \leq \langle {w_i}\bv_i, \br + \bv_0/\sqrt{\eps} \rangle \leq 0$ which implies that 
\begin{align}
|\langle {w_i}{\bv_i}^{\perp}, \br \rangle + c(x - t)| \leq y.\label{eq:cond-y}
\end{align}
Second, since $-y = \sum_{i}{\langle {w_i}\bv_i, \br + \bv_0/\sqrt{\eps} \rangle} = \langle \bu, \vec{r}\rangle + \frac{1}{\sqrt{\eps}}\langle \bu, \bv_0\rangle \geq \langle \bu, \vec{r}\rangle + \frac{1}{\sqrt{\eps}}(||\bu||^2 - 2\gamma)=||\bu||x + \frac{1}{\sqrt{\eps}}(||\bu||^2 - 2\gamma)$, we have
\begin{align}
|x| \geq -x \geq \frac{||\bu||}{\sqrt{\epsilon}} - \frac{2\gamma}{||\bu||\sqrt{\eps}} + \frac{y}{||\bu||}.\label{eq:cond-x}
\end{align}

To finish the proof, we split into cases based on the values of $(y, c(x-t)).$

\paragraph{Case 3a, $y \le 4 \frac{\eta}{\sqrt{\eps}}$.} 
Since $|c| \ge \frac{1}{2k}$, but $\|w_i \bv_i^{\perp}\| \le \frac{1}{10k}$, we have that  $\langle {w_i}\bv_i, \br + \bv_0/\sqrt{\eps} \rangle = c(x - t) + \langle {w_i}{\bv_i}^{\perp}, \br \rangle$ has standard deviation at least $|c|/2\geq\frac{1}{4k}$. Thus, by Proposition~\ref{prop:gauss}
\[
    \Pr\left[|\langle w_i\bv_i^{\perp}, \br\rangle + c(x-t)| \le 4\frac{\eta}{\sqrt{\eps}}\right] \le \frac{4\eta}{|c|\sqrt{\eps}/2} \le 16k \frac{\eta}{\sqrt{\eps}}.
\]
That is, if we condition on $y \le \frac{4\eta}{\sqrt{\eps}}$, the probability that (\ref{eq:cond-y}) holds (and thus the probability that (\ref{eq:MAJ-bad}) holds) is at most $16k \frac{\eta}{\sqrt{\eps}}$. 
\paragraph{Case 3b, there is $\delta \ge 4\eta/\sqrt{\eps}$ such that $y \ge \delta$ and $|c(x-t)| \le 4\delta$.}
Note that $\delta/2 \ge 2\gamma/\sqrt{\eps}$, so by (\ref{eq:cond-x}), we have that
\begin{align}
    |x| \ge \frac{||\bu||}{\sqrt{\epsilon}} + \frac{\delta}{2||\bu||}.\label{eq:cond-x'}
\end{align}
Thus, by the Proposition~\ref{prop:gauss}, we have that
\begin{align*}
    \Pr[|c(x-t)| \le 4\delta \wedge y \ge \delta] &\le \Pr[|c(x-t)| \le 4\delta \wedge (\ref{eq:cond-x'})]\\ &\le \frac{8\delta}{|c|}e^{-\frac{1}{2}\left(\frac{\|\bu\|}{\sqrt{\epsilon}} + \frac{\delta}{2\|\bu\|}\right)^2} \leq \frac{8\delta}{|c|}e^{-\frac{\delta}{\sqrt{\epsilon}}} \leq 16k{\delta}e^{-\frac{\delta}{\sqrt{\epsilon}}}.
\end{align*}
\paragraph{Case 3c, there is $\delta' \ge 4\delta \ge 16\eta/\sqrt{\eps}$ such that $y \in [\delta, 2\delta]$ and $|c(x-t)| \in [\delta',2\delta']$.} Note that (\ref{eq:cond-x'}) still holds, so we have that
\begin{align*}
    \Pr[|c(x-t)| \in [\delta',2\delta'] \wedge y \in [\delta, 2\delta]] &\le \Pr[|c(x-t)|  \in [\delta',2\delta']  \wedge (\ref{eq:cond-x'})]\\ &\le \frac{2\delta'}{|c|}e^{-\frac{1}{2}\left(\frac{\|\bu\|}{\sqrt{\epsilon}} + \frac{\delta}{2\|\bu\|}\right)^2} \leq \frac{2\delta'}{|c|}e^{-\frac{\delta}{\sqrt{\epsilon}}}.
\end{align*}
Let $d = \|w_i \bv_i^{\perp}\|$. We then have by Proposition~\ref{prop:gauss} that
\begin{align*}
    \Pr[(\ref{eq:MAJ-bad}) \mid |c(x-t)| \in [\delta',2\delta'] \wedge y \in [\delta, 2\delta]] &\le \Pr[(\ref{eq:cond-y}) \mid |c(x-t)| \in [\delta',2\delta'] \wedge y \in [\delta, 2\delta]]\\
    &\le \frac{4\delta}{d} e^{-\frac{\delta'^2}{8d^2}},
\end{align*}
where the last inequality follows from the fact that $\langle w_i \bv_i^{\perp}, \br\rangle$ must land in an interval of length at most $2\delta$ with minimum values of $\delta' - 2\delta \ge \frac{\delta'}{2}$. Therefore, combining the previous two calculations, we have that
\[
\Pr[(\ref{eq:MAJ-bad}) \wedge |c(x-t)| \in [\delta',2\delta'] \wedge y \in [\delta, 2\delta]]  \le \frac{8\delta\delta'}{|c|d}e^{-\frac{\delta}{\sqrt{\epsilon}} - \frac{\delta'^2}{8d^2}}.
\]

\paragraph{Finishing Case 3.} We now put these subcases together. First, condition on $\delta \ge 4\eta/\sqrt{\eps}$. We have by cases 3b and 3c that
\begin{align*}
    \Pr[(\ref{eq:MAJ-bad}) \wedge y \in [\delta, 2\delta]] &\le \Pr[|c(x-t)| \in [\delta',2\delta'] \wedge y \in [\delta, 2\delta]]\\&+ \sum_{i = 2}^{\infty}\Pr[(\ref{eq:MAJ-bad}) \mid |c(x-t)| \in [2^i\delta,2^{i+1}\delta] \wedge y \in [\delta, 2\delta]]\\
    &\le 6k{\delta}e^{-\frac{\delta}{\sqrt{\epsilon}}} + \sum_{i=2}^{\infty}\frac{2^{i+3} \delta^2}{|c|d}e^{-\frac{\delta}{\sqrt{\epsilon}} - \frac{2^{2i}\delta^2}{8d^2}}\\
    &= 6k{\delta}e^{-\frac{\delta}{\sqrt{\epsilon}}} + 16k \cdot \delta e^{-\delta/\sqrt{\eps}}\sum_{i=2}^{\infty} \frac{2^i\delta}{d} e^{-\frac{2^{2i}\delta^2}{8d^2}}. 
\end{align*}
Now,
\[
\sum_{i=2}^{\infty} \frac{2^i\delta}{d} e^{-\frac{2^{2i}\delta^2}{8d^2}} \le \sum_{i=2}^{\infty} 2\int_{2^{i-1}}^{2^{i}} \frac{\delta}{d}e^{-\frac{s^2\delta^2}{8d^2}}\,ds \le \int_{-\infty}^{\infty} \frac{\delta}{d}e^{-\frac{s^2\delta^2}{8d^2}} \,ds= \int_{-\infty}^{\infty} e^{-\frac{s^2}{8}} \,ds= 2\sqrt{2\pi}\le \frac{11}{2},
\]
so
\[
\Pr[(\ref{eq:MAJ-bad}) \wedge y \in [\delta, 2\delta]] \le 100k \delta e^{-\frac{\delta}{\sqrt{\eps}}}.
\]
We next incorporate the bound from case 3a to get that
\begin{align*}
    \Pr[(\ref{eq:MAJ-bad})] &\le \Pr[(\ref{eq:MAJ-bad}) \wedge y \le 4\eta/\sqrt{\eps}] + \sum_{i=2}^{\infty} \Pr[(\ref{eq:MAJ-bad}) \wedge y \in [2^i \eta/\sqrt{\eps}, 2^{i+1} \eta / \sqrt{\eps}]]\\
    &\le 16k \frac{\eta}{\sqrt{\eps}} + 100 k \cdot \frac{\eta}{\sqrt{\eps}}\sum_{i=2}^{\infty} 2^ie^{-\frac{2^i\eta}{\eps}}\\
    &\le 16k \frac{\eta}{\sqrt{\eps}} + 100 k \cdot \frac{\eta}{\sqrt{\eps}}\sum_{i=2}^{\infty} 2^ie^{-2^i}\\
    &\le 50 k \cdot \frac{\eta}{\sqrt{\eps}},
\end{align*}
as desired. This completes the proof of Theorem~\ref{thm:MAJ-analysis} and thus also of Theorem~\ref{thm:robust-MAJ}.

\section{Rounding schemes for separable PCSPs}
\label{sec_separablesPCSPs}
The goal of this section is to extend the results of Section~\ref{sec_majority} to the case of promise CSPs on non-Boolean domain. Throughout this section, we focus on promise CSPs such that the domains of the weak and strong constraint languages coincide.

We begin with two definitions.
Given a tuple $\ba\in A^k$, we denote by $\Pi_\ba\in\R^{k\times A}$ the one-hot encoding of $\ba$; i.e., the $(i,a)$-th entry of the matrix $\Pi_\ba$ is $1$ if $a_i=a$, $0$ otherwise. Also, we let $\Frob{-}{-}$ %
denote the Frobenius inner product of matrices, given by $\Frob{M}{N}=\Tr(M^\top N)$.

\begin{definition}
\label{defn_separable_rho}
Fix a PCSP template $(\A,\B)$ and a function $\rho:\R^A\to A$. We say that $(\A,\B)$ is \textit{$\rho$-separable} if $A=B$ and, for each promise relation $(P,Q)$ of some arity $k$ and each $\bb\in A^k\setminus Q$, there exists a matrix $W\in\R^{k\times A}$ such that
    \begin{enumerate}
    \item $\|W\|_1=1$;
    \item $W\bone=\bzero$;
        \item $\Frob{W}{\Pi_\ba}\geq 0$ for each $\ba\in P$; and
        \item $\Frob{W}{M}\leq 0$ for each $M\in \rho^{-1}(\bb)$.
    \end{enumerate}
\end{definition}

In part 4. of the above definition, the expression $\rho^{-1}(\bb)$ indicates the set of $k\times A$ real matrices whose $i$-th row is in $\rho^{-1}(b_i)$ for each $i\in [k]$.

Intuitively, in~\Cref{defn_separable_rho}, $\rho$ should be thought of as an arbitrary rounding scheme, and the idea is to consider polymorphism families that are compatible with the scheme in the sense that they can be used to avoid the occurrence of large integrality gaps.
We will also assume that the rounding function $\rho$ is well-behaved in the following weak sense.

\begin{definition}
\label{defn_conservativity}
    Fix some parameter $\alpha>0$. We say that a function $\rho:\R^A\to A$ is \textit{$\alpha$-conservative} if the inequality
    \begin{align*}
        q_{\rho(\bq)}\geq\alpha\sum_{b\in A}q_b
    \end{align*}
    holds for each $\bq\in\R^A$.
\end{definition}

\begin{theorem}
\label{thm_robust_SDP_separable_PCSPs}
    Let $(\A,\B)$ with $A = B$ be a $\rho$-separable PCSP template for an $\alpha$-conservative function $\rho:\R^A\to A$. Then $\PCSP(\A,\B)$ is robustly solvable via SDP with loss ${O}_{\alpha,A}(\sqrt{\eps}\log(1/\eps))$. 
\end{theorem}

\begin{remark}
    We point out that  the conditions imposed on $\rho$ in \Cref{defn_separable_rho} alone are too broad to ensure robust solvability with the loss guarantee of~\Cref{thm_robust_SDP_separable_PCSPs}, as they are met by CSP templates such as Horn-SAT which (assuming UGC) lack a $\eps^{O(1)}$-robust rounding scheme. Hence, the $\alpha$-conservativity of $\rho$ is needed.
\end{remark}

\begin{proof}

Without loss of generality, we shall consider the case in which $\PCSP(\A,\B)$ contains a single promise relation $(P,Q)$ of some arity $k$. Let $(\bv_0, \{\bv_{1,a} : a \in A\}, \hdots, \{\bv_{k,a} : a \in A\})$ be a $(1-\gamma)$-approximate vector assignment to $P$ for some $\gamma>0$, as per~\Cref{subsec_basic_SDP}.
Fix $\epsilon>0$, and run the $\epsilon$-CMM algorithm on such vector assignment, using $\rho$ as the rounding function (see~\Cref{sec_CMM_algorithm}). More precisely, we apply the following procedure: 
\begin{enumerate}
    \item Sample a random Gaussian $\br\in \cN(0_{n},I_{n})$ (where $n$ is the dimension of the space where the SDP vectors live).
    \item Define $\bz=\br+\frac{\bv_0}{\sqrt{\eps}}$, and let $\bc^{(i)}\in\R^A$ be the vector whose $a$-th entry is $c^{(i)}_a=\ang{\bv_{i,a}}{\bz}$ for each $i\in[k]$. %
    \item %
    Let $b_i=\rho(\bc^{(i)})$ for each $i\in[k]$.
\end{enumerate}
Consider now the vector $\bb=(b_i)_{i\in[k]}\in A^k$, and suppose that $\bb\not\in Q$. Take a corresponding matrix $W\in\R^{k,A}$ as guaranteed from the fact that $(\A,\B)$ is $\rho$-separable, and define the vector 
    \begin{align*}
        \bu=\sum_{i\in[k]}\sum_{a\in A}w_{i,
        a}\bv_{i,a}.
    \end{align*}
Consider also the matrix $M\in\R^{k\times A}$ whose $(i,a)$-th entry is $c^{(i)}_a$. %
Observe that $M\in\rho^{-1}(\bb)$, so part $4.$ of Definition~\ref{defn_separable_rho} guarantees that
\begin{align}
\label{eqn_949_1505}
    \ang{\bu}{\bz}=\Frob{W}{M}\leq 0.
\end{align}
Observe now that, for each $\ba\in A^k$, it holds that
\begin{align*}
    \Frob{W}{\Pi_\ba}=\sum_{i\in[k]}w_{i,a_i}.
\end{align*}
Hence, letting $\mu$ denote a probability distribution corresponding to the vector assignment,  we find
\begin{align*}
    \ang{\bu}{\bu}
    &=
    \sum_{i,i'\in[k]}\sum_{a,a'\in A}w_{i,a}w_{i',a'}\ang{\bv_{i,a}}{\bv_{i',a'}}
    =
    \sum_{i,i'\in[k]}\sum_{a,a'\in A}w_{i,a}w_{i',a'}\sum_{\substack{\ba\in A^k\\a_i=a\\a_{i'}=a'}}\mu(\ba)\\
    &=\sum_{\ba\in A^k}\mu(\ba)\sum_{i,i'\in [k]}w_{i,a_i}w_{i',a_{i'}}
    =
    \sum_{\ba\in A^k}\mu(\ba)\Frob{W}{\Pi_\ba}^2.
\end{align*}
Similarly,
\begin{align}
\label{eqn_uv_0_mu}
\notag
    \ang{\bu}{\bv_0}
    &=
    \sum_{i\in[k]}\sum_{a\in A}w_{i,a}\ang{\bv_{i,a}}{\bv_0}
    =
    \sum_{i\in[k]}\sum_{a\in A}w_{i,a}\sum_{\substack{\ba\in A^k\\a_i=a}}\mu(\ba)
    =
    \sum_{\ba\in A^k}\mu(\ba)\sum_{i\in[k]}w_{i,a_i}\\
    &=
    \sum_{\ba\in A^k}\mu(\ba)\Frob{W}{\Pi_\ba}.
\end{align}
Now, part $1.$ of Definition~\ref{defn_separable_rho} implies that, for each $\ba\in A^k$,
\begin{align*}
    |\Frob{W}{\Pi_\ba}|=|\sum_{i\in[k]}w_{i,a_i}|\leq\sum_{i\in[k],a\in A}|w_{i,a}|=\|W\|_1=1,
\end{align*}
while part $3.$ implies that $0\leq \Frob{W}{\Pi_\ba}$ for each $\ba\in P$.
Hence, the quantity $\Frob{W}{\Pi_\ba}^2-\Frob{W}{\Pi_\ba}$ is at most $2$ for every $\ba\in A^k$, and it is nonpositive if $\ba\in P$.
It follows that
\begin{align}
\label{eqn_1454_1505}
    \ang{\bu}{\bu}-\ang{\bu}{\bv_0}
    =
    \sum_{\ba\in A^k}\mu(\ba)(\Frob{W}{\Pi_\ba}^2-\Frob{W}{\Pi_\ba})\leq 2\gamma
\end{align}
(where we have used that the vectors $\bv_0,\{\bv_{i,a}\}$ form a $1-\gamma$ approximate vector assignment to $P$).

Take $\delta = 2\max(\eps,\gamma)\log(1/\eps)$ and $\eta=\sqrt{2(\delta+2\gamma)\log(1/\eps)}+\frac{\delta+\gamma}{\sqrt{\eps}}$.
Suppose first that $\ang{\bu}{\bv_0}>\delta$. 
Note that the random variable $\ang{\bu}{\bz}=\ang{\bu}{\br+\frac{\bv_0}{\sqrt{\epsilon}}}$ is distributed as a Gaussian with mean $\mu=\frac{1}{\sqrt{\epsilon}}\ang{\bu}{\bv_0}$ and standard deviation $\sigma=\|\bu\|$. Applying part (b) of~\Cref{thm:gauss} and observing that $\gamma\leq\frac{\delta}{2}$, we deduce that the probability that~\eqref{eqn_949_1505} happens is upper bounded by
\begin{align*}
    \exp\left(-\frac{\ang{\bu}{\bv_0}^2}{2\epsilon\|\bu\|^2}\right)
    \leq 
    \exp\left(-\frac{\ang{\bu}{\bv_0}^2}{2\epsilon(\ang{\bu}{\bv_0}+2\gamma)}\right)
    \leq
    \exp\left(-\frac{\delta^2}{2\epsilon(\delta+2\gamma)}\right)\leq
    \exp\left(-\frac{\delta}{4\epsilon}\right)\le 
    \sqrt\epsilon.
\end{align*}
In the expression above,
the second inequality uses the fact that the function $x\mapsto \frac{x^2}{x+c}$ is increasing if $x,c\ge 0$, while the last inequality holds since
\begin{align*}
    -\frac{\delta}{4\eps}
    =-\frac{2\max(\eps,\gamma)\log(1/\eps)}{4\eps}
    =
    \frac{\max(\eps,\gamma)\log\eps}{2\eps}\leq\frac{\log\eps}{2}=
    \log(\sqrt\eps).
\end{align*}
Thus, the probability that $\bb\not\in Q$ is at most $\sqrt\eps$ in this case.

Suppose now that $\ang{\bu}{\bv_0}\leq\delta$. Using~\eqref{eqn_1454_1505}, we deduce that, in this case, $\ang{\bu}{\bu}\leq2\gamma+\delta$.
Recall that $\br$ is a random Gaussian sampled from the distribution $\cN(0_{n},I_{n})$, so $\ang{\bu}{\br}$ is distributed as a Gaussian with mean $0$ and standard deviation $\|\bu\|$.
Using part $(b)$ of~\Cref{thm:gauss}, we find that 
\begin{align}
    \Pr\left(|\ang{\bu}{\br}|\geq\sqrt{2(\delta+2\gamma)\log(1/\eps)}\right)
    \leq
    2\exp\left(-\frac{2(\delta+2\gamma)\log(1/\eps)}{2\|\bu\|^2}\right)
    \leq
    2\exp\left(-\log(1/\eps)\right)
    =
    2\eps.\label{eqn_2341}
\end{align}
Observe now that $\ang{\bu}{\bv_0}\geq-\gamma$. Indeed, we have
\begin{align*}
       \ang{\bu}{\bv_0}
       &=
       \sum_{\ba\in A^k}\mu(\ba)\Frob{W}{\Pi_\ba}
       \geq\sum_{\ba\in A^k\setminus P}\mu(\ba)\Frob{W}{\Pi_\ba}
       \geq-\left\lVert W\right\rVert_1\sum_{\ba\in A^k\setminus P}\mu(\ba)
       \geq- \gamma,
    \end{align*}
where the first equality is~\eqref{eqn_uv_0_mu}, while the first and last inequalities come from part 3. and part 1. of~\Cref{defn_separable_rho}, respectively.
In particular, this implies that $|\ang{\bu}{\bv_0}|\leq\max(\delta,\gamma)<\delta+\gamma$.
Therefore, with probability at least $1-2\eps$ by \eqref{eqn_2341} it holds that 
\begin{align}
\label{eqn_1742_3005}
    |\ang{\bu}{\bz}|
    =
    |\ang{\bu}{\br}+\frac{1}{\sqrt{\eps}}\ang{\bu}{\bv_0}|
    \leq
    |\ang{\bu}{\br}|+\frac{1}{\sqrt{\eps}}|\ang{\bu}{\bv_0}|
    \leq\sqrt{2(\delta+2\gamma)\log(1/\eps)}+\frac{\delta+\gamma}{\sqrt{\eps}}=\eta.
\end{align}

Consider now, for each $i\in[k]$, 
the sets $A_{i,+}=\{a\in A:w_{i,a}> 0\}$ and $A_{i,-}=\{a\in A:w_{i,a} \leq%
0\}$, and let $c_i=\sum_{a\in A}|w_{i,a}|$. 
Using part 2 of Definition~\ref{defn_separable_rho},
we find that $\sum_{a\in A}w_{i,a}=0$, so 
\begin{align*}
\sum_{a\in A_{i,+}}w_{i,a}=-\sum_{a\in A_{i,-}}w_{i,a}.
\end{align*}
Fix $i\in[k]$ and take some $\lambda>0$. Consider the $k\times A$ matrix $M^{(i,\lambda)}$ having $(i,b_i)$-th entry $1$, $(j,b_j)$-th entry $\lambda$ for each $j\neq i\in[k]$, and all other entries $0$. Recall that the function $\rho$ is assumed to be $\alpha$-conservative for some $\alpha>0$. It follows that $M^{(i,\lambda)}\in\rho^{-1}(\bb)$, so $\Frob{W}{M^{(i,\lambda)}}\leq 0$ by part 4. of Definition~\ref{defn_separable_rho}. Since this holds for each positive $\lambda$, we deduce that $w_{i,b_i}\leq 0$, so $b_i\in A_{i,-}$ for each $i\in [k]$. %
Let now $\beta=\min(\frac{\alpha}{2},\frac{1}{2|A|})$.
Consider the matrix $N^{(i,\lambda)}$ defined by $N^{(i,\lambda)}_{i,b_i}=1-\beta|A_{i,+}|$, $N^{(i,\lambda)}_{i,a}=\beta$ for each $a\in A_{i,+}$, $N^{(i,\lambda)}_{i,a}=0$ for each $a\in A_{i,-}\setminus\{b_i\}$, $N^{(i,\lambda)}_{j,b_j}=\lambda$ for each $j\neq i\in [k]$, and $N^{(i,\lambda)}_{j,a}=0$ elsewhere.
Using again the $\alpha$-conservativity of $\rho$, we find that $N^{(i,\lambda)}\in\rho^{-1}(\bb)$ and, thus, $\Frob{W}{N^{(i,\lambda)}}\leq 0$. Taking the limit as $\lambda\to 0$, this means that $\beta\sum_{a\in A_{i,+}}w_{i,a}\leq (\beta|A_{i,+}|-1)w_{i,b_i}$.
Therefore, for each $i\in[k]$, we have
\begin{align}
\label{ineq_c_i_wibi}\notag
    c_i&=\sum_{a\in A}|w_{i,a}|
    =
    \sum_{a\in A_{i,+}}w_{i,a}-\sum_{a\in A_{i,-}}w_{i,a}
    =
    2\sum_{a\in A_{i,+}}w_{i,a}
    \leq 2\frac{\beta|A_{i,+}|-1}{\beta}w_{i,b_i}\\
    &= 2\frac{1-\beta|A_{i,+}|}{\beta}|w_{i,b_i}|.
\end{align}
The $\alpha$-conservativity of $\rho$ also guarantees that
\begin{align}
\label{eqn_1407_1245_a}
    \ang{\bv_{i,b_i}}{\bz}
    =c^{(i)}_{b_i}\geq\alpha\sum_{a\in A}c^{(i)}_a=\alpha\sum_{a\in A}\ang{\bv_{i,a}}{\bz}
    =\alpha\ang{\bv_0}{\bz}.
\end{align}
Using~\Cref{prop_concentration_chi_squared}, we find that 
\begin{align*}
    \Pr(\|\br\|>\frac{1}{2\sqrt{\eps}})\leq O(\exp(-1/\sqrt{\eps})).
\end{align*}
Hence, from Cauchy--Schwarz, we deduce that 
\begin{align}
\label{eqn_1407_1245_b}
    \ang{\bv_0}{\bz}
    =
    \ang{\bv_0}{\br+\frac{\bv_0}{\sqrt{\eps}}}
    =
    \ang{\bv_0}{\br}+\frac{1}{\sqrt{\eps}}
    \geq 
    -|\ang{\bv_0}{\br}|+\frac{1}{\sqrt{\eps}}
    \geq 
    -\|\br\|+\frac{1}{\sqrt{\eps}}
    \geq \frac{1}{2\sqrt{\eps}}
\end{align}
with probability $1-O(\exp(-1/\sqrt{\eps}))$. Similarly, we have
\begin{align}
\label{eqn_1407_1245_c}
    \ang{\bv_{i,b_i}}{\bz}
    \leq\|\bv_{i,b_i}\|\cdot\|\bz\|
    =
    \|\bv_{i,b_i}\|\cdot\|\br+\frac{\bv_0}{\sqrt\eps}\|
    \leq 
    \|\bv_{i,b_i}\|\cdot(\|\br\|+\frac{1}{\sqrt\eps})\leq \frac{2}{\sqrt\eps}\|\bv_{i,b_i}\|
\end{align}
with probability $1-O(\exp(-1/\sqrt{\eps}))$.
Combining~\eqref{eqn_1407_1245_a},~\eqref{eqn_1407_1245_b}, and~\eqref{eqn_1407_1245_c}, we deduce that
\begin{align}
\label{eqn_1407_1253}
    \|\bv_{i,b_i}\|
    \geq \frac{\sqrt\eps}{2}
    \ang{\bv_{i,b_i}}{\bz}
    \geq
   \frac{\alpha\sqrt\eps}{2}
    \ang{\bv_{0}}{\bz}
    \geq 
    \frac{\alpha}{4}
\end{align}
with probability $1-O(\exp(-1/\sqrt{\eps}))$.

Let us now define $\by_i=\sum_{a\in A}w_{i,a}\bv_{i,a}$ for each $i\in[k]$.
Observe that $\beta|A_{i,+}|\leq\beta|A|\leq\frac{1}{2}$.
Combining~\eqref{ineq_c_i_wibi} with~\eqref{eqn_1407_1253}, we obtain that
\begin{align*}
    \|\by_i\|=\sqrt{\sum_{a\in A}w_{i,a}^2\|\bv_{i,a}\|^2}\geq |w_{i,b_i}|\|\bv_{i,\bb_i}\|\geq \frac{\beta c_i}{2(1-\beta|A_{i,+}|)}\cdot\frac{\alpha}{4}
    =
    \frac{\alpha\beta c_i}{8(1-\beta|A_{i,+}|)}
    \geq
    \frac{\alpha\beta c_i}{8}
\end{align*}
with probability $1-O(\exp(-1/\sqrt{\eps}))$.
Note now that
\begin{align*}
    \Pr\left(|\ang{\by_i}{\bz}|<\frac{8}{\alpha\beta}\eta\|\by_i\|\right)\leq\frac{\frac{16}{\alpha\beta}\eta\|\by_i\|}{2\|\by_i\|}=\frac{8}{\alpha\beta}\eta
\end{align*}
by part $(a)$ of~\Cref{thm:gauss}.
Using part $4.$ of Definition~\ref{defn_separable_rho} in a similar fashion as before, we find that $\ang{\by_i}{\bz}\leq 0$; thus,
\begin{align*}
    \Pr\left(\ang{\by_i}{\bz}<-\frac{8}{\alpha\beta}\eta\|\by_i\|\right)\geq 1-\frac{8}{\alpha\beta}\eta.
\end{align*}
As a consequence, by the union bound, the property \[\ang{\by_i}{\bz}<-\frac{8}{\alpha\beta}\eta\|\by_i\| \quad\mbox{for each}\quad i\in [k]
\]
holds with probability $1-O_{\alpha,A}(\eta)$. 
Observe now that
\begin{align*}
    \eta\geq \frac{\delta+\gamma}{\sqrt\eps}\geq \frac{\delta}{\sqrt\eps}
    \geq\frac{2\eps\log(1/\eps)}{\sqrt\eps}\geq\sqrt{\eps}\geq\exp(-1/\sqrt\eps)
\end{align*}
for $\epsilon$ small enough.
Hence,
putting it all together, it follows that, with probability $1-O_{\alpha,A}(\eta)$, 
\begin{align*}
    \ang{\bu}{\bz}=\sum_{i\in[k]}\ang{\by_i}{\bz}<-\frac{8}{\alpha\beta}\eta\sum_{i\in[k]}\|\by_i\|\leq
    -\eta\sum_{i\in[k]}c_i
    =-\eta\|W\|_1=-\eta,
\end{align*}
thus contradicting~\eqref{eqn_1742_3005}. It follows that, in this case, $\bb\in Q$ with probability $1-O_{\alpha,A}(\eta)$.

Finally, recall that $\eta=\sqrt{2(\delta+2\gamma)\log(1/\eps)}+\frac{\delta+\gamma}{\sqrt{\eps}}$, where $\delta = 2\max(\eps,\gamma)\log(1/\eps)$ and $\gamma$ is the error of our vector assignment for a single clause $(P, Q)$. Since our instance is assumed to be $(1-\eps)$-satisfiable, we know that the average value of $\gamma$ across all clauses is at most $\eps$, i.e., $\E[\gamma] \leq \eps$, where the expectation is taken over the clauses of the instance of $\PCSP(\bA,\bB)$. As such,
\begin{align*}
    \E[\delta] &\le \E[2(\eps + \gamma) \log(1/\eps)] = 2(\eps + \E[\gamma])\log(1/\eps) \leq 4\eps \log(1/\eps), \text{ and}\\
    \E[\eta] &\le \sqrt{2(\E[\delta]+2\E[\gamma])\log(1/\eps)}+\frac{\E[\delta]+\E[\gamma]}{\sqrt{\eps}}&\text{ (Jensen's inequality)}\\
    &\le \sqrt{2(4\eps\log(1/\eps)+2\eps)\log(1/\eps)}+\frac{4\eps\log(1/\eps)+\eps}{\sqrt{\eps}}\\
    &\le 10\sqrt{\eps}{\log(1/\eps)}.
\end{align*}
In summary, whether $\ang{\bu}{\bv_0}> \delta$ or $\ang{\bu}{\bv_0}\leq\delta$, the loss of SDP is at most $O_{\alpha,A}(\sqrt{\eps}\log(1/\eps))$, as required.
\end{proof}

It is easy to check that a Boolean PCSP admitting majority polymorphisms of all (odd) arities satisfies the hypotheses of Theorem~\ref{thm_robust_SDP_separable_PCSPs}.
(Note, however, that the loss in the theorem above is slightly worse---by a $\log(1/\eps)$ factor---than the ${O}(\sqrt{\eps})$ loss we managed to obtain in the Boolean majority case, see Theorem~\ref{thm:robust-MAJ}.)
More interestingly, Theorem~\ref{thm_robust_SDP_separable_PCSPs} can be applied to other classes of PCSPs, whose robust solvability was not known.  
This is done by relating the polymorphism families defining such PCSPs to corresponding rounding functions as per Definition~\ref{defn_separable_rho}.
In particular, we now show that
any PCSP containing \textit{plurality} polymorphisms of all arities satisfies the hypotheses of Theorem~\ref{thm_robust_SDP_separable_PCSPs} and, thus, is robustly solvable.

\begin{definition}[\cite{BrandtsWZ21}]
A function $f:A^n\to A$ is a \textit{plurality} if 
\begin{align*}
    f(\ba)=\argmax_{a\in A}\{\#\mbox{ of occurrences of $a$ in $\ba$}\}
\end{align*}
for all $\ba\in A^n$,
with ties broken in a way that $f$ is symmetric.
\end{definition}
\begin{proposition}\label{prop:plur}
    Let $(\A,\B)$ with $A = B$ be a PCSP template such that $\Pol(\A,\B)$ contains plurality polymorphisms of all arities. Then $\PCSP(\A,\B)$ is robustly solvable by SDP with loss ${O}(\sqrt{\eps}\log(1/\eps))$.
\end{proposition}
\begin{proof}
Consider any function $\rho:\R^A\to A$ satisfying $q_{\rho(\bq)}\geq q_a$ for each $\bq\in\R^A$, $a\in A$ (i.e., $\rho$ is the $\argmax$ function, with ties broken arbitrarily). Note that $\rho$ is $\alpha$-conservative for $\alpha=\frac{1}{|A|}$. We claim that $(\A,\B)$ is $\rho$-separable.

Fix a promise relation $(P,Q)$ of $(\A,\B)$ of arity $k$, and choose some tuple $\bb\in A^k\setminus Q$.
Let $K \subseteq \R^{k \times A}$ be the convex hull of the set $\{\Pi_\ba:\ba\in P\}$. Let also $L$ be the subset of $\R^{k \times A}$ for which $\bb$ is the vector of (strict) plurality coordinates. That is, $M\in L$ if and only if, for each $i \in [k]$ and each $a\in A\setminus\{b_i\}$, it holds that $M_{i,b_i} > M_{i,a}$. Note that $K$ is closed and $L$ is open, and they are both convex sets. We claim that they are disjoint.
Suppose by contradiction that $M\in K\cap L$. Since $M$ is constrained by a system of inequalities with rational coefficients, we may assume there exists a family $\{\lambda_\ba:\ba\in P\}$ of \emph{rational} nonnegative weights summing up to $1$ such that $M=\sum_{\ba\in P}\lambda_\ba\Pi_\ba$. %
Take a positive integer $n$ such that $n\lambda_\ba$ is an integer for each $\ba$, and choose a plurality polymorphism $f\in\Pol(\A,\B)$ of arity $n$. Let $\textbf{c}\in A^k$ be the tuple obtained by applying $f$ coordinatewise to a list of tuples $\ba^{(1)},\ba^{(2)},\dots,\ba^{(n)}\in P$, where each $\ba$ appears $n\lambda_\ba$-many times. Since $f$ is a polymorphism, $\bc\in Q$. On the other hand, $M\in L$ implies that $M_{i,b_i}> M_{i,a}$ for each $i\in[k]$ and each $b_i\neq a\in A$. By definition of plurality, we deduce that $\textbf{c}=\bb\not\in Q$, a contradiction. It follows that $K\cap L=\emptyset$, as claimed.

Using well-known results on hyperplane separation of convex bodies~\cite{boyd2004convex}, since $\R^{k \times A}$ is finite dimensional, there exists a linear operator $W \in \R^{k\times A}$ and a scalar $\mu \in \R$ such that for every $M \in K$ it holds that $\Frob{W}{M} \ge \mu$, and for every $M \in L$ it holds that $\Frob{W}{M} < \mu$.
We claim that $W$ is our desired matrix of weights as required by Definition~\ref{defn_separable_rho}. 

First, 
observe that, since $L$ is an open cone, we have that $\Frob{W}{M} < \mu$ for all $M \in L$ if and only if $\Frob{W}{\lambda M} < \mu$ for all $M \in L$ and $\lambda > 0$. By sending $\lambda \to \infty$, we have that $\Frob{W}{M} < 0$ for all $M \in L$ and the RHS of this inequality cannot be made any smaller. In other words, we can assume without loss of generality that $\mu = 0$. 
Then, we may also assume without loss of generality that the $1$-norm of $W$ is $1$, by simply rescaling it if necessary. Hence, $1.$ holds.
The condition
$\Frob{W}{M} \ge 0$ for all $M \in K$ is equivalent to $\Frob{W}{\Pi_\ba} \ge 0$ for all $\ba\in P$, so $3.$ holds.
To prove $4.$, just observe that $\rho^{-1}(\bb)$ is a subset of the topological closure of $L$. Hence, $4.$ follows from the continuity of the map $\bullet\mapsto\Frob{W}{\bullet}$.

In order to prove condition $2.$, choose $i\in[k]$. For each $\alpha,\beta\geq 0$, consider the following matrices:
\begin{itemize}
    \item $D^{(\alpha)}$ is the $k\times k$ diagonal matrix whose $i$-th diagonal entry is $1$ and whose $j$-th diagonal entry is $\alpha$ for each $j\neq i$;
    \item $J$ is the $k\times A$ all-one matrix;
    \item $M^{(\alpha,\beta)}= D^{(\alpha)}(J+\beta\Pi_\bb)$;
    \item $N^{(\alpha,\beta)}= D^{(\alpha)}(-J+\beta\Pi_\bb)$.
\end{itemize}

Observe that $\rho(M^{(\alpha,\beta)})=\rho(N^{(\alpha,\beta)})=\bb$ for each $\alpha>0$, $\beta>0$. Hence, by condition $4.$, it holds that $\Frob{W}{M^{(\alpha,\beta)}}\leq 0$ and $\Frob{W}{N^{(\alpha,\beta)}}\leq 0$. Using the continuity of the Frobenius inner product, we deduce that 
\begin{align*}
    0
    &\geq
    \Frob{W}{M^{(0,0)}}
    =\sum_{a\in A}w_{i,a}\quad\quad\mbox{and}\quad\quad
    0
    \geq
    \Frob{W}{N^{(0,0)}}
    =-\sum_{a\in A}w_{i,a},
\end{align*}
thus proving that $W\bone=\bzero$, as required.

In summary, we have shown that $\PCSP(\A,\B)$ is $\rho$-separable for a $\frac{1}{|A|}$-conservative function $\rho$. Hence, the result follows from Theorem~\ref{thm_robust_SDP_separable_PCSPs}.
\end{proof}

\section{Equality Preserves Robustness}\label{sec:eq-robust}

Given a domain $A$, we define $\EQ_A := \{(a,a) : a \in A\}$. Likewise, for a domain pair $(A,B)$, we define $\EQ_{A,B}$ to be the promise relation $(\EQ_A, \EQ_B)$. As an informal shorthand, given a promise template $(\bA, \bB)$, we let $(\bA, \bB)+\EQ$ denote the promise template with $\EQ_{A,B}$ added. 
Fix parameters $0 < \alpha < \beta < 1$. We say that $\PCSP(\A,\B)$ is $(\alpha,\beta)$-\textit{robust} if there exists a polynomial-time algorithm such that when the input is an instance on variable set $V$ such that there is a map $g:V\to A$ that strongly satisfies at least $1-\alpha$ fraction of the constraints, the algorithm outputs a map $g':V\to B$ that weakly satisfies at least $1-\beta$ fraction of the constraints.

The main goal of this section is to prove the following theorem.

\begin{restatable}{theorem}{robusteq}
\label{thm:robust-eq}
Assume UGC. If $\PCSP(\bA, \bB)$ is $(\alpha, \beta)$-robust %
for some $0 < \alpha < \beta < 1$, then $\PCSP((\bA, \bB)+\EQ)$ is $(\alpha_{\EQ}, \beta_{\EQ})$-robust for
\begin{align}
    \alpha_{\EQ} &= \Omega_{\bA}(\alpha^2\beta^4),\label{eq:alpha-eq}\\
    \beta_{\EQ} &= O_{\bB}(\beta).\label{eq:beta-eq}
\end{align}
\end{restatable}

The assumption of UGC is required to tie the existence of robust algorithms to integrality ratios of the basic SDP, see Theorem~\ref{thm:BCR}. One could make Theorem~\ref{thm:robust-eq} unconditional by redefining robustness in terms of basic SDP integrality gaps, but we omit such a reformulation in our presentation.

It directly follows from the result above that the existence of a \textit{minion homomorphism}\footnote{For the definition of minions and their homomorphisms, we refer the reader to~\cite{BBKO21:jacm}.} $\Pol(\bA, \bB) \to \Pol(\bC, \bD)$ between the polymorphism sets of two PCSP templates preserves robust solvability under UGC.

\begin{corollary}\label{cor:gadget-robust}
Consider a function $f:(0,1)\to (0,1)$ with $f(\eps)\to 0$ as $\eps\to 0$. Under UGC, if $\Pol(\bA, \bB) \to \Pol(\bC, \bD)$ and $(\bA, \bB)$ is $f$-robust,  then $(\bC, \bD)$ is $g$-robust, where $g(\eps) = O_{\bA,\bB,\bC,\bD}(f(\eps^{1/6}))$.
\end{corollary}

\begin{proof}
We start by introducing some terminology from algebraic PCSP theory, following~\cite{BG21:sicomp,BBKO21:jacm}.
Fix a relational signature $\sigma$ (i.e., a finite set of relation symbols $R_1,R_2,\dots$, each with an integer arity $r_1,r_2,\dots$). A \textit{pp-formula} over $\sigma$ is a formal expression $\psi$ consisting of an existentially quantified conjunction of predicates of the form $(i)$ ``$x=y$'', or $(ii)$ ``$(x_{i_1},\dots,x_{i_r})\in R$'' for some $R\in\sigma$ of arity $r$, where $x,y,x_{i_1},\dots,x_{i_r}$ are variables. 
Let $k$ be the number of free (i.e., unquantified) variables in $\psi$.
Given a $\sigma$-structure $\A$, the interpretation of $\psi$ in $\A$ is the set $\psi(\A)\subseteq A^k$ containing all tuples $(a_1,\dots,a_k)\in A^k$ that  satisfy $\psi$, where each symbol $R$ appearing in $\psi$ is interpreted in $\A$. Let now $\A'$ be a $\sigma'$-structure for some possibly different signature $\sigma'$, such that $A'=A$. We say that $\A'$ is \textit{pp-definable} from $\A$ if for 
each symbol $S\in\sigma'$ it holds that $S^{\A'}=\psi_S(\A)$ for some pp-formula $\psi_S$ over $\sigma$. Suppose now that $A'=A^n$ for some $n\in\N$, and
let $\operatorname{vec}_n(\A')$ be the structure with domain $A$ and relations defined as follows: For each $r$-ary symbol $S\in\sigma'$, $\operatorname{vec}_n(\A')$ has an $rn$-ary relation containing the tuple $(b^{(1)}_1,\dots,b^{(1)}_n,\dots,b^{(r)}_1,\dots,b^{(r)}_n)$ for each tuple $(\bb^{(1)},\dots,\bb^{(r)})\in S^{\A'}$.
We say that $\A'$ is an $n$-fold \textit{pp-power} of $\A$ if $\operatorname{vec}_n(\A')$ is pp-definable from $\A$.

Take now $\bA$, $\bB$, $\bC$, and $\bD$ as in the statement of the corollary, let $\sigma$ be the signature of $\bA,\bB$, and let $\sigma'$ be the signature of $\bC,\bD$.
It was shown in~\cite[Theorem~4.12]{BBKO21:jacm} that the existence of a minion homomorphism $\Pol(\bA, \bB) \to \Pol(\bC, \bD)$ is equivalent to the existence of two $\sigma'$-structures $\tilde\bA$, $\tilde\bB$ such that $(i)$ $\bC\to\tilde\bA\to\tilde\bB\to\bD$, and $(ii)$ $\tilde\bA$ and $\tilde\bB$ are pp-powers of $\A$ and $\B$, respectively, for the same power $n\in\N$ and via the same pp-formulae.

Recall that we are assuming that $\PCSP(\bA,\bB)$ is $f$-robust. Using~\Cref{thm:robust-eq} we have that $\PCSP((\bA, \bB)+\EQ)$ is $g$-robust for $g(\epsilon)=f(\epsilon^{1/6})$.
Using~\cite[Proposition 2.10]{BGS23:stoc}, we deduce that $\PCSP(\operatorname{vec}_n(\tilde\bA), \operatorname{vec}_n(\tilde\bB))$ is $O_{\A,\B,\tilde\A,\tilde\B}(1)\cdot g$-robust.
Consider now the (standard) reduction from $\PCSP(\tilde\bA,\tilde\bB)$ to $\PCSP(\operatorname{vec}_n(\tilde\bA), \operatorname{vec}_n(\tilde\bB))$ defined as follows. For an instance $\X$ of $\PCSP(\tilde\bA,\tilde\bB)$, output an instance $\hat\X$ of $\PCSP(\operatorname{vec}_n(\tilde\bA), \operatorname{vec}_n(\tilde\bB))$ whose domain contains $n$ copies $x^{(1)},\dots,x^{(n)}$ of each $x\in X$, and whose relations $R^{\hat\X}$ contain the tuple $(x^{(1)}_1,\dots,x^{(n)}_1,\dots,x^{(1)}_r,\dots,x^{(n)}_r)$ for each tuple $(x_1,\dots,x_r)$ in $R^\X$. Clearly, such reduction preserves robust solvability with the same loss, so we deduce that $\PCSP(\tilde\bA,\tilde\bB)$ is $O_{\A,\B,\tilde\A,\tilde\B}(1)\cdot g$-robust. Finally, since $(\bC,\bD)$ is a homomorphic relaxation of $(\tilde\bA,\tilde\bB)$, we conclude that the same holds for $\PCSP(\bC,\bD)$, as required.
\end{proof}

As discussed in the introduction, \Cref{cor:gadget-robust} partially confirms \Cref{ques:eq} due to Barto and Kozik~\cite{BK16:sicomp}.
The proof of \Cref{thm:robust-eq} is rather intricate as it involves substantial modifications to the algorithm of Brown-Cohen and Raghavendra~\cite{BR16:correlation}, which itself is a modified version of Raghavendra's algorithm~\cite{Raghavendra08:everycsp}. We begin this section with an overview of Brown-Cohen and Raghavendra's algorithm (henceforth called the \emph{BCR algorithm}). %

\subsection{An Overview of the Algorithm of Brown-Cohen and Raghavendra}\label{subsec:BCR}

In this section, we describe the BCR algorithm and its main correctness guarantee. We begin with the necessary preliminaries.

\subsubsection{BCR Preliminaries}
A key ingredient in the BCR algorithm is a notion of \emph{approximate polymorphisms}. For a promise template $(\bA, \bB)$ and a positive integer $R$, consider a probability distribution $\cP$ of maps $p : A^R \to B$. 

\begin{definition}[Approximate Polymorphism~\cite{BR16:correlation}]
We say that $\cP$ is a $(c,s)$-approximate polymorphism of $(\bA, \bB)$ if for all (normalized) weighted instances $\cX$ %
of $\PCSP(\bA, \bB)$ on variable set $X$, given $R$ assignments $a_1, \hdots, a_R : X \to A$ each with value at least $c$, then
\[
    \E_{p \sim \cP}[\val_{\cX}(p \circ (a_1, \hdots, a_R))] \ge s,
\]
where for all $x \in X$, we have that $(p \circ (a_1, \hdots, a_R))(x) = p(a_1(x), \hdots, a_R(x))$.
\end{definition}

We also define $\cP$ to be an \textit{$(\alpha, \beta)$-robust polymorphism} if $\cP$ is a $(1-\alpha, 1-\beta)$-approximate polymorphism. Observe that there are many trivial approximate polymorphisms with good parameters, such as probability distributions supported only on dictators. To identify the ``meaningful'' approximate polymorphisms, BCR use a notion of \emph{quasirandomness} to ensure that the robust/approximate polymorphism is sufficiently nontrivial.

We start by recalling some standard Fourier-analytic notions. Fix a probability distribution $\mu\in\Delta_A$, and let $\{\chi_0=1,\dots,\chi_{|A|-1}\}$ be an orthonormal basis for the vector space $L_2(A,\mu)$. Any function $f:A^R\to\R$ can be written as $f=\sum_{\sigma\in\N^R}\hat{f_{\sigma}}\chi_\sigma$, where $\chi_\sigma(x)=\prod_{j=1}^R\chi_{\sigma_j}(x_j)$, and the $\hat f_\sigma$ are the Fourier coefficients of $f$.
Let $|\sigma|$ denote the number of non-zero coordinates of $\sigma$.
The \textit{degree $d$ influence} of the $i$-th coordinate of $f$ (under the probability measure $\mu$) is defined as
\begin{align*}
    \operatorname{Inf}_{i,\mu}^{<d}(f)
    =
    \sum_{\substack{\sigma\in\N^R\\\sigma_i\neq 0\\|\sigma|<d}}\hat f_\sigma^2.
\end{align*}

We say that an approximate polymorphism $\cP$ is \textit{$(\tau,d)$-quasirandom} if, for every probability measure $\mu\in\Delta_A$, it holds that
\begin{align*}
    \mathbb{E}_{p\sim\cP}\left[\max_{i\in[R]}\operatorname{Inf}_{i,\mu}^{<d}(p)\right]\leq\tau.
\end{align*}

In order to describe the BCR algorithm, we shall also need to use \textit{noise operators}. Given a parameter $\rho\in[0,1]$, we let $T_\rho$ be the operator defined as follows: 
for each function $f:A^R\to\R$,
\begin{align*}
    T_\rho f(x)=\E_{y\sim_\rho x}[f(y)],
\end{align*}
where the expression ``$y\sim_\rho x$'' denotes that, for each $i\in [R]$, $y_i$ equals $x_i$ with probability $\rho$, and $y_i$ is independently sampled according to $\mu$ with probability $1-\rho$. It follows from this definition that the Fourier expansion of the operator is given by $T_\rho f=\sum_{\sigma\in\N^R}\hat f_\sigma \rho^{|\sigma|}\chi_\sigma$.

\subsubsection{The BCR Algorithm}

The main result of Brown-Cohen and Raghavendra~\cite{BR16:correlation} is as follows.

\begin{theorem}\label{thm:BCR}
For any PCSP template $(\A,\B)$ and any $\theta>0$, there exists $\tau>0$ and $d\in\N$ such that, for each $0<\alpha<\beta<1$, if $(\bA, \bB)$ admits an $(\alpha, \beta)$-robust $(\tau,d)$-quasirandom polymorphism, then the BCR algorithm is $(\alpha-\theta,\beta+\theta)$-robust for $\PCSP(\bA, \bB)$.
Conversely, assuming UGC, if $\PCSP(\A,\B)$ has any $(\alpha,\beta)$-robust algorithm, then for any $\theta > 0$, $\PCSP(\bA, \bB)$ admits an $(\alpha-\theta, \beta+\theta)$-robust algorithm via the aforementioned algorithmic procedure. %
\end{theorem}

We now present the BCR algorithm. We let $\Psi : \mathbb R^A \to \Delta_A$ be a Lipschitz-continuous function extending the identity function on $\Delta_A$.

\begin{itemize}
\item Let $\cP$ be an $R$-arity $(d, \tau)$-quasirandom $(\alpha,\beta)$-robust polymorphism. Pick $\eta > 0$ suitably small.
\item Let $\cX$ be a weighted instance of $\PCSP(\bA, \bB)$ on variable set $X$. 
\item Solve the basic SDP for $\cX$ with respect to $\CSP(\bA)$ to get a unit vector $\bv_0 \in \R^N$ and vectors $\bv_{x,a} \in \mathbb R^N$ for all $x \in X$ and $a \in A$ with completeness at least $1-\alpha$ (see \Cref{subsec:SDP-rounding}). %
\item Sample $p \sim \cP$ and independently sample $R$ vectors $\zeta^{(1)}, \hdots, \zeta^{(R)} \sim \mathcal N(0^N, I_N)$. 
\item For all $x \in X$, $a \in A$, and $i \in [R]$, compute
\[
  g_{x,a,i} := \langle \bv_{x,a}, \bv_0\rangle + \langle \bv_{x,a} - \bv_0 \langle \bv_0, \bv_{x,a}\rangle, \zeta^{(i)} \rangle.
\]
\item  
For each $x \in X$ and $i \in [R]$, compute
\[
  h_{x,i} := \Psi(g_{x,a,i} : a \in A).
\]

\item Let $H_\eta$ be the multilinear polynomial corresponding to the function $T_{1-\eta} p:\Delta_A^R\to\Delta_B$.\footnote{Here, we interpret the map $p : A^R \to B$ as real-valued function $p : A^R \to \R^B$ where $p(\ba)_b = 1$ if $p(\ba) = b$ and $0$ otherwise. Then, we apply the noise operator $T_{1-\eta}$ coordinate-wise. One can show for any $\ba \in \Delta_A^R$, we have that $T_{1-\eta} p(\ba) \in \Delta_B$.}
For each $x \in X,$ 
\[
  q_x := H_{\eta}(h_{x,1}, \hdots, h_{x,R}) \in \Delta_B.
\]
\item For each $x \in X$, assign $\sigma(x)$ according to the probability distribution $q_x$.
\end{itemize}

\subsection{Configurations and Rounding Schemes}

In order to describe our modifications to the BCR algorithm, we need to first abstract out the essential properties of the SDP rounding scheme.

\subsubsection{Configurations}
\label{subsec_configurations}

Given a unit vector $\bv_0 \in \R^N$, we define $\S^N_A(\bv_0) \subseteq (\R^N)^A$ to be the set of all $A$-tuples of vectors $\{\bv_a : a \in A\}$ satisfying the following criteria:
\begin{align*}
  \langle \bv_a, \bv_{a'} \rangle &= 0\ \ \  \forall a \neq a' \in A\\
  \sum_{a\in A} \bv_a &= \bv_0.
\end{align*}
We call these $A$-tuples \emph{local configurations} and typically denote them by $\bV \in \S^N_A$. Note for each $a \in A$, there is an \emph{integral} local configuration $\bV \in \S^N_A$ for which $\bv_a = \bv_0$ and $\bv_{a'} = 0$ for all $a' \in A \setminus \{a\}$. We let $\bI_a$ denote this integral local configuration.

When $\bv_0$ is fixed in a context, for succinctness we let $\S^N_A$ denote $\S^N_A(\bv_0)$. We define a (normed) metric on $\S^N_A$ by 
\[
  \|\bV-\bV'\|_2 = \sqrt{\sum_{a \in A}\|\bv_a - \bv'_a\|^2_2}.
\]
Note that, if $\bV$ and $\bV'$ are interpreted as
matrices in $\R^{N\times A}$, the above is the Frobenius matrix norm.

We now define a \emph{global configuration} to be a tuple of local configurations \[\cV := (\bV_x \in \S^N_A : x \in X)\in (\S^N_A)^X.\] In particular, given $\ba \in A^X$, we let $\cI_{\ba} := (\bI_{a_x} : x \in X)$ denote the integral global configuration. Given $\cV \in (\S^N_A)^X$, we let $\Mat(\cV) \in \R^{(X\times A)^2}$ denote the Gram matrix of all dot products. That is, for any $(x,a),(x',a') \in X \times A$, we have that
\[
    \Mat(\cV)_{(x,a),(x',a')} := \langle \bv_{x,a}, \bv_{x',a'}\rangle.
\]
For $\ba \in A^X$, we let $\Mat(\ba)$ be shorthand for $\Mat(\cI_{\ba})$. In other words $\Mat(\ba)_{(x,a),(x',a')} = 1$ iff $\ba_x = a$ and $\ba_{x'} = a'$, and $0$ otherwise. Observe that $\Mat(\ba)$ can be identified with the Kronecker product $\Pi_\ba\otimes\Pi_\ba$, with $\Pi_\ba$ defined as in Section~\ref{sec_separablesPCSPs}.

It shall be useful to notice that, for any global configuration $\cV$, the following equality holds:
\begin{align}
\label{eqn_trace_gram_matrix_configurations}
    \Tr(\Mat(\cV))
    =
    \sum_{(x,a)\in X\times A}\|\bv_{x,a}\|_2^2
    =
    \sum_{x\in X}\|\bv_0\|_2^2
    =
    |X|.
\end{align}

\subsubsection{SDP solutions and Rounding Schemes}\label{subsec:SDP-rounding}

Given a (weighted) instance $\cX$ of $\PCSP(\bA,\bB)$ on variable set $X$, any valid solution to the SDP relaxation of $\cX$ can be viewed as a global configuration $\cV \in (\S^N_A)^X$ for some suitably large dimension $N$ (which can be made at most $|X|\cdot |A|$, although we will often consider $N$ to be larger). 
However, not every $\cV \in (\S^N_A)^X$ is a valid solution to the SDP as each clause $(Y, P, Q)$ of $\cX$ imposes some structure on $\cV$. More precisely, if we let $\cV|_Y := (\bV_y : y \in Y)$, then there exists a probability distribution $\Lambda_{Y,P}$ on $A^Y$ such that
\begin{align}
  \Mat(\cV|_Y) = \underset{\ba \sim \Lambda_{Y,P}}{\E}[\Mat(\ba)].\label{eq:SDP-Mat}
\end{align}
We say that $\Lambda_{Y,P}$ is \emph{consistent} with $\cV$ if (\ref{eq:SDP-Mat}) holds. Note that this is the translation of the second-moment condition~\eqref{eq:2nd-moment}, while the first-moment condition~\eqref{eq:1st-moment} is implicit in the definition of local configurations.

We let $\V^N_{\cX,A} \subseteq (\S^N_A)^X$ be the set of all $\cV \in (\S^N_A)^X$ that satisfy (\ref{eq:SDP-Mat}) for some choice of $\Lambda_{Y,P}$ for all constraints $(Y,P,Q)$ of $\cX$. We further define the map $\comp_{\cX,\bA} : \V^N_{\cX,\bA} \to \mathbb [0,1]$ to be
\[
  \comp_{\cX,\bA}(\cV) := \max_{\substack{\{\Lambda_{Y,P} \mid (Y,P,Q) \in \cX\}\\\text{ consistent with $\cV$}}}\underset{(Y,P,Q) \sim \cX}{\E} \left[\Pr_{\ba \sim \Lambda_{Y,P}}[\ba \in P]\right].
\]
Analogously, we define $\sound_{\cX,\bB} : \Delta_B^{X} \to [0,1]$ to be 
\[
  \sound_{\cX,\bB}(\bz) := \underset{(Y,P,Q) \sim \cX}{\E}\left[\sum_{\bq \in Q}\prod_{y \in Y}z_y(q_y)\right].
\]
Intuitively, $\sound_{\cX,\bB}(\bz)$ is the expected value of a solution obtained by independently assigning each $x \in X$ a value from $B$ according to the distribution $z_x$.
To maximize the soundness, we seek to have control on how these distributions $z_x \in \Delta_B$ are selected. This motivates the notion of an (oblivious) rounding scheme.
\begin{definition}[Rounding scheme]
\label{defn_rounding_scheme}
Let $\Xi$ be a probability distribution. We define an (oblivious) \emph{rounding scheme} to be a function $S : \S^N_{A} \times \Xi \to \Delta_B$. If $S$ can be computed efficiently, then we say that $S$ is a \emph{rounding algorithm}.
\end{definition}

\begin{definition}[Robust rounding scheme]\label{def:robust-rounding-scheme}
We say that a rounding scheme $S : \S^N_{A} \times \Xi \to \Delta_B$ is $(\alpha,\beta)$-robust for an instance $\cX$ of $\PCSP(\bA,\bB)$ if for all global configurations $\cV \in \V^N_{\cX,\bA}$, we have that
\[
\comp_{\cX,\bA}(\cV) \ge 1-\alpha \implies \underset{\substack{{\xi \sim \Xi}}}{\E}[\sound_{\cX,\bB}(S(\bV_x, \xi) : x \in X)] \ge 1-\beta.
\] 
\end{definition}

Note that once $\xi \sim \Xi$ is selected, each $x \in X$ is rounded independently according to $S(\bV_x, \xi)$. However, the randomness in $\Xi$ will allow for nontrivial correlations which we systematically analyze. For the BCR algorithm, we consider the probability distribution $\Xi^{\BCR}_{N,\theta} := \cP \times \cN(0^N,I_N)^R$, where $\cP$ is the $R$-arity $(\alpha,\beta)$-robust, quasirandom polymorphism that gives a loss of $\theta$ in \Cref{thm:BCR}. Then, we can think of the BCR algorithm as a rounding algorithm $\BCR^N_{\theta} : \S^N_{A} \times \Xi^{\BCR}_{N,\theta} \to \Delta_B$ for any choices of $N \in \mathbb N$ and $\theta > 0$. In particular, once the SDP solution is found, the assigned probability distribution to each variable only depends on the global randomness chosen by $\Xi^{\BCR}_{N,\theta}$.  We can then restate the second implication in~\Cref{thm:BCR} as follows.

\begin{corollary}\label{cor:BCR}
Assume UGC. If $\PCSP(\bA,\bB)$ is $(\alpha,\beta)$-robust, then for any $\theta > 0$ and any instance $\cX$ of $\PCSP(\bA, \bB)$, for all sufficiently large $N$ (as a function of $\theta$ and $|\cX|$), %
$\BCR^N_{\theta} : \S^N_{A} \times \Xi^{\BCR}_{N,\theta} \to \Delta_B$ is $(\alpha-\theta,\beta+\theta)$-robust.
\end{corollary}

Using \Cref{cor:BCR} as the only property of the $\BCR$ algorithm we need, we shall construct another rounding procedure meeting the requirements of \Cref{thm:robust-eq}.

\subsection{Smoothing Configurations}

We say that $\cV, \cW \in (\S^N_{A})^X$ are \textit{$\delta$-modifications} of each other for some $\delta \ge 0$ if for all $x \in X$, we have that $\|\bV_x - \bW_x\|^2_2 \le \delta$. We let $B(\cV,\delta) \subseteq \V^N_{\cX,\bA}$ denote the set of all $\delta$-modifications of $\cV$. Note that if $\cV$ and $\cW$ are $\delta$-modifications with $\cV \in \V^N_{\cX,\bA}$, it may not be the case that $\cW \in \V^N_{\cX,\bA}$ as there may be a clause $(Y,P,Q)$ of $\cX$ for which no marginal distribution $\Lambda_{Y,P}$ is consistent with $\cW$. %

As a key definition, let $\V^{N,\delta}_{\cX,\bA} := \bigcup_{\cV \in \V^N_{\cX,\bA}}B(\cV,\delta)$. As an important step in constructing our rounding scheme, we first seek to extend $\BCR^N_{\theta}$ from being robust for $\V^N_{\cX,\bA}$ to being robust for $\V^{N,\delta}_{\cX,\bA}$.

\begin{definition}[Spaciously Robust rounding scheme]\label{def:spaciously-robust-rounding-scheme}
Given an instance $\cX$ of $\PCSP(\bA, \bB)$, we say that a  rounding scheme $S : \S^N_{A} \times \Xi \to \Delta_B$ is $(\delta,\alpha,\beta)$-spaciously robust for $\cX$ if for all global configurations $\cV \in \V^N_{\cX,\bA}$, and all $\cW \in B(\cV, \delta)$,
we have that
\[
\comp_{\cX,\bA}(\cV) \ge 1-\alpha \implies \underset{\substack{{\xi \sim \Xi}}}{\E}[\sound_{\cX,\bB}(S(\bW_x, \xi) : x \in X)] \ge 1-\beta.
\] 
\end{definition}
Clearly, for all weighted instances of the PCSP, any spaciously robust rounding scheme is in particular robust, as $\cV\in B(\cV,\delta)$ for any $\delta$.
Note also that in general $\comp_{\cX,\bA}(\cW)$ is not defined (since, as noted above, it may be that $\cW \not\in \V^N_{\cX,\bA}$), so we use the completeness of $\cV$ as a proxy. 

\subsubsection{Smooth Lifting}

Given an $(\alpha,\beta)$-robust scheme $S$, we would like a general method to turn it into a $(\delta, \alpha, \beta)$-spaciously robust scheme $S'$. To that end, we use a padding trick similar to one used by Raghavendra and Steurer~\cite{raghavendra2009round}. 

\begin{definition}
Given $\bu \in \R^M$, $\bv \in \R^N$, and $\rho \in [0,1]$, we let $\bu \oplus_{\rho} \bv \in \R^{M+N}$ be the $\rho$-concatenation of $\bu$ and $\bv$; i.e., the vector defined by
\[
  (\bu \oplus_{\rho} \bv)_i := \begin{cases}
    \sqrt{\rho} u_i & i \in [M]\\
    \sqrt{1-\rho} v_i & i \in [M+N] \setminus [M].
  \end{cases}
\]
for each $i \in [M+N]$.
\end{definition}

\begin{definition}
Given $\bU \in \S^M_A(\bu_0)$, $\bV \in \S^N_A(\bv_0)$, and $\rho \in [0,1]$, we define the $\rho$-concatenation of $\bU$ and $\bV$ as the list $\bU \oplus_{\rho} \bV := (\bu_{a} \oplus_{\rho} \bv_a : a \in A)$.
\end{definition}

\begin{proposition}
$\bU \oplus_{\rho} \bV \in \S^{M+N}_A(\bu_0 \oplus_{\rho} \bv_0)$. 
\end{proposition}

\begin{proof}
First, for any distinct $a,a' \in A$, we have that
\[
  \langle \bu_{a} \oplus_{\rho} \bv_a, \bu_{a'} \oplus_{\rho} \bv_{a'}\rangle = \rho \langle \bu_a, \bu_{a'}\rangle + (1-\rho) \langle \bv_a, \bv_{a'}\rangle = 0.
\]
Second,
\[
  \sum_{a \in A} \bu_{a} \oplus_{\rho} \bv_a = \left(\sum_{a \in A} \bu_{a}\right) \oplus_{\rho} \left(\sum_{a \in A} \bv_{a}\right) = \bu_0 \oplus_{\rho} \bv_0.\qedhere
\]
\end{proof}

Likewise, for $\cU \in (\S^M_A)^X$ and $\cV \in (\S^N_A)^X$, we can define $\cU \oplus_{\rho} \cV := (\bU_x \oplus_{\rho} \bV_x : x \in X)$. Recall the Gram matrix $\Mat(\cV)$ of a global configuration $\cV$ defined in~\Cref{subsec_configurations}.

\begin{proposition}\label{prop:alpha-concat}
$\Mat(\cU \oplus_{\rho} \cV) = \rho \Mat(\cU) + (1-\rho) \Mat(\cV).$
\end{proposition}
\begin{proof}
For any $x,x' \in X$ and $a,a' \in A$, we have that
\begin{align*}
  \Mat(\cU \oplus_{\rho} \cV)_{(x,a),(x',a')} &= \langle \bu_{x,a} \oplus_{\rho} \bv_{x,a}, \bu_{x',a'} \oplus_{\rho} \bv_{x',a'}\rangle.\\
&= \rho \langle \bu_{x,a}, \bu_{x',a'}\rangle + (1-\rho) \langle \bv_{x,a}, \bv_{x',a'}\rangle\\
&= \rho \Mat(\cU)_{(x,a),(x',a')} + (1-\rho) \Mat(\cV)_{(x,a),(x',a')},
\end{align*}
as desired.
\end{proof}

We now construct a specific global configuration $\cU$ that will be useful for our rounding scheme.

\begin{definition}
We say that $\cU_{X,A} \in (\S^M_A)^X$ is a \emph{uniform} configuration if for all $x, x' \in X$ and $a, a' \in A$, we have that
\[
    \langle \bu_{x,a}, \bu_{x',a'}\rangle = \begin{cases}
    \frac{1}{|A|} & x=x' \text{ and } a=a'\\
    0 & x=x' \text{ and } a\neq a'\\
    \frac{1}{|A|^2} & \text{otherwise}. 
    \end{cases}
\]
\end{definition}

\begin{proposition}\label{prop:uniform}
For $M = |X|\cdot|A|$, $\cU_{X,A}$ is well-defined. Further, for any instance $\cX$ of $\CSP(\bA)$, we have that $\cU_{X,A} \in \V^M_{\cX,\bA}$.
\end{proposition}
\begin{proof}
First, it is straightforward to verify that
\[
    \Mat(\cU_{X,A}) = \frac{1}{|A|^{|X|}}\sum_{\ba \in A^X} \Mat(\ba).
\]
Thus, $\Mat(\cU_{X,A})$ is PSD and thus has a realization in $\R^{|X| \cdot |A|}$. Now, for any clause $P$ of $\cX$ on variable set $Y \subseteq X$, we can also verify that
\[
    \Mat(\cU_{X,A})|_{(Y\times A)^2} = \frac{1}{|A|^{|Y|}}\sum_{\ba \in A^Y} \Mat(\ba).
\]
Thus, for each clause of $\cX$, $\cU_{X,A}$ is consistent with a uniformly random assignment. Thus, $\cU_{X,A} \in \V^M_{\cX,\bA}.$
\end{proof}

Given $\cV \in (\S^N_{A})^X$ and $\rho \in [0,1]$, we let $\cU_{X,A} \oplus_{\rho} \cV \in (\S^{M+N}_A)^X$ be the \emph{$\rho$-lift} of $\cV$. We next use $\rho$-lifting to make (\ref{eq:SDP-Mat}) more likely to be true.

\subsubsection{Lifting $\delta$-modifications}

Our current goal is to prove the following lemma.

\begin{lemma}\label{lem:spacious-lifting}
Let $\cX$ be a weighted instance of $\PCSP(\bA, \bB)$. For any $\delta > 0$, there exists $\rho = O_{\bA}(\sqrt{\delta}) \in [0,1]$ such that for all $\cV \in \V^{N}_{\cX,\bA}$ and $\cW \in B(\cV, \delta)$, we have that
\begin{enumerate}
\item[(1)] $\cU_{X,A} \oplus_{\rho} \cW \in \V^{|X|\cdot|A| + N}_{\cX,\bA}$.
\item[(2)] $\comp_{\cX,\bA}(\cU_{X,A} \oplus_{\rho} \cW) \ge \comp_{\cX,\bA}(\cV) - O_{\bA}(\sqrt{\delta}).$
\end{enumerate}
\end{lemma}

To prove this lemma, we need some machinery from convex geometry. Some ideas are inspired by Raghavendra--Steurer \cite{raghavendra2009round} and Appendix~A of \cite{BGS23:stoc}. We start with a simple proposition. 

\begin{proposition}\label{prop:dot-bound}
Let $\bu, \bu', \bv, \bv' \in \R^N$ be vectors of length at most $1$, such that $\|\bu-\bu'\|^2_2 \le \delta$ and $\|\bv -\bv'\|_2^2 \le \delta$. Then $|\langle \bu, \bv\rangle - \langle \bu', \bv'\rangle| \le 2\sqrt{\delta}$.
\end{proposition}
\begin{proof}
Using Cauchy--Schwarz, notice that
\begin{align*}
|\langle \bu, \bv\rangle - \langle \bu', \bv'\rangle| &\le |\langle \bu, \bv\rangle - \langle \bu, \bv'\rangle| + |\langle \bu, \bv'\rangle - \langle \bu', \bv'\rangle|\\
&\le \|\bu\|_2 \|\bv - \bv'\|_2 + \|\bu - \bu'\|_2 \|\bv'\|_2\\
&\le 2\sqrt{\delta}. \qedhere
\end{align*}
\end{proof}

Next, we fix a finite (nonempty) set $Y$. We let $K_{Y,A}, K^0_{Y,A} \subseteq \R^{(Y \times A)^2}$ denote the following two convex sets, that are linear subspaces of $\R^{(Y \times A)^2}$:
\begin{align*}
    K_{Y,A} &:= \left\{\sum_{\ba \in A^Y}\lambda_{\ba}\Mat(\ba) : \lambda_{\ba} \in \R\right\},\\
    K^0_{Y,A} &:= \left\{\sum_{\ba \in A^Y}\lambda_{\ba}\Mat(\ba) : \sum_{\ba \in A^Y} \lambda_{\ba} = 0\right\}.
\end{align*}

We will need the following key lemma. It can be directly derived from a combinatorial characterisation of low-dimensional projections of a hypermatrix proved in~\cite{ciardo2023approximate}. Moreover, we also give a self-contained proof involving the dual space $K_{Y,A}^*$ of $K_{Y,A}$. Full details of both arguments are postponed to~\Cref{app_omitted_proofs}.

\begin{restatable}{lemma}{lemdiffinK}
\label{lem:diff-in-K}
For any $\cV, \cW \in (\S_A^N)^Y$, we have that
\begin{itemize}
\item[(1)]  $\Mat(\cV)\in K_{Y,A}$.
\item[(2)] $\Mat(\cV) - \Mat(\cW) \in K_{Y,A}^0$.
\end{itemize}
\end{restatable}

We just need a bit more machinery, then we can prove \Cref{lem:spacious-lifting}. Define the set
\[
  L_{Y,A} := \{M \in K^0_{Y,A} : \|M\|_{\infty} \le 1\},
\]
where $\|M\|_{\infty}$ is the absolute value of the greatest entry of $M$. Note that $L_{Y,A}$ is convex and compact. Define also the function $f : K^0_{Y,A} \to \R_{\ge 0}$ by
\[
    f(M) := \min_{\substack{\lambda\\\sum_{\ba} \lambda_{\ba} = 0\\\sum_{\ba} \lambda_{\ba} \Mat(\ba) = M}} \max_{\ba \in A^Y} |\lambda_{\ba}|.
\]

We now prove the following.

\begin{proposition}\label{prop:kappa-finite}
$\kappa_{Y,A} := \sup_{M \in L_{Y,A}} f(M)$ is finite.
\end{proposition}

\begin{proof}
Pick an arbitrary basis $M_1, \hdots, M_r$ of $K^0_{Y,A}$. It is straightforward to see that $f(M_i) < \infty$ for all $i \in [r]$. Furthermore, if $M = \beta_1 M_1 + \cdots + \beta_r M_r$, then
\begin{align}
    f(M) \le \sum_{i=1}^r |\beta_i| f(M_i).\label{eq:beta-f}
\end{align}
Let $g(M)$ be the RHS of (\ref{eq:beta-f}). Since $M_1, \hdots, M_r$ form a basis, the map $M \mapsto (\beta_1, \hdots, \beta_r)$ is continuous. Thus, $g(M)$ is continuous. Thus, since $L$ is compact, $\max_{M \in L_{Y,A}} g(M)$ exists and is finite. Thus, $\kappa_{Y,A} := \sup_{M \in L_{Y,A}} f(M)$ is finite as well.
\end{proof}

We also need the following fact.

\begin{proposition}\label{prop:concat-comp}
We have that $\cU_{X,A} \oplus_{\rho} \cV \in \V^{|X|\cdot|A|+N}_{\cX,\bA}$ and
\begin{align}
    \comp_{\cX,\bA}(\cU_{X,A} \oplus_{\rho} \cV) \ge \comp_{\cX,\bA}(\cV) - O_{\bA}(\rho).\label{eq:concat-comp}
\end{align}
\end{proposition}
\begin{proof}
Fix a clause $(Y, P)$ of $\cX$ with $Y \subseteq X$ and $P \in \bA$. Pick a probability distribution $(p_{\ba} : \ba \in A^Y)$ such that $\Mat(\cV|_Y) = \sum_{\ba \in A^Y} p_{\ba}\Mat(\ba)$ and $\comp_{Y,P}(\cV|_{Y}) = \sum_{\ba \in P} p_{\ba}$. By methods similar to that of \Cref{prop:uniform}, we can observe that
\begin{align}
\Mat((\cU_{X,A} \oplus_{\rho} \cV)|_Y) = \sum_{\ba \in A^Y} \left(\frac{\rho}{|A|^{|Y|}} + (1-\rho)p_{\ba}\right)\Mat(\ba).\label{eq:U-V-rho}
\end{align}
Thus, $\cU_{X,A} \oplus_{\rho} \cV \in \V^{|X|\cdot|A|+N}_{\cX,\bA}$ and
\begin{align*}
    \comp_{Y,P}((\cU_{X,A} \oplus_{\rho} \cV)|_{Y}) &\ge \sum_{\ba \in P} \left(\frac{\rho}{|A|^{|Y|}} + (1-\rho)p_{\ba}\right)\Mat(\ba)\\&\ge (1-\rho) \comp_{Y,P}(\cV|_{Y})\\
    &\ge  \comp_{Y,P}(\cV|_{Y})  - O_{\bA}(\rho),
\end{align*}
where we use that the completeness is at most $1$. Taking the average of the above inequality for all clauses of $\cX$, we get (\ref{eq:concat-comp}).
\end{proof}

\begin{proof}[Proof of \Cref{lem:spacious-lifting}]
Keep $\rho \in [0,1]$ arbitrary for the moment.  Fix a clause $(Y,P,Q)$ of $\cX$. Let $\cV' := (\cU_{X,A} \oplus_{\rho} \cV)|_{Y}$ and $\cW' := (\cU_{X,A} \oplus_{\rho} \cW)|_{Y}$.  To prove (1) it suffices to show that $\cW'$ has a solution to (\ref{eq:SDP-Mat}). Recall from \Cref{prop:concat-comp} and (\ref{eq:U-V-rho}) that $\cV' \in \V^{|X|\cdot|A| +N}_{Y,P}$ and
\[
\Mat(\cV') = \sum_{\ba \in A^Y} \left(\frac{\rho}{|A|^{|Y|}} + (1-\rho)p_{\ba}\right)\Mat(\ba).
\]
Since $\cW \in B(\cV, \delta)$, we have that $\|(\cU_{X,A} \oplus_{\rho} \cV) - (\cU_{X,A} \oplus_{\rho} \cW)\|_2^2 \le (1-\rho)\delta$. By \Cref{prop:dot-bound}, we have that
\[
\|\Mat(\cV') - \Mat(\cW')\|_{\infty} \le 2\sqrt{(1-\rho)\delta}.
\]
That is, $\frac{1}{2\sqrt{(1-\rho)\delta}}(\Mat(\cW') - \Mat(\cV')) \in L_{Y,A}$. Thus, by \Cref{lem:diff-in-K} and \Cref{prop:kappa-finite}, there exists a choice of $\lambda_{\ba} \in \R$ for all $\ba \in A^Y$ such that $\Mat(\cW') - \Mat(\cV') = \sum_{\ba \in A^Y} \lambda_{\ba} \Mat(\ba)$ and $|\lambda_{\ba}| \le 2\kappa_{Y,A}\sqrt{(1-\rho)\delta}$. Pick $\rho$ minimal such that 
\[
2\kappa_{Y,A}\sqrt{(1-\rho)\delta} \le \frac{\rho}{|A|^{|Y|}}
\]
for every clause $(Y,P,Q)$ of $\cX$. Note that $\rho = O_{\bA}(\sqrt{\delta})$ and $\rho \in [0,1]$ because the LHS is $0$ when $\rho = 1$. As such, we now have that
\[
  \Mat(\cW') = \sum_{\ba \in A^Y} \left(\frac{\rho}{|A|^{|Y|}} + \lambda_{\ba} + (1-\rho)p_{\ba}\right)\Mat(\ba),
\]
where every coefficient is nonzero. Thus, $\cW'$ satisfies (\ref{eq:SDP-Mat}), so $\cU_{X,A} \oplus_{\rho} \cW \in \V^{|X|\cdot|A|+N}_{\bX,\bA}$. This proves (1).

To prove (2), note that for a fixed clause $(Y,P,Q)$ of $\cX$, we have that
\begin{align*}
  \comp_{Y,P}(\cU_{X,A} \oplus_{\rho} \cW) &\ge \sum_{\ba \in P} \left(\frac{\rho}{|A|^{|Y|}} + \lambda_{\ba} + (1-\rho)p_{\ba}\right)\\
&\ge (1-\rho)\sum_{\ba \in P}p_{\ba}\\
&= (1-\rho) \comp_{Y,P}(\cV)\\
&\ge \comp_{Y,P}(\cV) - O_{\bA}(\sqrt{\delta}).
\end{align*}
Averaging this inequality over all clauses $(Y,P,Q)$ of $\cX$ proves (2), as desired.
\end{proof}

\subsubsection{The Pullback Scheme}\label{subsec:pullback}

In general, for an SDP rounding scheme to be robust, we need that SDP vectors which are close to each other have similar soundness when rounded. \emph{A priori}, we do not know if the $\BCR$ rounding schemes have such a property. Toward, we show how to smooth the $\BCR$ rounding scheme using a technique we call the ``Pullback Scheme.'' In \Cref{subsec:REQ} we will further modify the pullback scheme to have even stronger robustness properties.

For a fixed $\theta > 0$ and a fixed instance $\cX$ of $\PCSP(\bA, \bB)$, by \Cref{cor:BCR} there exists $N$ such that $\BCR^{N'}_{\theta}$ is an $(\alpha-\beta, \beta+\theta)$-robust scheme for $\cX$ for all $N' \ge N$. In particular, for any finite set $Z$, %
we have that $\BCR^{|Z|\cdot |A| + N}_{\theta}$ is an $(\alpha-\theta, \beta+\theta)$-robust scheme for $\cX$. We now define a \emph{pullback scheme} $\widehat{\BCR}^{N,Z}_{\theta,\rho}$ for any $\delta \in (0,1)$ as follows.

\begin{definition}[Pullback Scheme]\label{def:pullback}
For any finite set $Z$, we define $\widehat{\BCR}^{N,Z}_{\theta,\delta} : \S^N_A \times \Xi^{\BCR}_{|Z|\cdot |A| + N,\theta} \to \Delta_B$ to be
\[
  \widehat{\BCR}^{N,Z}_{\theta,\delta}(\bV, \xi) = \underset{z \sim Z}{\E}\left[\BCR^{|Z|\cdot |A| + N}_{\theta}(\bU_{z} \oplus_{\rho} \bV, \xi)\right],
\]
where $z\sim Z$ is the uniform distribution, $\bU_z$ is the $z$-th local configuration of $\cU_{Z,A}$, and $\rho = O_{\bA}(\sqrt{\delta})$ is chosen according to \Cref{lem:spacious-lifting}.
\end{definition}

Note that we call this a ``pullback'' scheme as we are  mapping a higher-dimensional rounding scheme to a lower dimension. We now show this pullback is spaciously robust.

\begin{lemma}\label{lem:pullback-spaciously-robust}
If $\BCR^{|X|\cdot |A| + N}_{\theta}$ is an $(\alpha-\theta, \beta+\theta)$-robust scheme for an instance $\cX$ of $\PCSP(\bA,\bB)$, then $\widehat{\BCR}^{N,Z}_{\theta,\delta}$ is a $(\delta, \alpha', \beta+\theta+\frac{|X|^2}{|Z|})$-spaciously robust scheme for $\cX$, where $\alpha' = \alpha - \theta - O_{\bA}(\sqrt{\delta})$.
\end{lemma}

\begin{proof}
Let $\cV \in \V^N_{\cX,\bA}$ satisfy $\comp_{\cX,\bA}(\cV) \ge 1 - \alpha'$. Let $\cW \in B(\cV, \delta)$. By \Cref{def:pullback} and the fact that $\sound_{\cX,\bB}$ is multilinear, we have %
\begin{align*}
\underset{\substack{{\xi \sim \Xi^{\BCR}_{|Z|\cdot |A| + N, \theta}}}}{\E}&[\sound_{\cX,\bB}(\widehat{\BCR}^{N,Z}_{\theta,\delta}(\bW_x, \xi) : x \in X)]\\
&= \underset{\substack{{\xi \sim \Xi^{\BCR}_{|Z|\cdot |A| + N, \theta}}}}{\E}\left[\sound_{\cX,\bB}\left(\underset{z \sim Z}{\E}\left[\BCR^{|Z|\cdot |A| + N}_{\theta}(\bU_{z} \oplus_{\rho} \bW_x, \xi)\right] : x \in X\right)\right]\\
&= \underset{\substack{{\xi \sim \Xi^{\BCR}_{|Z|\cdot |A| + N, \theta}}\\{\bz \sim Z^X}}}{\E}\left[\sound_{\cX,\bB}\left(\BCR^{|Z|\cdot |A| + N}_{\theta}(\bU_{z_x} \oplus_{\rho} \bW_x, \xi) : x \in X\right)\right],
\end{align*}
where $\bz \sim Z^X$ is chosen uniformly at random. For each $\bz \in Z^X$, let $\cW_{\bz} := (\bU_{z_x} \oplus_{\rho} \bW_x : x \in X)$. If $\bz$ is injective---that is, if $z_x = z_{x'}$ implies $x = x'$---then we claim that
\[
  \Mat(\cW_{\bz}) = \Mat(\cU_{X,A} \oplus_{\rho} \cW).
\]
To see why, for all $x, x' \in X$ and $a, a' \in A$, we have that
\begin{align*}
\Mat(\cW_{\bz})_{(x,a),(x',a')} &= \rho \langle \bU_{z_x,a}, \bU_{z_{x'},a'}\rangle + (1-\rho) \langle \bW_{x,a}, \bU_{x',a'}\rangle\\
&= \rho \langle \bU_{x,a}, \bU_{x',a'}\rangle + (1-\rho) \langle \bW_{x,a}, \bU_{x',a'}\rangle&\text{($z$ injective)}\\
&= \Mat(\cU_{X,A} \oplus_{\rho} \cW)_{(x,a),(x',a')}.
\end{align*}

Thus, by \Cref{lem:spacious-lifting}, for all injective $\bz$, we have that $\cW_{\bz} \in \V^{|Z|\cdot |A| + N}_{\cX,\bA}$ and
\[
 \comp_{\cX,\bA}(\cW_{\bz}) \ge \comp_{\cX,\bA}(\cV) - O_{\bA}(\sqrt{\delta}) = 1 - \alpha + \theta.
\]
Thus, since $\BCR^{|X|\cdot |A| + N}_{\theta}$ is $(\alpha-\theta, \beta+\theta)$-robust, we have that
\begin{align*}
   1-\beta-\theta &\le \underset{\substack{{\xi \sim \Xi^{\BCR}_{|Z|\cdot |A| + N, \theta}}\\{\bz \sim Z^X}}}{\E}\left[\sound_{\cX,\bB}\left(\BCR^{|Z|\cdot |A| + N}_{\theta}(\bU_{z_x} \oplus_{\rho} \bW_x, \xi) : x \in X\right) \middle| \text{ $\bz$ injective}\right]\\
&\le \frac{1}{\Pr_{\bz \sim Z^X}[\bz\text{ injective}]}\underset{\substack{{\xi \sim \Xi^{\BCR}_{|Z|\cdot |A| + N, \theta}}\\{\bz \sim Z^X}}}{\E}\left[\sound_{\cX,\bB}\left(\BCR^{|Z|\cdot |A| + N}_{\theta}(\bU_{z_x} \oplus_{\rho} \bW_x, \xi) : x \in X\right)\right]\\
&= \frac{1}{\Pr_{\bz \sim Z^X}[\bz\text{ injective}]} \underset{\substack{{\xi \sim \Xi^{\BCR}_{|Z|\cdot |A| + N, \theta}}}}{\E}[\sound_{\cX,\bB}(\widehat{\BCR}^{N,Z}_{\theta,\delta}(\bW_x, \xi) : x \in X)].
\end{align*}
To finish, we note that

\[
\Pr_{\bz \sim Z^X}[\bz\text{ injective}] = \prod_{i=1}^{|X|} \left(1 - \frac{i-1}{|Z|}\right) \ge 1 - \frac{|X|^2}{|Z|},
\]
so
\[
\underset{\substack{{\xi \sim \Xi^{\BCR}_{|Z|\cdot |A| + N, \theta}}}}{\E}[\sound_{\cX,\bB}(\widehat{\BCR}^{N,Z}_{\theta,\delta}(\bW_x, \xi) : x \in X)] \ge (1 - \beta - \theta) \cdot (1 - \frac{|X|^2}{|Z|}) \ge 1 - \beta - \theta - \frac{|X|^2}{|Z|},
\]
as desired.
\end{proof}

\subsection{Measure Theory Fundamentals}

Our next step toward the proof of \Cref{thm:robust-eq} is to modify $\widehat{\BCR}$ to another scheme (which we will call $\REQ$) so that similar inputs give similar outputs. In order to do this, we first need some fundamentals of measure theory. We refer to the textbook of DiBenedetto~\cite{dibenedetto2016Real} for foundational concepts. 

Let $R$ be a possibly infinite set, and let $\mathcal R := (R, \mu)$ be a probability measure over $R$. See Chapter 3 of \cite{dibenedetto2016Real} for a precise definition.
Recall that a function $f : R \to \mathbb R^d$ is \textit{measurable} if the preimage under $f$ of every measurable subset of $\R^d$ is measurable.  
Let also $\int_{R}$ denote the Lebesgue integral over $R$.
We shall consider the set  $L^{2,2}(\cR, \mathbb R^d)$ consisting of the quotient space of the set of all measurable functions $f : R \to \mathbb R^d$ 
 such that $\int_{R} \|f(x)\|_2^2\,d\mu(x) < \infty$,
modulo the equivalence relation ``$\sim$'' defined by $f \sim g$ if and only if $f-g$ is only nonzero on a set of zero measure. 
Observe that $L^{2,2}(\cR, \mathbb R^d)$ is a normed space, with the norm
\begin{align}
    \|f\|_{2,2} = \sqrt{\int_{R} \|f(x)\|_2^2\,d\mu(x)}.\label{eq:2-norm}
\end{align}

We say that $f \in L^{2,2}(\cR, \R^d)$ is a limit of a sequence of functions $f_1, f_2, \hdots \in L^{2,2}(\cR, \R^d)$ if, for all $\eps > 0$, there exists $N \in \mathbb N$ such that $\|f - f_i\|_{2,2} \le \eps$ for all $i \ge N$. It is clear that every sequence has at most one limit. A \emph{Cauchy sequence} is a sequence of functions $f_1, f_2, \hdots \in L^{2,2}(\cR, \mathbb R^d)$ such that, for all $\eps > 0$, there exists $N \in \mathbb N$ such that $\|f_i - f_j\|_{2,2} \le \eps$ for any $i,j \ge N$. We say that $L^{2,2}(\cR, \mathbb R^d)$ is a \emph{Banach space} if every Cauchy sequence has a limit. The following is a simple consequence of the Riesz--Fischer theorem.

\begin{theorem}[Riesz--Fischer, see Chapter 6, Theorem 5.1 of \cite{dibenedetto2016Real}]\label{thm:rf}
$L^{2,2}(\cR, \R) = L^2(\cR)$ is a Banach space.\footnote{Technically, the treatment in \cite{dibenedetto2016Real} allows for $f$ to take on the value $\pm \infty$, but it is clear from (\ref{eq:2-norm}) that $f^{-1}(\{-\infty,\infty\})$ must have zero measure in order for $\|f\|_{2,2} < \infty$. Thus, we can assume $f$ has finite output. 
}
\end{theorem}

\begin{corollary}\label{cor:banach}
For all $d \ge 1$, $L^{2,2}(\cR, \R^d)$ is a Banach space.
\end{corollary}

\begin{proof}
Consider any Cauchy sequence $f_1, f_2, \hdots, \in L^{2,2}(\cR, \R^d)$. For each $i \in [d]$, note that $f_{1,i}, f_{2,i}, \hdots$ is a Cauchy sequence in $L^{2,2}(\cR, \R)$. Thus, by \Cref{thm:rf}, this sequence converges to $g_i \in L^{2,2}(\cR, \R)$. Let $g = (g_1, \hdots, g_d)$. It is easy to check that $g \in L^{2,2}(\cR, \R^d)$. Furthermore, since $d$ is finite, we have that
\begin{align*}
  0 = \sqrt{\sum_{j=1}^d 0}
    &= \sqrt{\sum_{j=1}^d \lim_{i \to \infty} \|f_{i,j} - g_j\|^2_{2,2}}\\
    &= \lim_{i \to \infty} \sqrt{\sum_{j=1}^d  \|f_{i,j} - g_j\|^2_{2,2}}\\
    &=  \lim_{i \to \infty} \|f_i - g\|_{2,2}.
\end{align*}
Thus, $g$ is indeed the limit of the Cauchy sequence $f_1, f_2, \hdots$, as desired.
\end{proof}

We say that $K \subseteq L^{2,2}(\cR, \R^d)$ is \textit{convex} if, for any $f_1, f_2 \in K$ and $\lambda \in [0,1]$, we have that $\lambda f_1 + (1-\lambda)f_2 \in K$. We let $L^{2,2}(\cR, \Delta_d) \subseteq L^{2,2}(\cR, \mathbb R^d)$ be the subset of functions whose range lies in $\Delta_d$. Note that $L^{2,2}(\cR, \Delta_d)$ is convex. 

Given $K \subseteq L^{2,2}(\cR, \R^d)$, we let $\overline{K} \subseteq L^{2,2}(\cR, \R^d)$ denote its (sequential) \textit{closure}. That is, $f \in \overline{K}$ if and only if, for all $\eps > 0$, there exists $g \in K$ with $\|f-g\|_{2,2} \le \eps$. Clearly if $K$ is convex then so is $\overline{K}$. 
The following fact is well known. To get some familiarity with the definitions, we include a simple proof.

\begin{proposition}
If $K \subseteq L^{2,2}(\cR, \Delta_d)$, then $\overline{K} \subseteq L^{2,2}(\cR, \Delta_d)$.
\end{proposition}

\begin{proof}
Assume for sake of contradiction that there is $f \in \overline{K}$ with $f \not\in L^{2,2}(\cR, \Delta_d)$. Let $S = f^{-1}(\R^d \setminus \Delta_d)$. Note that $\mu(S) > 0$. Since $f$ is measurable, there must exist $\delta > 0$ and $T \subseteq S$ such that $\mu(T) > 0$ and, for all $x \in T$, the $\ell_1$ distance between $x$ and $\Delta_d$ is greater than $d$. Therefore, for all $g \in L^{2,2}(\cR, \Delta_d)$, we have that
\[
    \|f - g\|_{2,2}^2 \ge \int_{T} \|f(x)-g(x)\|_1^2\,d\mu(x) \ge \delta^2\mu(T).
\]
Since $\delta^2\mu(T) > 0$, this contradicts the fact that $f \in \overline{K}$.
\end{proof}

 We next observe an elementary fact about our norm for $L^{2,2}(\cR,\R^d)$.

\begin{proposition}\label{prop:l2-average}
If $f,g \in L^{2,2}(\cR, \R^d)$ satisfy $\|f - g\|_{2,2} \ge \delta$ then \[\|(f+g)/2\|^2_{2,2} \le \max(\|f\|^2_{2,2},\|g\|^2_{2,2}) - \delta^2/4.\]
\end{proposition}
\begin{proof}
Note that
\begin{align*}
\|(f+g)/2\|_{2,2}^2 &= \frac{1}{4}\int_{R} \|f(x)+g(x)\|_2^2\,d\mu(x) \\
&= \frac{1}{4}\int_{R} 2\|f(x)\|_2^2 + 2\|g(x)\|_2^2 - \|f(x) - g(x)\|_2^2\,d\mu(x) \\
&= \frac{1}{4}(2\|f\|_{2,2}^2+2\|g\|_{2,2}^2 - \|f-g\|_{2,2}^2)\\
&\le \max(\|f\|^2_{2,2},\|g\|^2_{2,2}) - \delta^2/4. \qedhere
\end{align*}
\end{proof}

We now show that we can uniquely solve a suitable optimization problem over $L^{2,2}(\cR, \Delta_d)$. 

\begin{lemma}\label{lem:l2-min}
Let $K \subseteq L^{2,2}(\cR, \R^d)$ be nonempty and convex, then there exists a unique $f \in \overline{K}$ for which $\|f\|_{2,2}$ is minimized.
\end{lemma}

\begin{proof}
First, note that if a minimizer exists it must be unique. Otherwise, there are $f,g \in \overline{K}$ for which $\|f\|_{2,2} = \|g\|_{2,2}$ is the minimum, but \Cref{prop:l2-average} shows that $(f+g)/2 \in \overline{K}$ must have a smaller norm for $\|f-g\|_{2,2} > 0$ by definition of being distinct (recall that for $f$ and $g$ to be distinct we require that $f\neq g$ on a set of nonzero measure).

Thus, it suffices to prove that a minimizer exists. Let $\alpha = \inf \{\|f\|_{2,2}^2 : f \in \overline{K}\} < \infty$. By definition of the infimum, for all $n \in \mathbb N$, there exists $f_n \in \overline{K}$ for which $\|f_n\|_{2,2}^2 \le \alpha+\frac{1}{n}$. Now consider any positive integers $n \ge m$ and observe by \Cref{prop:l2-average} that
\[
    \alpha \le \|(f_n + f_m)/2\|^2_{2,2} \le \max(\|f_n\|^2_{2,2}, \|f_m\|^2_{2,2}) - \frac{1}{4}\|f_n - f_m\|^2_{2,2} \le \alpha + \frac{1}{m} - \frac{1}{4}\|f_n - f_m\|^2_{2,2}.
\]
Thus, $\|f_n - f_m\|^2_{2,2} \le 4/m$. Hence $f_1, f_2, \hdots$ is a Cauchy sequence. By \Cref{thm:rf}, we have that its limit $f$ lives in $\overline{K}$. Furthermore, since $\|f_m\|^2_{2,2}$ approaches $\alpha$ as $m \to \infty$, it is straightforward to show that $\|f\|^2_{2,2} = \alpha$. Thus, $f$ is the unique minimizer of the norm in $\overline{K}$. 
\end{proof}

\subsection{Defining $\REQ$}\label{subsec:REQ}

Recall that in \Cref{subsec:pullback}, we defined a family of spaciously robust rounding schemes $\widehat{\BCR}$. We now seek to further modify $\widehat{\BCR}$ into another scheme called $\REQ$ which supports an additional property: any two close vectors round to essentially the same distribution over $\Delta_B$. We achieve this via convex optimization.

Recall we previously defined a \emph{rounding scheme} to be a function $S : \S^N_A \times \Xi \to \Delta_B$ (\Cref{defn_rounding_scheme}). Equivalently, we can think of $S$ as having signature $\S^N_A \to L^{2,2}(\Xi, \Delta_B)$. We assume this latter presentation from now on. 

\begin{definition}[$\delta$-spread]\label{def:spread}
Recall that, for any $\bV \in \S^N_A$, $B(\bV, \delta)$ denotes the set of all $\bW \in \S^N_A$ with $\|\bV -\bW\|^2_2 \le \delta$. We define a \emph{$\delta$-spread} of $\bV$ to be a finitely supported
probability distribution $\Lambda$ over $B(\bV, \delta)$ such that for all $p \in [0,1]$, we have that
\[
  \Pr_{\bW \sim \Lambda}[\|\bV -\bW\|^2_2 \le p^2\delta] \ge p.
\]
\end{definition}

We say that $S' : \S^N_A \to L^{2,2}(\Xi, \Delta_B)$ is a \textit{$\delta$-smoothing} of $S : \S^N_A \to L^{2,2}(\Xi, \Delta_B)$ if, for all $\bV \in \S^N_A$, there exists a $\delta$-spread $\Lambda_{\bV}$ of $\bV$ such that
\[
  S'(\bV) = \underset{\bW \sim \Lambda_{\bV}}{\E}\left[S(\bW)\right].
\] 
We let $B(S, \bV, \delta) \subseteq L^{2,2}(\Xi, \Delta_B)$ denote all possible values of $S'(\bV)$.

\begin{proposition}
$B(S, \bV, \delta)$ is convex.
\end{proposition}
\begin{proof}
Consider any $S', S'' \in B(S, \bV, \delta)$ and let $\Lambda', \Lambda''$ be $\delta$-spreads of $\bV$ such that $S'(\bV) = \underset{\bW \sim \Lambda'}{\E}\left[S(\bW)\right]$ and $S''(\bV) = \underset{\bW \sim \Lambda''}{\E}\left[S(\bW)\right]$. For any $\theta \in [0,1]$, let $\Lambda_{\theta}$ be the probability distribution that samples from $\Lambda'$ with probability $\theta$ and from $\Lambda''$ with probability $1-\theta$. It is not hard to see that
\[
    (\theta S' + (1-\theta)S'')(\bV) = \underset{\bW \sim \Lambda_{\theta}}{\E}\left[S(\bW)\right].
\]
It remains to verify that $\Lambda_{\theta}$ is a $\delta$-spread of $\bV$. To see why, note that for any $p \in [0,1]$, we have that
    \begin{align*}
      \Pr_{\bW \sim \Lambda_{\theta}}[\|\bV -\bW\|^2_2 \le p^2\delta] &= \theta\Pr_{\bW \sim \Lambda'}[\|\bV -\bW\|^2_2 \le p^2\delta] + (1-\theta)\Pr_{\bW \sim \Lambda''}[\|\bV -\bW\|^2_2 \le p^2\delta]\\
      &\ge \theta p + (1 - \theta) p = p
    \end{align*}
as desired.
\end{proof}

We now define $\REQ^{N,Z}_{\theta,\delta} : \S^N_A \to L^{2,2}(\Xi^{\BCR}_{|Z|\cdot|A| + N,\theta}, \Delta_B)$ in terms of $\widehat{\BCR}^{N,Z}_{\theta,\delta} : \S^{N}_A \to L^{2,2}(\Xi^{\BCR}_{|Z|\cdot|A| + N,\theta}, \Delta_B)$ as follows.

\begin{definition}\label{def:REQ}
For all $\bV \in \S^N_A$ and sets $Z$, we let
\[
  \REQ^{N,Z}_{\theta,\delta}(\bV) := \operatorname{argmin} \{\|f\|_{2,2} : f \in \overline{B(\widehat{\BCR}^{N,Z}_{\theta,\delta}, \bV, \delta)}\}.
\]
\end{definition}
Note that the scheme is well defined by \Cref{lem:l2-min}. Intuitively, the argmin performs a local optimization over many $\widehat{\BCR}$ rounding schemes. This minimum is fairly robust to modification of the inputs. In particular, we have the following two lemmas which are proved in the next two subsections. The first lemma says that the robustness of $\REQ$ for $\PCSP(\bA, \bB)$ barely degrades in comparison to $\widehat{\BCR}$. The second lemma says that $\REQ$ rounds nearby vectors to similar distributions.

\begin{lemma}\label{lem:req-ab-robust}
Let $\cX$ be a weighted instance of $\PCSP(\bA, \bB)$. Assume $\BCR^{|Z|\cdot|A|+N}_{\theta}$ is $(\alpha-\theta, \beta+\theta)$-robust for $\cX$. Then $\REQ^{N,Z}_{\theta,\delta}$ is $(\alpha-\theta-O_{\bA}(\sqrt{\delta}), \beta+\theta+\frac{|X|^2}{|Z|})$-robust for $\cX$, where $O_{\bA}(\sqrt{\delta})$ is from \Cref{lem:pullback-spaciously-robust}.
\end{lemma}

\begin{lemma}\label{lem:req-eq-robust}
Assume $\bV, \bW \in \S^N_A$ with $\|\bV-\bW\|^2_2 \le \delta \cdot \eta^2$ for some $\eta > 0$. Then, $\|\REQ^{N,Z}_{\theta,\delta}(\bV) - \REQ^{N,Z}_{\theta,\delta}(\bW)\|_{2,2} \le 8\sqrt{\eta}$.
\end{lemma}

\subsubsection{Robustness is Preserved: \Cref{lem:req-ab-robust}}

Proving \Cref{lem:req-ab-robust} is relatively straightforward. Since, $\widehat{\BCR}^{N,Z}_{\theta,\delta}$ is already spaciously robust, it suffices to show that any $\delta$-spread which $\REQ^{N,Z}_{\theta,\delta}$ could select has similar soundness.

\begin{proof}[Proof of \Cref{lem:req-ab-robust}]
By \Cref{lem:l2-min}, the definition of $\REQ^{N,Z}_{\theta,\delta}$ is well-defined. %
Thus, for all $\epsilon > 0$, there exists a $\delta$-smoothing $S$ of $\widehat{\BCR}^{N,Z}_{\theta,\delta}$ such that for all $\bV \in \S^N_A$, we have that \[\|\REQ^{N,Z}_{\theta,\delta}(\bV) - S(\bV)\|_{2,2} = \sqrt{\underset{\xi \in \Xi^{\BCR}_{|Z|\cdot|A|+N,\theta}}{\E}[(\REQ^{N,Z}_{\theta,\delta}(\bV,\xi) - S(\bV,\xi))^2]} < \eps.\]
Let $\alpha' = \alpha - \theta - O_{\bA}(\sqrt{\delta})$. Fix $\cV \in \V^N_{\cX,\bA}$ for which $\comp_{\cX,\bA}(\cV) \ge 1 - \alpha'$. For each $x \in X$, let $\Lambda_x$ be a probability distribution over $B(\bV_x, \delta)$ such that
\begin{align}
  S(\bV) = \underset{\bW \sim \Lambda_{x}}{\E}\left[\widehat{\BCR}^{N,Z}_{\theta,\delta}(\bW)\right].\label{eq:S'-Lambda}
\end{align}
Let $\W$ be a finite probability distribution over $B(\cV,\delta)$, where for each $x \in X$, the distribution $\bW_x$ (with $\cW \sim \W$) is an independent sample from $\Lambda_{x}$. Note that every $\cW \in \supp \W$ is a $\delta$-modification of $\cV$, thus by \Cref{lem:pullback-spaciously-robust} we have that
\[
\underset{\substack{W \sim \W\\{\xi \sim \Xi^{\BCR}_{|Z|\cdot |A| + N, \theta}}}}{\E}\left[\sound_{\cX,\bB}(\widehat{\BCR}^{N,Z}_{\theta,\delta}(\bW, \xi) : x \in X)\right] \ge 1 - \beta - \theta - \frac{|X|^2}{|Z|}.
\]
Like in the proof of \Cref{lem:pullback-spaciously-robust}, we use the fact that $\sound_{\cX,\bB} : \Delta_B^X \to [0,1]$ is a multilinear function. Thus, since $\W$ is a product distribution, we have that
\begin{align*}
1-\beta -\theta - \frac{|X|^2}{|Z|}&\le \underset{\substack{{\xi \sim \Xi^{\BCR}_{|Z|\cdot |A| + N, \theta}}}}{\E}\left[\sound_{\cX,\bB}(\E_{\bW_x \sim \Lambda_x}[\widehat{\BCR}^{N,Z}_{\theta,\delta}(\bW_x, \xi)] : x \in X)\right]\\
&= \underset{\substack{{\xi \sim \Xi^{\BCR}_{|Z|\cdot |A| + N, \theta}}}}{\E}\left[\sound_{\cX,\bB}(S(\bW_x, \xi) : x \in X)\right],
\end{align*}
where the last equality uses (\ref{eq:S'-Lambda}). Since $\sound_{\cX,\bB} : \Delta_B^X \to [0,1]$ is bounded and multilinear, we have for any $\bp, \bq \in \Delta_B^X$ the coarse bound
\[
    \sound_{\cX,\bB}(\bp) - \sound_{\cX,\bB}(\bq) \le \sum_{x \in X} |p_x - q_x| \le \sqrt{|X| \sum_{x \in X} (p_x - q_x)^2},
\]
where the latter inequality follows from Cauchy--Schwarz. Thus, since $\sqrt{\cdot}$ is a concave function,
\begin{align*}
    \underset{\substack{\xi \sim \Xi^{\BCR}_{|Z|\cdot|A| + N,\theta}}}{\E} &\left[\sound_{\cX,\bB}(S(\bV_x, \xi) : x \in X) - \sound_{\cX,\bB}(\REQ^{N,Z}_{\theta,\delta}(\bV_x, \xi) : x \in X)\right]\\&\le \sqrt{|X|} \underset{\substack{\xi \sim \Xi^{\BCR}_{|Z|\cdot|A| + N,\theta}}}{\E} \left[\sqrt{\sum_{x \in X} (S(\bV_x, \xi) - \REQ^{N,Z}_{\theta,\delta}(\bV_x, \xi))^2}\right]\\
    &\le \sqrt{|X| \cdot \sum_{x \in X} \underset{\xi \in \Xi^{\BCR}_{|Z|\cdot|A|+N,\theta}}{\E}[(S(\bV,\xi) - \REQ^{N,Z}_{\theta,\delta}(\bV,\xi))^2]}\\
    &< \sqrt{|X| \cdot \sum_{x \in X} \eps^2} = |X|\eps.
\end{align*}
Therefore, \[
    \underset{\substack{\xi \sim \Xi^{\BCR}_{|Z|\cdot|A| + N,\theta}}}{\E} \left[\sound_{\cX,\bB}(\REQ^N_{\theta,\delta,x}(\bV_x, \xi) : x \in X)\right] \ge 1-\beta -\theta- |X|\eps.
    \]
Taking the limit as $\eps \to 0$ proves that $\REQ^N_{\theta,\delta}$ is $(\alpha-\theta-O_{\bA}(\sqrt{\delta}), \beta+\theta)$-robust for $\cX$.
\end{proof}

\subsubsection{Smoothness of $\REQ^N_{\theta,\delta}$: \Cref{lem:req-eq-robust}}

We now prove \Cref{lem:req-eq-robust}. Intuitively, we seek to argue that for two nearby vectors $\bV, \bW \in \S^N_A$, the optimization in \Cref{def:REQ} gives similar answers. We do this by a ``strategy stealing'' argument, for any $\delta$-spread for $\bV$, there is a similar $\delta$-spread for $\bW$. This is done by taking the $\delta$-spread for $\bV$ and modifying it as little as possible to be a legal $\delta$-spread for $\bW$. This leads to the following proposition.

\begin{proposition}\label{prop:smooth-close}
Let $S : \S^N_A \times \Xi \to \Delta_B$ be a rounding scheme. Consider $\bV, \bW \in \S^N_A$ with $\|\bV -\bW\|^2_2 \le \delta \cdot \eta^2$ for $\eta \in (0,1)$. For any $f \in B(S, \bV, \delta)$, there is $g \in B(S, \bW, \delta)$ with $\|f-g\|_{2,2} \le 2\eta$.
\end{proposition}

\begin{proof}
Let $\Lambda$ be a $\delta$-spread of $\bV$ such that
\begin{align*}
    \forall p \in [0,1],\; p &\le \Pr_{\bV' \sim \Lambda}[\|\bV -\bV'\|^2_2 \le p^2\delta]\\
    f &= \underset{\bV' \sim \Lambda}{\E}\left[S(\bV')\right].
\end{align*}

Recall that $\Lambda$ is a discrete probability distribution. Let $\bV'_1, \hdots, \bV'_k \in B(\bV,\delta)$ be the support of this distribution with $p_1, \hdots, p_k \in [0,1]$ the corresponding probabilities. We assume that $\|\bV - \bV'_1\|_2^2 \le \cdots \le \|\bV - \bV'_k\|_2^2$. Pick $\ell \in [k]$ minimal such that $p_1 + \cdots + p_\ell \ge 1 - \eta$. Let $\Lambda'$ be a probability distribution supported on $\{\bW, \bV'_1, \hdots, \bV'_\ell\}$ which samples $\bW$ with probability $\eta$, $\bV'_i$ for $i \in [\ell-1]$ with probability $p_i$, and samples $\bV'_\ell$ otherwise (with probability $p'_\ell := 1 - \eta - p_1 - \cdots - p_{\ell-1}$). We use the fact that
\begin{align}
(p_\ell - p'_\ell) + \sum_{i=\ell+1}^{k} p_i \le \eta. \label{eq:eta-bound}
\end{align}

Informally, $\Lambda'$ samples $\bW$ with probability $\eta$ and otherwise samples randomly from the $1-\eta$ fraction of $\Lambda$ that is closest to $\bV$.

We next justify why $\Lambda'$ is a $\delta$-spread of $\bW$. First note that if $\|\bV - \bV'\|^2_2 \le p^2\delta$ then
\[
    \|\bW -\bV'\|^2_2 \le (\|\bW-\bV\|_2 + \|\bV - \bV'\|_2)^2 \le (\eta \sqrt{\delta} + p \sqrt{\delta})^2 = (p+\eta)^2\delta.
\]
As such, for any $p \in [0,1]$, we have that
\[
\Pr_{\bV' \sim \Lambda}\left[\|\bW -\bV'\|^2_2 \le p^2\delta\right] \ge \Pr_{\bV' \sim \Lambda}\left[\|\bV -\bV'\|^2_2 \le \max(p-\eta,0)^2\delta\right] \ge \max(p-\eta, 0).
\]
Therefore,
\[
    \Pr_{\bV' \sim \Lambda'}\left[\|\bW -\bV'\|^2_2 \le p^2\delta\right] \ge \max(p-\eta, 0) + \eta \ge p.
\]
Thus, $\Lambda'$ is a $\delta$-spread of $\bW$ and thus supported in $B(\bW, \delta)$. Let $g = \underset{\bV' \sim \Lambda'}{\E}\left[S(\bV')\right].$ We have that
\begin{align*}
\|f - g\|_{2,2} &= \left\| \underset{\bV' \sim \Lambda}{\E}\left[S(\bV')\right] -  \underset{\bV' \sim \Lambda'}{\E}\left[S(\bV')\right]\right\|_{2,2}\\
&= \left\|\sum_{i \in [k]}p_i S(\bV'_i) - \eta S(\bW) - \sum_{i \in [\ell-1]} p_i S(\bV'_i) - p'_\ell S(\bV'_\ell)\right\|_{2,2}  \\
&\le (p_\ell-p'_\ell) \|S(\bV_\ell)\|_{2,2} + \sum_{i = \ell+1}^k p_i \|S(\bV'_i)\|_{2,2} + \eta \|S(\bW)\|_{2,2}\\
&\le 2\eta \max_{\bV' \in B(\bV, \delta) \cup B(\bW,\delta)} \|S(\bV')\|_{2,2}%
& \text{(by (\ref{eq:eta-bound}))}\\
&\le 2\eta,
\end{align*}
as desired. %
\end{proof}

\begin{proof}[Proof of \Cref{lem:req-eq-robust}]
If $\eta \ge 1$, then $\|\REQ^{N,Z}_{\theta,\delta}(\bV) - \REQ^{N,Z}_{\theta,\delta}(\bW)\|_{2,2}  \le \sqrt{2} < 8\sqrt{\eta}$, as desired. Thus, assume $\eta \in (0,1)$.

Fix $x \in X$ and let $f = \REQ^{N,Z}_{\theta,\delta}(\bV)$ and $g = \REQ^{N,Z}_{\theta,\delta}(\bW)$. By \Cref{prop:smooth-close}, we have that there is $f' \in B(\widehat{\BCR}^{N,Z}_{\theta,\delta}, \bV, \delta)$ and $g' \in B(\widehat{\BCR}^{N,Z}_{\theta,\delta}, \bW, \delta)$ with $\|f'-g\|_{2,2} \le 2\eta$ and $\|g'-f\|_{2,2} \le 2\eta$. Let $g'' = (g + g')/2$ and $f'' = (f + f')/2$. Assume WLOG that $\|f\|_{2,2} \ge \|g\|_{2,2}$. By \Cref{prop:l2-average} and the fact that $\|f\|_{2,2} \le 1$, we have that
\begin{align*}
\|f''\|^2_{2,2} &\le \max(\|f\|^2_{2,2},\|f'\|^2_{2,2}) - \frac{\|f-f'\|^2_{2,2}}{4}\\
&\le \max(\|f\|^2_{2,2},(\|g\|_{2,2}+2\eta)^2) - \frac{(\|f-g\|_{2,2}-2\eta)^2}{4}\\
&\le (\|f\|_{2,2}+2\eta)^2 - \frac{(\|f-g\|_{2,2}-2\eta)^2}{4}\\
&\le \|f\|^2_{2,2} + 8\eta - \frac{(\|f-g\|_{2,2}-2\eta)^2}{4}.
\end{align*}
Since $f'' \in B(\widehat{\BCR}^{N,Z}_{\theta,\delta}, \bV, \delta)$ and $f = \REQ^{N,Z}_{\theta,\delta}(\bV)$, we have that $\|f''\|^2_{2,2} \ge \|f\|^2_{2,2}$. Thus, 
\[
8\eta \ge \frac{(\|f-g\|_{2,2}-2\eta)^2}{4}.
\]
Thus, $\|f-g\|_{2,2} \le \sqrt{32\eta} + 2\eta \le 8\sqrt{\eta}$, as desired.
\end{proof}

\subsection{Correlated Rounding}\label{subsec:cor-round}

The final necessary ingredient to prove \Cref{thm:robust-eq} is to convert the \emph{independent} rounding of BCR into a \emph{correlated} rounding. More precisely, by \Cref{lem:req-eq-robust}, we know that $\REQ$ maps similar vectors to similar probability distributions. However, the final step of independent rounding in the BCR algorithm breaks this property. For example, for $A = \{0,1\}$, consider the points $p_x = (0.5, 0.5)$ and $p_y = (0.51, 0.49)$ in $\Delta_A$. The two distributions are highly similar, but with independent rounding there is a $1/2$ chance that $x$ and $y$ are rounded differently. We can make these distributions correlated by sampling a common random real number $r \sim [0, 1]$ and then rounding $x$ according to whether $r < 0.5$ and $y$ according to whether $r < 0.51$. Now, there is a $99\%$ chance $x$ and $y$ are rounded to the same value.

We can perform analogous correlated rounding in general (with $\xi_{\COR}$ corresponding the variable $r$). Let $\plur : \R^B \to B$ be any function with the property that for all $q \in \Delta_B$, $q(\plur(q)) \ge q(b)$ for all $b \in B$. That is, $\plur$ is a plurality function with ties broken arbitrarily (cf.~\Cref{sec_separablesPCSPs}).
\begin{definition}
Let $S : \S^N_{A} \to L^{2,2}(\Xi, \Delta_B)$ be a rounding scheme. Let $\Xi_{\COR,B}$ be the uniform distribution over $[0,\frac{1}{2|B|}]^B$. We then define $\S_{\COR} : \S^N_A \to L^{2,2}(\Xi \times \Xi_{\COR,B}, \Delta_B)$ as follows:
\begin{align}
S_{\COR}(\bV, \xi, \xi_{\COR}) = \plur(S(\bV, \xi) + \xi_{\COR}).\label{eq:S-COR}
\end{align} 
Note that $S_{\COR}$ only outputs elements of $B$, which we identify with the vertices of $\Delta_B$.
\end{definition}

We seek to prove two facts about (\ref{eq:S-COR}): First, if $S$ is robust, then $S_{\COR}$ is robust; second, if $\|S(\bV) - S(\bW)\|_{2,2}$ is small, then $S_{\COR}(\bV)$ and $S_{\COR}(\bW)$ have mostly equal outputs.

\begin{lemma}\label{lem:CUT-round}
Assume $S : \S^N_{A} \times \Xi \to \Delta_B$ is a rounding scheme that is $(\alpha, \beta)$-robust for a weighted instance $\cX$ of $\PCSP(\bA, \bB)$. Then, $S_{\COR}$ is $(\alpha,O_{\bB}(\beta))$-robust for $\cX$.
\end{lemma}

\begin{proof}
Let $\cV \in \V^N_{\cX,\bA}$ be an SDP solution with $\comp_{\cX,\bA}(\cV) \ge 1-\alpha$. By \Cref{def:robust-rounding-scheme}, we have that
\begin{align}
  \underset{\xi \sim \Xi}{\E}[\sound_{\cX,\B}(S(\bV_x, \xi) : x \in X)] \ge 1-\beta. \label{eq:sound-beta}
\end{align}
Let $(Y, P, Q)$ be an arbitrary clause of $\cX$ and fix $\xi \in \supp \Xi$. Assume that $\sound_{Y,Q}(S(\bV_y, \xi) : y \in Y) \ge 1-\gamma$. We claim that
\begin{align}
\underset{\bp \sim  [0,1/(2|B|)]^B}{\E}\left[\sound_{Y,Q}(S_{\COR}(\bV_y, \xi, \bp) : y \in Y)\right] \ge 1-(2|B|)^{|Y|}\gamma.\label{eq:sound-cor-gamma}
\end{align}
For notational convenience, we let $\br := (S(\bV_y, \xi) : y \in Y) \in \Delta_B^Y$. The assumption $\sound_{Y,Q}(\bq) \ge 1-\gamma$ is equivalent to
\begin{align}
  \sum_{\bq \in B^Y \setminus Q} \prod_{y \in Y} r_{y, q_{y}} \le \gamma.\label{eq:gamma-lb}
\end{align}
and we seek to prove that
\[
  \Pr_{\bp \sim [0,1/(2|B|)]^B}[(\plur(\br_y + \bp) : y \in Y) \in Q] \ge 1 - (2|B|)^{|Y|}\gamma.
\]
Consider any $\bq \in B^Y \setminus Q$ for which there exists $\bp \in [0,1/(2|B|)]^B$ for which $(\plur(\br_y + \bp) : y \in Y) = \bq$. If no such $\bq$ exists, then the LHS of (\ref{eq:sound-beta}) equals $1$.

For each $y \in Y$, in order for $\plur(\br_y + \bp) = q_y$, we must have that $r_{y,q_{y}} + \frac{1}{2|B|} \ge r_{y,b}$ for all $b \in B$. Since $\sum_{b \in B} r_{y,b} \ge 1$, this means that $r_{y,q_{y}} \ge \frac{1}{2|B|}.$ Thus, by (\ref{eq:gamma-lb}), we have that
\[
  \gamma \ge \prod_{y \in Y} r_{y, q_{y}} \ge \frac{1}{(2|B|)^{|Y|}}.
\]
Therefore, the RHS of (\ref{eq:sound-cor-gamma}) is at most zero, so the inequality is true.  To finish, we apply (\ref{eq:sound-cor-gamma}) and (\ref{eq:sound-beta}) as follows.

\begin{align*}
\underset{\xi \sim \Xi}{\E}&[\sound_{\cX,\B}(S_{\COR}(\bV_x, \xi) : x \in X)]\\
&= \underset{\substack{\xi \sim \Xi\\(Y,P,Q) \sim \cX}}{\E}[\sound_{Y,Q}(S_{\COR}(\bV_y, \xi) : y \in Y)]\\
&\ge 1 - \underset{\substack{\xi \sim \Xi\\(Y,P,Q) \sim \cX}}{\E}[|B|^{|Y|}(1 -\sound_{Y,Q}(S(\bV_y, \xi) : y \in Y))] & \text{by (\ref{eq:sound-cor-gamma})}\\
&\ge 1 - \underset{\substack{\xi \sim \Xi\\(Y,P,Q) \sim \cX}}{\E}[O_{\bB}(1)(1 -\sound_{Y,Q}(S(\bV_y, \xi) : y \in Y))] & \text{$\bB$ has bounded arity}\\
&= 1 - O_{\bB}(1)\underset{\xi \sim \Xi}{\E}[1 - \sound_{\cX,\B}(S(\bV_x, \xi) : x \in X)]\\
&\ge 1 - O_{\bB}(\beta), & \text{by (\ref{eq:sound-beta})}
\end{align*}
as desired.
\end{proof}

Next, we show that $S_{\COR}$ rounds similar local configurations to near-identical outputs.

\begin{lemma}\label{lem:cor-eq}
Let $\bV, \bW \in \S^N_A$ be such that $\|S(\bV) - S(\bW)\|_{2,2} \le \eta$. Then, 
\begin{align*}
\Pr_{\substack{{\xi \sim \Xi}\\{\bp \sim \Xi_{\COR,B}}}}&[S_{\COR}(\bV, \xi, \bp) \neq S_{\COR}(\bW, \xi, \bp)]\le O_{\bB}(\eta).
\end{align*}
\end{lemma}

\begin{proof}
We have that
\begin{align*}
\Pr_{\substack{{\xi \sim \Xi}\\{\bp \sim \Xi_{\COR,B}}}}&[S_{\COR}(\bV, \xi, \bp) \neq S_{\COR}(\bW, \xi, \bp)]\\&= \Pr_{\substack{{\xi \sim \Xi}\\{\bp \sim \Xi_{\COR,B}}}}[\plur(S(\bV, \xi) + \bp) \neq \plur(S(\bW, \xi) + \bp)]\\
&= \sum_{b \neq b' \in B} \Pr_{\substack{{\xi \sim \Xi}\\{\bp \sim \Xi_{\COR,B}}}}[\plur(S(\bV, \xi) + \bp) = b \wedge \plur(S(\bW, \xi) + \bp) = b']\\
&\le \sum_{b \neq b' \in B} \Pr_{\substack{{\xi \sim \Xi}\\{\bp \sim \Xi_{\COR,B}}}}[S(\bV, \xi)_b + p_b \ge S(\bV, \xi)_{b'} + p_{b'} \wedge S(\bW, \xi)_{b'} + p_{b'} \ge S(\bW, \xi)_{b} + p_{b}]\\
&= \sum_{b \neq b' \in B} \Pr_{\substack{{\xi \sim \Xi}\\{\bp \sim [0, 1/(2|B|)]^B}}}[S(\bV, \xi)_b - S(\bV, \xi)_{b'} \ge p_{b'} - p_{b} \ge S(\bW, \xi)_{b} - S(\bW, \xi)_{b'}]\\
&\le \sum_{b \neq b' \in B} \underset{\xi \sim \Xi}{\E}[2|B| \cdot |S(\bV, \xi)_b - S(\bV, \xi)_{b'}  - (S(\bW, \xi)_{b} - S(\bW, \xi)_{b'})|] & (*)\\
&\le 2|B|^3 \underset{\xi \sim \Xi}{\E}[\|S(\bV, \xi) - S(\bW, \xi)\|_1] & (**)\\
&\le 2|B|^3 \|S(\bV) - S(\bW)\|_{2,2} = O_{\bB}(\eta),
\end{align*}
where (*) follows from the fact that the CDF of the distribution $p_b - p_{b'}$ is bounded by $2|B|$ everywhere and (**) follows from H\"older's inequality (and the fact that $\E_{\xi \sim \Xi}[1] = 1$).
\end{proof}

\subsection{Proof of \Cref{thm:robust-eq}}

We now finally have all the pieces needed to prove \Cref{thm:robust-eq}, which we restate for convenience.

\robusteq*

\begin{proof}
By \Cref{thm:BCR}, it suffices to prove that, for every instance $\cX$ of $\PCSP((\bA, \bB)+\EQ)$ with SDP value at least $1 - 2\alpha_{\EQ}$, there exists an integral assignment $\cX \to \B$ with soundness at least $1 - \beta_{\EQ}/2$.
In particular, we do not need to present an efficient algorithm as that automatically follows from Raghavendra's theorem.

Let $\cX$ be an arbitrary instance of $\PCSP((\bA, \bB)+\EQ)$ and let $\cV \in \V^N_{\cX,\bA+\EQ}$ have completeness at least $1 - 2\alpha_{\EQ}$. Split $\cX$ into two parts: $\cX_{\bA,\bB}$ as an instance of $\PCSP(\bA, \bB)$ and $\cX_{\EQ}$ as an instance of $\PCSP(\EQ)$. Let $p \in [0,1]$ be the fraction of weight on $\cX_{\bA,\bB}$. Note that any assignment to the variables that sets all variables to the same value satisfies every constraint in $\cX_{EQ}$. Thus, the best integral assignment to $\cX$ has value at least $1 - p$. If $1 -  p \ge 1 - \beta$, we are done, so we may assume that $1 - p < 1 - \beta$ (i.e., $p > \beta > \alpha$). Note further that
\[
  1-2\alpha_{\EQ} \le \comp_{\cX,\bA+\EQ}(\cV) = p \comp_{\cX_{\bA,\bB},\bA}(\cV) + (1-p)\comp_{\cX_{\EQ},\EQ}(\cV).
\]
Since $\comp$ has range $[0,1]$, we can deduce that 
\begin{align}
\comp_{\cX_{\bA,\bB},\bA}(\cV) &\ge 1 - \frac{2\alpha_{\EQ}}{p} \ge 1 - \frac{2\alpha_{\EQ}}{\alpha},\label{eq:comp-ab}\\
\comp_{\cX_{\EQ},\EQ}(\cV) &\ge 1 -\frac{2\alpha_{\EQ}}{1-p}. \label{eq:comp-eq}
\end{align}
By \Cref{cor:BCR}, we may pick $\theta = \alpha/2$ so that, for any finite set $Z$, $\BCR^{|Z|\cdot |A| + N}_{\theta}$ is $(\alpha/2, \beta+\alpha/2)$-robust for $\cX$ (for sufficiently large $N$). Pick $|Z| \ge 2|X|^2/\alpha$ and $\delta = \Theta_{\bA}(\alpha^2)$. By \Cref{lem:req-ab-robust}, we have that the collection of rounding schemes $\REQ^{N,Z}_{\theta,\delta}$ is $(\alpha/2 - O_{\bA}(\sqrt{\delta}), \beta + \alpha/2 + \frac{|X|^2}{|Z|})$-robust, where we pick $\delta$
 sufficiently small so that $\REQ^{N,Z}_{\theta,\delta}$ is $(\alpha/4, 2\beta)$-robust. Since by (\ref{eq:comp-ab}) we have that $\comp_{\cX_{\bA,\bB},\bA}(\cV) \ge 1-2\alpha_{\EQ}/\alpha \ge 1 - \alpha/4$, we have that applying $\REQ^{N,Z}_{\theta,\delta}$ to $\bV$ satisfies at least $1-2\beta$ fraction of the constraints of $\cX_{\bA,\bB}$.
 
In particular, when $\REQ^{N,Z}_{\theta,\delta}$ is applied to $\bV$, we satisfy $1 - 2\beta$ fraction of the constraints of $\cX_{\bA,\bB}$ on average. Furthermore, by \Cref{lem:CUT-round}, $\REQ^{N,Z}_{\theta,\delta, \COR}$ is $(\alpha / 4, \Theta_{\bB}(\beta))$-robust for $\cX$. So, $\REQ^{N,Z}_{\theta,\delta,\COR}$ when applied to $\cV$ satisfies $1 - O_{\bB}(\beta)$-fraction of the constraints of $\cX_{\bA,\bB}$.

Now, we evaluate the performance of $\REQ^{N,Z}_{\theta, \delta, \COR}$ rounding $\cV$ for $\cX_{\EQ}$. Fix a constraint $(x,y)$, and assume that $\comp_{\{x,y\},\EQ}(\cV) = 1-\gamma$. Observe that
\[
  \comp_{\{x,y\},\EQ}(\cV) = 1 - \frac{\|\bV_{x} - \bV_{y}\|^2_2}{2},
\]
so $\|\bV_{x} - \bV_{y}\|^2_2 = 2\gamma$. By \Cref{lem:req-eq-robust} with $\eta = \sqrt{2\gamma/\delta}$, we have that
\[
  \|\REQ^{N,Z}_{\theta,\delta}(\bV_x) - \REQ^{N,Z}_{\theta,\delta}(\bV_y)\|_{2,2} \le 8\sqrt[4]{\frac{2\gamma}{\delta}} \le 10\sqrt[4]{\frac{\gamma}{\delta}}.
\]
Thus, by \Cref{lem:cor-eq}, we have that
\[
\Pr_{\substack{{\xi \sim \Xi^{\BCR}_{|Z|\cdot|A|+N,\theta}}\\{\bp \sim \Xi_{\COR,B}}}}[\REQ^{N,Z}_{\theta,\delta,\COR}(\bV_x, \xi, \bp) \neq \REQ^{N,Z}_{\theta,\delta,\COR}(\bV_y, \xi, \bp)]\le O_{\bB}((\gamma/\delta)^{1/4}) = O_{\bB}(\gamma^{1/4}\alpha^{-1/2}).
\]
In other words, the probability that $x$ and $y$ are rounded to the same value by applying $\REQ^{N,Z}_{\theta,\delta,\COR}$ to $\cV$ is $1 - O_{\bB}(\gamma^{1/4}\alpha^{-1/2})$. Since this soundness is a convex function of $\gamma$, we have that, if $\comp_{\cX_{\EQ},\EQ}(\cV) = 1 - \lambda$, then $\REQ^{N,Z}_{\theta,\delta,\COR}$ satisfies $1 - O_{\bB}(\lambda^{1/4}\alpha^{-1/2})$ fraction of the equality constraints on average. By (\ref{eq:comp-eq}), we have that $\lambda \le \frac{2\alpha_{\EQ}}{1-p}$, so the expected number of constraints of $\cX$ satisfied by applying $\REQ^{N,Z}_{\theta,\delta,\COR}$ to $\cV$ is at least
\begin{align*}
p (1 - O_{\bB}(\beta)) + (1-p)(1 - O_{\bB}(\lambda^{1/4}\alpha^{-1/2})) &\ge 1 - p \cdot O_{\bB}(\beta) - (1-p) O_{\bB}\left(\frac{(2\alpha_{EQ})^{1/4}}{(1-p)^{1/4}\sqrt{\alpha}}\right)\\
&\ge 1 - O_{\bB}\left(p \beta + (1-p)^{3/4} \frac{\alpha_{\EQ}^{1/4}}{\sqrt{\alpha}}\right)\\
&\ge 1 - O_{\bB}\left(\beta + \frac{\alpha_{\EQ}^{1/4}}{\sqrt{\alpha}}\right),
\end{align*}
where the last line bound $p \le 1$ and $(1-p)^{3/4} \le 1$ for any $p \in [0,1]$.
In particular, for $\alpha_{\EQ} = \Theta_{\bB}(\alpha^2\beta^4)$, we get $1 - O_{\bB}(\beta) = 1 - \beta_{\EQ}$ fraction of the constraints are satisfied, as desired.
\end{proof}

\section{Conclusion}\label{sec:concl}

In this paper, we proved a number of new results on the nature of Robust Promise CSPs. First, our integrality gap for $\OneNAE$ shows that the $(\eps, O(\log \log(1/\eps) /\log(1/\eps)))$-robust algorithm for the $\AT$ polymorphism due to Brakensiek--Guruswami--Sandeep~\cite{BGS23:stoc} is nearly tight. Second, our algorithms for the Majority, Plurality, and separable polymorphisms show that robust algorithms due to Charikar--Makarychev--Makarychev~\cite{CMM06:stoc,CMM09:talg} can be extended to quite general settings. Finally, our robust algorithm for a robust PCSP template with equality added shows that a rich theory of gadget reductions between PCSPs can be used in the theory of robust CSPs in a black-box manner. Moving forward, there are a number of directions for further exploration.

\begin{itemize}

\item In \Cref{sec:AT}, we presented an integrality gap for $\fiPCSP\OneNAE$. More generally, can we get an integrality gap for $\fiPCSP(\textsc{$a$-in-$b$-SAT}, \NAESAT)$ with similar asymptotics? In the case for which $b = 2a$, then MAJ are polymorphisms, so in general we consider $a/b \in (0,1) \setminus \{1/2\}$. %
As far as we know, one cannot use only gadget reductions to establish this broader family of gaps.

\item Recall that Brakensiek, Guruswami, and Sandeep~\cite{BGS23:stoc} constructed robust algorithms for all Boolean PCSPs with the MAJ or AT polymorphisms. In \Cref{sec_majority} and \Cref{sec_separablesPCSPs}, we greatly improved the analysis and scope of BGS's MAJ algorithm. What would such a generalization look like for AT? One possible direction is to remove the ``$\alpha$-conservative'' assumption on the map $\rho$ in \Cref{thm_robust_SDP_separable_PCSPs}.

\item Can one generalize the analysis in \Cref{sec:eq-robust} to fully resolve \Cref{ques:eq}? With the current techniques, a loss of at least $\eps^{1/2}$ seems necessary due to the use of the geometry of $L^2$-spaces (e.g., \Cref{prop:smooth-close}). Perhaps sharper asymptotics could be achieved by a direct construction of a quasirandom approximate polymorphisms for $\Gamma \cup \{\EQ\}$.

\end{itemize}

\section*{Acknowledgments}
The authors are grateful to the American Institute of Mathematics (AIM) SQuaRE program that partly funded and allowed for this collaboration, and to the AIM for their wonderful hospitality. The authors also thank Marcin Kozik for many useful discussions.

\bibliographystyle{alpha}
\bibliography{square}

\appendix

\appendix
\section{Omitted proofs}
\label{app_omitted_proofs}
In this section, we give the proof of some of the technical results stated in the previous sections of the paper.

\subsection{Omitted proofs from~\Cref{sec:AT}}
\label{subapp_omitted_sec_3}

We start by formally establishing the first statement in the following theorem (cf. the sketch of the proof given in~\Cref{subsec_discretized_IGI}).

\thmoneinthreeversusNAE*

We shall use some technical lemmas.
\begin{lemma}\label{lem:1in3actualdistribution}
If $\vec{v}_1$, $\vec{v}_2$, and $\vec{v}_3$ are unit vectors such that $\vec{v}_1 + \vec{v}_2 + \vec{v}_3 = -\vec{v}_0$ then there is a probability distribution $D$ supported on the assignments $\{(-1,-1,1),(-1,1,-1),(1,-1,-1)\}$ such that 
\begin{enumerate}
\item For all $i \in [3]$, $\E_{D}[x_i] = \vec{v}_i \cdot \vec{v}_0$.
\item For all distinct $i,j \in [3]$, $\E_{D}[{x_i}{x_j}] = \vec{v}_i \cdot \vec{v}_j$.
\end{enumerate}
\end{lemma}
\begin{proof}
Let $D$ be the probability distribution where we take the assignment $(-1,-1,1)$ with probability $\frac{1+\vec{v}_3 \cdot \vec{v}_0}{2}$, we take the assignment $(-1,1,-1)$ with probability $\frac{1+\vec{v}_2 \cdot \vec{v}_0}{2}$, and we take the assignment $(1,-1,-1)$ with probability $\frac{1+\vec{v}_1 \cdot \vec{v}_0}{2}$. We now make the following observations:
\begin{enumerate}
\item $\E_{D}[x_1] = \frac{1+\vec{v}_1 \cdot \vec{v}_0}{2} - \frac{1+\vec{v}_2 \cdot \vec{v}_0}{2} - \frac{1+\vec{v}_3 \cdot \vec{v}_0}{2} = \vec{v}_1 \cdot \vec{v}_0 - \frac{1 + (\vec{v}_1 + \vec{v}_2 + \vec{v}_3) \cdot \vec{v}_0}{2} = \vec{v}_1 \cdot \vec{v}_0$.
\item \begin{align*}
\E_{D}[{x_1}{x_2}] &= -\frac{1+\vec{v}_1 \cdot \vec{v}_0}{2} - \frac{1+\vec{v}_2 \cdot \vec{v}_0}{2} + \frac{1+\vec{v}_3 \cdot \vec{v}_0}{2} = \frac{-1 + (\vec{v}_1 + \vec{v}_2 - \vec{v}_3) \cdot (\vec{v}_1 + \vec{v}_2 + \vec{v}_3)}{2}\\
&= \frac{-1 + \vec{v}_1 \cdot \vec{v}_1 + 2\vec{v}_1 \cdot \vec{v}_2 + \vec{v}_2 \cdot \vec{v}_2 - \vec{v}_3 \cdot \vec{v}_3}{2} = \vec{v}_1 \cdot \vec{v}_2
\end{align*}
\end{enumerate}
The remaining cases follow by symmetry.
\end{proof}
\begin{lemma}\label{lem:1in3shiftingeffect}
Given unit vectors $\vec{v}_1$, $\vec{v}_2$, $\vec{v}_3$, $\vec{v}'_1$, $\vec{v}'_2$, and $\vec{v}'_3$ such that 
\begin{enumerate}
\item For all $i \in [3]$, $\vec{v}'_i \cdot \vec{v}_0 = \vec{v}_i \cdot \vec{v}_0 + \Delta_i$.
\item For all distinct $i,j \in [3]$, $\vec{v}'_i \cdot \vec{v}'_j = \vec{v}_i \cdot \vec{v}_j + \Delta_{ij}$.
\end{enumerate}
for some $\Delta_1,\Delta_2,\Delta_3,\Delta_{12},\Delta_{13},\Delta_{23} \in \mathbb{R}$, if there is a probability distribution $D$ such that 
\begin{enumerate}
\item For all $i \in [3]$, $\E_{D}[x_i] = \vec{v}_i \cdot \vec{v}_0$.
\item For all distinct $i,j \in [3]$, $\E_{D}[{x_i}{x_j}] = \vec{v}_i \cdot \vec{v}_j$.
\end{enumerate}
then there is a pseudo-distribution $D'$ (i.e., $D'$ may give negative probabilities to some events but the probabilities should still sum to $1$) such that 
\begin{enumerate}
\item For all $i \in [3]$, $\E_{D'}[x_i] = \vec{v}'_i \cdot \vec{v}_0$.
\item For all distinct $i,j \in [3]$, $\E_{D'}[{x_i}{x_j}] = \vec{v}'_i \cdot \vec{v}'_j$.
\item For each assignment $(x_1,x_2,x_3) \in \{-1,1\}^3$, the (pseudo-)probabilities of that assignment under $D'$ and $D$ differ by at most $\frac{|\Delta_1| + |\Delta_2| + |\Delta_3| + |\Delta_{12}| + |\Delta_{13}| + |\Delta_{23}|}{8}$.
\end{enumerate}
\end{lemma}
This follows from the following proposition.
\begin{proposition}
If $\vec{v}_1$, $\vec{v}_2$, and $\vec{v}_3$ are unit vectors and $p_{-1,-1,-1}$, $p_{-1,-1,1}$, $p_{-1,1,-1}$, $p_{-1,1,1}$, $p_{1,-1,-1}$, $p_{1,-1,1}$,$p_{1,1,-1}$, and $p_{1,1,1}$ are real numbers such that 
\begin{enumerate}
\item $p_{1,-1,-1} + p_{1,-1,1} + p_{1,1,-1} + p_{1,1,1} + p_{-1,-1,-1} + p_{-1,-1,1} + p_{-1,1,-1} + p_{-1,1,1} = 1$
\item $\vec{v}_1 \cdot \vec{v}_0 = p_{1,-1,-1} + p_{1,-1,1} + p_{1,1,-1} + p_{1,1,1} - p_{-1,-1,-1} - p_{-1,-1,1} - p_{-1,1,-1} - p_{-1,1,1}$.
\item $\vec{v}_2 \cdot \vec{v}_0 = p_{-1,1,-1} + p_{-1,1,1} + p_{1,1,-1} + p_{1,1,1} - p_{-1,-1,-1} - p_{-1,-1,1} - p_{1,-1,-1} - p_{1,-1,1}$.
\item $\vec{v}_3 \cdot \vec{v}_0 = p_{-1,-1,1} + p_{-1,1,1} + p_{1,-1,1} + p_{1,1,1} - p_{-1,-1,-1} - p_{-1,1,-1} - p_{1,-1,-1} - p_{1,1,-1}$.
\item $\vec{v}_1 \cdot \vec{v}_2 = p_{-1,-1,-1} + p_{-1,-1,1} + p_{1,1,-1} + p_{1,1,1} - p_{-1,1,-1} - p_{-1,1,1} - p_{1,-1,-1} - p_{1,-1,1}$.
\item $\vec{v}_1 \cdot \vec{v}_3 = p_{-1,-1,-1} + p_{-1,1,-1} + p_{1,-1,1} + p_{1,1,1} - p_{-1,-1,1} - p_{-1,1,1} - p_{1,-1,-1} - p_{1,1,-1}$.
\item $\vec{v}_2 \cdot \vec{v}_3 = p_{-1,-1,-1} + p_{1,-1,-1} + p_{-1,1,1} + p_{1,1,1} - p_{-1,-1,1} - p_{1,-1,1} - p_{-1,1,-1} - p_{1,1,-1}$.
\end{enumerate}
and $\vec{v}'_1$, $\vec{v}'_2$, and $\vec{v}'_3$ are unit vectors such that 
\begin{enumerate}
\item For all $i \in [3]$, $\vec{v}'_i \cdot \vec{v}_0 = \vec{v}_i \cdot \vec{v}_0 + \Delta_i$.
\item For all distinct $i,j \in [3]$, $\vec{v}'_i \cdot \vec{v}'_j = \vec{v}_i \cdot \vec{v}_j + \Delta_{ij}$.
\end{enumerate}
for some $\Delta_1,\Delta_2,\Delta_3,\Delta_{12},\Delta_{13},\Delta_{23} \in \mathbb{R}$ then if we take 
\begin{enumerate}
\item $p'_{-1,-1,-1} = p_{-1,-1,-1} + \frac{-\Delta_1-\Delta_2-\Delta_3 + \Delta_{12} + \Delta_{13} + \Delta_{23}}{8}$
\item $p'_{-1,-1,1} = p_{-1,-1,1} + \frac{-\Delta_1-\Delta_2+\Delta_3 + \Delta_{12} - \Delta_{13} - \Delta_{23}}{8}$
\item $p'_{-1,1,-1} = p_{-1,1,-1} + \frac{-\Delta_1+\Delta_2-\Delta_3 - \Delta_{12} + \Delta_{13} - \Delta_{23}}{8}$
\item $p'_{1,-1,-1} = p_{1,-1,-1} + \frac{\Delta_1-\Delta_2-\Delta_3 - \Delta_{12} - \Delta_{13} + \Delta_{23}}{8}$
\item $p'_{-1,1,1} = p_{-1,1,1} + \frac{-\Delta_1+\Delta_2+\Delta_3 - \Delta_{12} - \Delta_{13} + \Delta_{23}}{8}$
\item $p'_{1,-1,1} = p_{1,-1,1} + \frac{\Delta_1-\Delta_2+\Delta_3 - \Delta_{12} + \Delta_{13} - \Delta_{23}}{8}$
\item $p'_{1,1,-1} = p_{1,1,-1} + \frac{\Delta_1+\Delta_2-\Delta_3 + \Delta_{12} - \Delta_{13} - \Delta_{23}}{8}$
\item $p'_{1,1,1} = p_{1,1,1} + \frac{\Delta_1+\Delta_2+\Delta_3 + \Delta_{12} + \Delta_{13} + \Delta_{23}}{8}$
\end{enumerate}
we will have that 
\begin{enumerate}
\item $p'_{1,-1,-1} + p'_{1,-1,1} + p'_{1,1,-1} + p'_{1,1,1} + p'_{-1,-1,-1} + p'_{-1,-1,1} + p'_{-1,1,-1} + p'_{-1,1,1} = 1$
\item $\vec{v}'_1 \cdot \vec{v}_0 = p'_{1,-1,-1} + p'_{1,-1,1} + p'_{1,1,-1} + p'_{1,1,1} - p'_{-1,-1,-1} - p'_{-1,-1,1} - p'_{-1,1,-1} - p'_{-1,1,1}$.
\item $\vec{v}'_2 \cdot \vec{v}_0 = p'_{-1,1,-1} + p'_{-1,1,1} + p'_{1,1,-1} + p'_{1,1,1} - p'_{-1,-1,-1} - p'_{-1,-1,1} - p'_{1,-1,-1} - p'_{1,-1,1}$.
\item $\vec{v}'_3 \cdot \vec{v}_0 = p'_{-1,-1,1} + p'_{-1,1,1} + p'_{1,-1,1} + p'_{1,1,1} - p'_{-1,-1,-1} - p'_{-1,1,-1} - p'_{1,-1,-1} - p'_{1,1,-1}$.
\item $\vec{v}'_1 \cdot \vec{v}'_2 = p'_{-1,-1,-1} + p'_{-1,-1,1} + p'_{1,1,-1} + p'_{1,1,1} - p'_{-1,1,-1} - p'_{-1,1,1} - p'_{1,-1,-1} - p'_{1,-1,1}$.
\item $\vec{v}'_1 \cdot \vec{v}'_3 = p'_{-1,-1,-1} + p'_{-1,1,-1} + p'_{1,-1,1} + p'_{1,1,1} - p'_{-1,-1,1} - p'_{-1,1,1} - p'_{1,-1,-1} - p'_{1,1,-1}$.
\item $\vec{v}'_2 \cdot \vec{v}'_3 = p'_{-1,-1,-1} + p'_{1,-1,-1} + p'_{-1,1,1} + p'_{1,1,1} - p'_{-1,-1,1} - p'_{1,-1,1} - p'_{-1,1,-1} - p'_{1,1,-1}$.
\end{enumerate}
\end{proposition}
\begin{proof}
This follows from the following observations: 
\begin{enumerate}
\item The $\pm\frac{\Delta_i}{8}$ terms contribute $\Delta_i$ to the sum for $\vec{v}'_i \cdot \vec{v}_0$ and cancel out for the other sums.
\item Similarly, the $\pm\frac{\Delta_{ij}}{8}$ terms contribute $\Delta_{ij}$ to the sum for $\vec{v}'_i \cdot \vec{v}'_j$ and cancel out for the other sums. \qedhere
\end{enumerate}
\end{proof}
\begin{lemma}\label{lem:1in3noiseeffect}
If $\vec{v}'_1$, $\vec{v}'_2$, and $\vec{v}'_3$ are unit vectors and $D'$ is a pseudo-distribution such that
\begin{enumerate}
\item For all $i \in [3]$, $\E_{D'}[x_i] = \vec{v}'_i \cdot \vec{v}_0$.
\item For all distinct $i,j \in [3]$, $\E_{D'}[{x_i}{x_j}] = \vec{v}'_i \cdot \vec{v}'_j$.
\end{enumerate}
then if we take $\vec{v}''_1 = \sqrt{1 - {\epsilon'}^2}\vec{v}'_1 + {\epsilon'}\vec{z}_1$, $\vec{v}''_2 = \sqrt{1 - {\epsilon'}^2}\vec{v}'_2 + {\epsilon'}\vec{z}_2$, and $\vec{v}''_3 = \sqrt{1 - {\epsilon'}^2}\vec{v}'_3 + {\epsilon'}\vec{z}_3$ where $\vec{z}_1$, $\vec{z}_2$, and $\vec{z}_3$ are orthogonal to $\vec{v}'_1$, $\vec{v}'_2$, $\vec{v}'_3$, and each other for some $\epsilon' > 0$ and take $D''$ to be the pseudo-distribution where the probability that each assignment occurs under $D''$ is $(1-\epsilon')$ times the probability that the assignment occurs under $D'$ plus $\frac{\epsilon'}{8}$ then 
\begin{enumerate}
\item For all $i \in [3]$, $\E_{D''}[x_i] = \vec{v}''_i \cdot \vec{v}_0$.
\item For all distinct $i,j \in [3]$, $\E_{D''}[{x_i}{x_j}] = \vec{v}''_i \cdot \vec{v}''_j$.
\end{enumerate}
\end{lemma}
This follows from the following proposition. 
\begin{proposition}
If $\vec{v}'_1$, $\vec{v}'_2$, and $\vec{v}'_3$ are unit vectors and $p'_{-1,-1,-1}$, $p'_{-1,-1,1}$, $p'_{-1,1,-1}$, $p'_{-1,1,1}$, $p'_{1,-1,-1}$, $p'_{1,-1,1}$,$p'_{1,1,-1}$, and $p'_{1,1,1}$ are real numbers (not necessarily non-negative) such that 
\begin{enumerate}
\item $p'_{1,-1,-1} + p'_{1,-1,1} + p'_{1,1,-1} + p'_{1,1,1} + p'_{-1,-1,-1} + p'_{-1,-1,1} + p'_{-1,1,-1} + p'_{-1,1,1} = 1$
\item $\vec{v}'_1 \cdot \vec{v}_0 = p'_{1,-1,-1} + p'_{1,-1,1} + p'_{1,1,-1} + p'_{1,1,1} - p'_{-1,-1,-1} - p'_{-1,-1,1} - p'_{-1,1,-1} - p'_{-1,1,1}$.
\item $\vec{v}'_2 \cdot \vec{v}_0 = p'_{-1,1,-1} + p'_{-1,1,1} + p'_{1,1,-1} + p'_{1,1,1} - p'_{-1,-1,-1} - p'_{-1,-1,1} - p'_{1,-1,-1} - p'_{1,-1,1}$.
\item $\vec{v}'_3 \cdot \vec{v}_0 = p'_{-1,-1,1} + p'_{-1,1,1} + p'_{1,-1,1} + p'_{1,1,1} - p'_{-1,-1,-1} - p'_{-1,1,-1} - p'_{1,-1,-1} - p'_{1,1,-1}$.
\item $\vec{v}'_1 \cdot \vec{v}'_2 = p'_{-1,-1,-1} + p'_{-1,-1,1} + p'_{1,1,-1} + p'_{1,1,1} - p'_{-1,1,-1} - p'_{-1,1,1} - p'_{1,-1,-1} - p'_{1,-1,1}$.
\item $\vec{v}'_1 \cdot \vec{v}'_3 = p'_{-1,-1,-1} + p'_{-1,1,-1} + p'_{1,-1,1} + p'_{1,1,1} - p'_{-1,-1,1} - p'_{-1,1,1} - p'_{1,-1,-1} - p'_{1,1,-1}$.
\item $\vec{v}'_2 \cdot \vec{v}'_3 = p'_{-1,-1,-1} + p'_{1,-1,-1} + p'_{-1,1,1} + p'_{1,1,1} - p'_{-1,-1,1} - p'_{1,-1,1} - p'_{-1,1,-1} - p'_{1,1,-1}$.
\end{enumerate}
then if we take $\vec{v}''_1 = \sqrt{1 - {\epsilon'}^2}\vec{v}'_1 + {\epsilon'}\vec{z}_1$, $\vec{v}''_2 = \sqrt{1 - {\epsilon'}^2}\vec{v}'_2 + {\epsilon'}\vec{z}_2$, and $\vec{v}''_3 = \sqrt{1 - {\epsilon'}^2}\vec{v}'_3 + {\epsilon'}\vec{z}_3$ where $\vec{z}_1$, $\vec{z}_2$, and $\vec{z}_3$ are orthogonal to $\vec{v}'_1$, $\vec{v}'_2$, $\vec{v}'_3$, and each other then if we take $p''_{-1,-1,-1} = (1 - \epsilon')p'_{-1,-1,-1} + \frac{\epsilon'}{8}$, $p''_{-1,-1,1} = (1 - \epsilon')p'_{-1,-1,1} + \frac{\epsilon'}{8}$, $p''_{-1,1,-1} = (1 - \epsilon')p'_{-1,1,-1} + \frac{\epsilon'}{8}$, $p''_{-1,1,1} = (1 - \epsilon')p'_{-1,1,1} + \frac{\epsilon'}{8}$, $p''_{1,-1,-1} = (1 - \epsilon')p'_{1,-1,-1} + \frac{\epsilon'}{8}$, $p''_{1,-1,1} = (1 - \epsilon')p'_{1,-1,1} + \frac{\epsilon'}{8}$, $p''_{1,1,-1} = (1 - \epsilon')p'_{1,1,-1} + \frac{\epsilon'}{8}$, and $p''_{1,1,1} = (1 - \epsilon')p'_{1,1,1} + \frac{\epsilon'}{8}$ then 
\begin{enumerate}
\item $p''_{1,-1,-1} + p''_{1,-1,1} + p''_{1,1,-1} + p''_{1,1,1} + p''_{-1,-1,-1} + p''_{-1,-1,1} + p''_{-1,1,-1} + p''_{-1,1,1} = 1$
\item $\vec{v}''_1 \cdot \vec{v}_0 = p''_{1,-1,-1} + p''_{1,-1,1} + p''_{1,1,-1} + p''_{1,1,1} - p''_{-1,-1,-1} - p''_{-1,-1,1} - p''_{-1,1,-1} - p''_{-1,1,1}$.
\item $\vec{v}''_2 \cdot \vec{v}_0 = p''_{-1,1,-1} + p''_{-1,1,1} + p''_{1,1,-1} + p''_{1,1,1} - p''_{-1,-1,-1} - p''_{-1,-1,1} - p''_{1,-1,-1} - p''_{1,-1,1}$.
\item $\vec{v}''_3 \cdot \vec{v}_0 = p''_{-1,-1,1} + p''_{-1,1,1} + p''_{1,-1,1} + p''_{1,1,1} - p''_{-1,-1,-1} - p''_{-1,1,-1} - p''_{1,-1,-1} - p''_{1,1,-1}$.
\item $\vec{v}''_1 \cdot \vec{v}''_2 = p''_{-1,-1,-1} + p''_{-1,-1,1} + p''_{1,1,-1} + p''_{1,1,1} - p''_{-1,1,-1} - p''_{-1,1,1} - p''_{1,-1,-1} - p''_{1,-1,1}$.
\item $\vec{v}''_1 \cdot \vec{v}''_3 = p''_{-1,-1,-1} + p''_{-1,1,-1} + p''_{1,-1,1} + p''_{1,1,1} - p''_{-1,-1,1} - p''_{-1,1,1} - p''_{1,-1,-1} - p''_{1,1,-1}$.
\item $\vec{v}''_2 \cdot \vec{v}''_3 = p''_{-1,-1,-1} + p''_{1,-1,-1} + p''_{-1,1,1} + p''_{1,1,1} - p''_{-1,-1,1} - p''_{1,-1,1} - p''_{-1,1,-1} - p''_{1,1,-1}$.
\end{enumerate}
\end{proposition}
\begin{proof}
This follows from the following observations:
\begin{enumerate}
\item For all $i \in [3]$, $\vec{v}''_i \cdot \vec{v}_0 = (1 - {\epsilon'})\vec{v}'_i \cdot \vec{v}_0$.
\item For all distinct $i,j \in [3]$, $\vec{v}''_i \cdot \vec{v}''_j = (1 - {\epsilon'})\vec{v}'_i \cdot \vec{v}'_j$.
\item In each of the above sums, the $\frac{\epsilon'}{8}$ terms cancel out as assigning probability $\frac{\epsilon'}{8}$ to each of the possible assignments does not contribute anything to the biases or pairwise biases. \qedhere
\end{enumerate}
\end{proof}
\begin{proof}[Proof of the first statement of Theorem \ref{thm:1in3versusNAE}:]
To prove the first statement of Theorem \ref{thm:1in3versusNAE}, we observe that for each triple of vectors $\vec{v}'_1,\vec{v}'_2,\vec{v}'_3$ sampled from $T'_{x}$ or $T'_{-x}$ that is not discarded, there is a triple of unit vectors $\vec{v}_1$, $\vec{v}_2$, and $\vec{v}_3$ such that $\vec{v}_1 + \vec{v}_2 + \vec{v}_3 = \pm\vec{v}_0$ and for all $i \in [3]$, $||\frac{\vec{v}'_i}{||\vec{v}'_i||} - \vec{v}_i|| \leq \frac{\epsilon}{100}$.

To see this, observe that when we sample $\vec{w}$ and $\vec{u}$ from
$\cN(0,1/d)^d$,
we can take 
\begin{enumerate}
\item $\vec{v}_1 = bx\vec{v}_0 + \sqrt{1 - x^2}\left(-\frac{\sqrt{1 - (2|x|-1)^2}}{2\sqrt{1 - x^2}}\frac{\vec{w}}{||\vec{w}||} + \sqrt{1 - \frac{1 - (2|x|-1)^2}{4(1 - x^2)}}\frac{\vec{u}}{||\vec{u}||}\right)$,
\item $\vec{v}_2 = bx\vec{v}_0 + \sqrt{1 - x^2}\left(-\frac{\sqrt{1 - (2|x|-1)^2}}{2\sqrt{1 - x^2}}\frac{\vec{w}}{||\vec{w}||} - \sqrt{1 - \frac{1 - (2|x|-1)^2}{4(1 - x^2)}}\frac{\vec{u}}{||\vec{u}||}\right)$,
\item $\vec{v}_3 = b(1-2x)\vec{v}_0 + \sqrt{1 - (2|x|-1)^2}\frac{\vec{w}}{||\vec{w}||}$,
\end{enumerate}
where $b \in \{-1,1\}$. The adjustments we need to shift from $\vec{v}_i$ to $\frac{\vec{v}'_i}{||\vec{v}'_i||}$ are as follows:
\begin{enumerate}
\item $||w||$ and $||u||$ may differ from $1$ by up to $\frac{\epsilon}{1000}$.
\item We need to shift from the vectors we obtain based on $\vec{w}$ and $\vec{u}$ to the centers $\vec{v}'_1$, $\vec{v}'_2$, and $\vec{v}'_3$ of the regions they are contained in.
\item We need to rescale $\vec{v}'_1$, $\vec{v}'_2$, and $\vec{v}'_3$.
\end{enumerate}
It is not hard to show that these adjustments have a total length of at most $\frac{\epsilon}{100}$. 

A similar argument applies for the triples of vectors $(\vec{v}_0,\vec{v}_0,\vec{v}')$ where $\vec{v}'$ is sampled from $A'_{-x_0}$ and the triples of vectors $(-\vec{v}_0, -\vec{v}_0, \vec{v}')$ where $\vec{v}'$ is sampled from $A'_{x_0}$. For these triples, we use $(\vec{v}_1,\vec{v}_2,\vec{v}_3) = (\vec{v}_0,\vec{v}_0,-\vec{v}_0)$ and $(\vec{v}_1,\vec{v}_2,\vec{v}_3) =  = (-\vec{v}_0,-\vec{v}_0,\vec{v}_0)$ respectively.

The first statement now follows by applying Lemmas \ref{lem:1in3actualdistribution}, \ref{lem:1in3shiftingeffect}, and \ref{lem:1in3noiseeffect}. For each triple that is not discarded, letting $\vec{v}_1$, $\vec{v}_2$, and $\vec{v}_3$ be the vectors from the argument above, by Lemma \ref{lem:1in3actualdistribution}, there is a distribution $D$ supported on satisfying assignments such that 
\begin{enumerate}
\item For all $i \in [3]$, $\E_{D}[x_i] = \vec{v}_i \cdot \vec{v}_0$.
\item For all distinct $i,j \in [3]$, $\E_{D}[{x_i}{x_j}] = \vec{v}_i \cdot \vec{v}_j$.
\end{enumerate}
By Lemma \ref{lem:1in3shiftingeffect}, there is a pseudo-distribution $D'$ such that 
\begin{enumerate}
\item For all $i \in [3]$, $\E_{D'}[x_i] = \frac{\vec{v}'_i}{||\vec{v}'_i||} \cdot \vec{v}_0$.
\item For all distinct $i,j \in [3]$, $\E_{D'}[{x_i}{x_j}] = \frac{\vec{v}'_i}{||\vec{v}'_i||} \cdot \frac{\vec{v}'_j}{||\vec{v}'_j||}$.
\item For each assignment $(x_1,x_2,x_3) \in \{-1,1\}^3$, the (pseudo-)probabilities of this assignment under $D$ and $D'$ differ by at most $\frac{\epsilon}{50}$.
\end{enumerate}
By Lemma \ref{lem:1in3noiseeffect}, Letting $D''$ be the probability distribution where if an assignment has pseudo-probability $p$ under $D'$ then it has probability $(1 - \frac{\epsilon}{4})p + \frac{\epsilon}{32}$ under $D''$, we have that 
\begin{enumerate}
\item For all $i \in [3]$, $\E_{D''}[x_i] = \left(\sqrt{1 - \frac{\epsilon}{4}}\frac{\vec{v}'_i}{||\vec{v}'_i||} + \frac{\sqrt{\epsilon}}{2}\vec{z}_{\vec{v}'_i}\right) \cdot \vec{v}_0$.
\item For all distinct $i,j \in [3]$, $\E_{D''}[{x_i}{x_j}] = \left(\sqrt{1 - \frac{\epsilon}{4}}\frac{\vec{v}'_i}{||\vec{v}'_i||} + \frac{\sqrt{\epsilon}}{2}\vec{z}_{\vec{v}'_j}\right) \cdot \left(\sqrt{1 - \frac{\epsilon}{4}}\frac{\vec{v}'_j}{||\vec{v}'_j||} + \frac{\sqrt{\epsilon}}{2}\vec{z}_{\vec{v}'_j}\right)$. \qedhere
\end{enumerate}
\end{proof}
We now prove Lemma \ref{lem:Chernoffbounds}.

\lemChernoffbounds*

\begin{proof}
For the first statement, writing $||\vec{w}||^2 - 1 = \sum_{j=1}^{d}{(w_j^2 - \frac{1}{d})}$ where each $w_j$ is drawn independently from $N(0,\frac{1}{d})$, since 
\[
\E\left[e^{\beta(w_j^2 - \frac{1}{d})}\right] = \int_{-\infty}^{\infty}{\frac{\sqrt{d}}{\sqrt{2{\pi}}}e^{-\frac{d{w}^2}{2}}e^{\beta(w^2 - \frac{1}{d})}dw} =  \sqrt{\frac{d}{d-2\beta}}e^{-\frac{\beta}{d}}\int_{-\infty}^{\infty}{\frac{1}{\sqrt{2\pi}}e^{-\frac{z^2}{2}}dz} = \sqrt{\frac{d}{d-2\beta}}e^{-\frac{\beta}{d}}
\]
whenever $|\beta| < \frac{d}{2}$, if we let $\beta' = \frac{\beta}{d}$ then we have that 
\[
P\left(\sum_{j=1}^{d}{\left({w_j}^2 - \frac{1}{d}\right)} \geq t\right) \leq \frac{\E\left[e^{\beta\sum_{j=1}^{d}{({w_j}^2 - \frac{1}{d})}}\right]}{e^{{\beta}t}} \leq e^{-d({\beta'}t + \beta' + \frac{1}{2}ln(1-2\beta'))}
\]
and 
\[
P\left(\sum_{j=1}^{d}{\left({w_j}^2 - \frac{1}{d}\right)} \leq -t\right) \leq \frac{\E\left[e^{\beta\sum_{j=1}^{d}{({w_j}^2 - \frac{1}{d})}}\right]}{e^{-{\beta}t}} \leq e^{-d(-{\beta'}t + \beta' + \frac{1}{2}ln(1-2\beta'))}.
\]
whenever $|\beta'| < \frac{1}{2}$. Since $\ln(1-x) = -\sum_{j=1}^{\infty}{\frac{x^j}{j}}$, if $|x| \leq \frac{1}{2}$ then $\ln(1-x) \geq -x-x^2$ so setting $\beta' = \frac{t}{4}$, we have that $P\left(\sum_{j=1}^{d}{\left({w_j}^2 - \frac{1}{d}\right)} \geq t\right) \leq e^{-\frac{dt^2}{8}}$.
Similarly, setting $\beta' = -\frac{t}{4}$, we have that $
P\left(\sum_{j=1}^{d}{\left({w_j}^2 - \frac{1}{d}\right)} \leq -t\right) \leq e^{-\frac{dt^2}{8}}$.

For the second statement, observe that for all $\beta \geq 0$, since $\vec{w} \cdot \vec{u} = \sum_{j=1}^{d}{{w_j}{u_j}}$ and 
\[
\E\left[e^{\beta{w_j}{u_j}}\right] = \int_{-\infty}^{\infty}{\frac{\sqrt{d}}{\sqrt{2{\pi}}}e^{-\frac{d{u}^2}{2}}e^{\beta{w_j}u}du} = e^{\frac{{\beta}^2{w_j^2}}{2d}}\int_{-\infty}^{\infty}{\frac{1}{\sqrt{2\pi}}e^{-\frac{\left(z-\frac{{\beta}w_j}{\sqrt{d}}\right)^2}{2}}dz} = e^{\frac{{\beta}^2{w_j^2}}{2d}}
\]
we have that for all $t \geq 0$,
$
P(\vec{w} \cdot \vec{u} \geq t) \leq \frac{\E\left[e^{{\beta}\sum_{j=1}^{d}{{w_j}u_j}}\right]}{e^{{\beta}t}} = e^{\frac{{\beta}^2{||w||^2}}{2d} - {\beta}t}
$. Plugging in $\beta = \frac{td}{||w||^2}$, we obtain that 
$P(\vec{w} \cdot \vec{u} \geq t) \leq 2e^{-\frac{d{t^2}}{2||\vec{w}||^2}}$. By symmetry, $P(\vec{w} \cdot \vec{u} \leq -t) \leq 2e^{-\frac{d{t^2}}{2||\vec{w}||^2}}$.
\end{proof}

\subsection{Omitted proofs from~\Cref{sec:eq-robust}}
\label{subapp_omitted_sec_6}

We shall prove the following.

\lemdiffinK*

\begin{proof}
First we prove that (1) implies (2). To see why, by (1), we have that for all $\ba \in A^Y$, there exists $\lambda_{\ba} \in \R$ such that
\begin{align}
  \Mat(\cV) - \Mat(\cW) = \sum_{\ba \in A^Y} \lambda_{\ba} \Mat(\ba). \label{eq:V-W-diff}
\end{align}
Observe, using~\eqref{eqn_trace_gram_matrix_configurations}, that $\Tr(\Mat(\cV)) = \Tr(\Mat(\cW)) = \Tr(\Mat(\ba)) = |Y|$ for all $\ba \in A^Y$.
Thus, the trace of (\ref{eq:V-W-diff}) implies that $\sum_{\ba \in A^Y} \lambda_{\ba} = 0$. 
Hence, $\Mat(\cV) - \Mat(\cW) \in K_{Y,A}^0$, as desired.

We next focus on proving (1). We start by giving a self-contained proof, and then we show how the same result can be derived from a result on the projection of hypermatrices.

Let $K_{Y,A}^* : \R^{(Y \times A)^2} \to \R$ be the dual linear space of $K_{Y,A}$. It suffices to show for every $M \in K_{Y,A}^*$ that $M(\Mat(\cV)) = 0$. We can parameterize each element $M \in K_{Y,A}$ as a list of $|Y|^2$ matrices $(M^{y,y'} \in \R^{A \times A} : y, y' \in Y)$ such that, for all $\ba \in A$, we have that the condition $M(\Mat(\ba))$ is equivalent to
\begin{align}
  \sum_{y,y' \in Y} M^{y,y'}_{a_y, a_{y'}} = 0.\label{eq:M}
\end{align}

We now see to generalize (\ref{eq:M}) as follows.

\begin{claim}\label{claim:M}
For any list of vectors $(\bu_y \in \R^A : y \in Y)$ for which there exists $c \in \R$ such that for all $y \in Y$, $\sum_{a \in A}u_{y,a} = c$, then 
\begin{align}
  \sum_{\substack{y \in Y\\a \in A}} M^{y,y'}_{a,a}u_{y,a} + \sum_{y \neq y' \in Y} \bu_y^{T} M^{y,y'} \bu_{y'} = 0.\label{eq:M-u}
\end{align}
\end{claim}
\begin{proof}
Assume for now that $c = 1$. For each $a \in A$, let $\be_a \in \R^A$ be the vector such that $\be_{a,a'} = \one[a = a']$ for all $a' \in A$.  Then, for all $\ba \in A^Y$, (\ref{eq:M}) can be reinterpreted as
\[
  \sum_{y \in Y} M^{y,y}_{a_y,a_y} + \sum_{y \neq y' \in Y} \be_{a_y}^{T} M^{y,y'} \be_{a_{y'}} = 0.
\]
For each $y \in Y$, observe that $\bu_{y} = \sum_{a \in A} u_{y,a}\be_{a}$. Now, for all $\ba \in A^Y$, let $\lambda_{\ba} = \prod_{y \in Y} u_{y,a_y}$. Then, we can see that
\begin{align*}
0 &= \sum_{\ba \in A^Y} \prod_{z \in Y} u_{z,a_z} \left(\sum_{\substack{y \in Y}} M^{y,y}_{a_y,a_y} + \sum_{y \neq y' \in Y} \be_{a_y}^{T} M^{y,y'} \be_{a_{y'}}\right)\\
  &= \sum_{\substack{y \in Y\\a \in A}} M^{y,y}_{a,a} \sum_{\substack{\ba \in A^Y\\a_y = a}}\prod_{z \in Y} u_{z,a_z}+ \sum_{\substack{y \neq y' \in Y\\a,a' \in A}} \be_{a}^{T} M^{y,y'} \be_{a'} \sum_{\substack{\ba \in A^Y\\a_y = a\\a_{y'} = a'}} \prod_{z \in Y} u_{z,a_z}\\
&= \sum_{\substack{y \in Y\\a \in A}} M^{y,y}_{a,a} u_{y,a} \cdot \prod_{z \in Y \setminus \{y\}} \sum_{a_z \in A} u_{z,a_z}+ \sum_{\substack{y \neq y' \in Y\\a,a' \in A}} \be_{a}^{T} M^{y,y'} \be_{a'} \cdot u_{y,a} u_{y',a'} \cdot  \prod_{z \in Y \setminus \{y,y'\}} \sum_{a_z \in A} u_{z,a_z}\\
&= \sum_{\substack{y \in Y\\a \in A}} M^{y,y}_{a,a} u_{y,a} + \sum_{\substack{y \neq y' \in Y\\a,a' \in A}} \be_{a}^{T} M^{y,y'} \be_{a'} \cdot u_{y,a} u_{y',a'} & \text{ ($c = 1$)}\\
& \sum_{\substack{y \in Y\\a \in A}} M^{y,y'}_{a,a}u_{y,a} + \sum_{y \neq y' \in Y} \bu_y^{T} M^{y,y'} \bu_{y'},
\end{align*} 
as desired. Now consider $c \neq 1$. If $c \neq 0$ we can rescale $(\bu_y \in \R^A : y \in Y)$ so that $c = 1$. If $c = 0$, we can construct a sequence of sequences $\{(\bu^{(i)}_y \in \R^A : y \in Y) : i \in \N\}$ approaching $(\bu_y : y \in Y)$ in the $\ell_2$ norm, but each $c_i := \sum_{a \in A} u^{(i)}_{y,a}$ is nonzero. Then, we get (\ref{eq:M-u}) in the limit.
\end{proof}
For all $i \in [N]$ and $y \in Y$, let $\bz^{i,y} \in \R^A$ satisfying $z^{i,y}_a = v_{y,a,i}$, where we recall that $\bv_{y,a}$ is the $a$-th vector defining $\bV_y \in \S^N_A$. Since $\sum_{a \in A} \bv_{y,a} = \bv_0$ for all $y \in Y$, we have for all $y \in Y$ and $i \in [N]$ that $\sum_{a \in A} z^{i,y}_a = v_{0,i}$. Assume that $v_{0,i} \neq 0$ for all $i \in [N]$. Then, apply \Cref{claim:M} with $(z^{i,y}/v_{0,i} : y \in Y)$ for all $i \in [N]$ to get that
\begin{align*}
0 &= \sum_{i=1}^N v_{i,0}^2\left(\sum_{\substack{y \in Y\\a \in A}} M^{y,y'}_{a,a}\frac{z^{i,y}_{a}}{v_{0,i}} + \sum_{y \neq y' \in Y} \frac{1}{v_{i,0}^2}(\bz^{i,y})^{T} M^{y,y'} \bz^{i,y'}\right)\\
&= \sum_{\substack{y \in Y\\a \in A}} M^{y,y'}_{a,a} \sum_{i=1}^N v_{y,a,i}v_{0,i} + \sum_{y \neq y' \in Y} \sum_{i=1}^N (\bz^{i,y})^{T} M^{y,y'} \bz^{i,y'}\\
&= \sum_{\substack{y \in Y\\a \in A}} M^{y,y'}_{a,a} \langle \bv_{y,a}, \bv_0\rangle + \sum_{y \neq y' \in Y} \sum_{a,a' \in A}M^{y,y'}_{a,a'}\sum_{i=1}^N v_{y,a,i}v_{y',a',i}\\
&= \sum_{\substack{y \in Y\\a \in A}} M^{y,y'}_{a,a} \langle \bv_{y,a}, \bv_{y,a}\rangle + \sum_{y \neq y' \in Y} \sum_{a,a' \in A}M^{y,y'}_{a,a'}\langle \bv_{y,a}, \bv_{y',a'}\rangle\\
&= \sum_{y, y' \in Y} \sum_{a,a' \in A}M^{y,y'}_{a,a'}\langle \bv_{y,a}, \bv_{y',a'}\rangle\\
&= M(\Mat(\cV)),
\end{align*}
as desired, so $\cV \in K_{Y,A}$.
\end{proof}

\begin{proof}[Alternative proof of~\Cref{lem:diff-in-K} (1)]
    We shall use a result on the projection of hypermatrices proved in~\cite{ciardo2023approximate}. For $n,q\in\N$, consider a set $\mathcal{M}=\{M^{y,y'}:1\leq y<y'\leq q\}$ of $n\times n$ matrices with integral entries. We say that $\mathcal{M}$ is a \textit{realisable} system if the matrices of $\mathcal{M}$ are the 2-dimensional projections of a unique hypermatrix; formally, if 
    there exists a hypermatrix $\Lambda\in\Z^{\tiny\underbrace{n\times n\times\dots\times n}_q}$ such that 
    \begin{align*}
        M^{y,y'}_{a,b}=\sum_{\substack{\bz\in [n]^q\\z_{y}=a,\,z_{y'}=b}}\Lambda_{\bz}
    \end{align*}
    for each $M^{y,y'}\in\mathcal{M}$ and each $a,b\in [n]$. We say that $\mathcal{M}$ is a \textit{realistic} system if its matrices are locally compatible in the following sense: For each $1\leq y_1<y_2\leq q$, each $1\leq w_1<w_2\leq q$, and each $r,s\in [2]$ such that $y_r=w_s$, it holds that
    \begin{align*}
        (M^{y_1,y_2})^{\circ r}\bone=(M^{w_1,w_2})^{\circ s}\bone,
    \end{align*}
    where we use the notation $M^{\circ 1}\coloneqq M$ and $M^{\circ 2}\coloneqq M^\top$. We will use the following result.
    \begin{theorem}[\cite{ciardo2023approximate}]
    \label{thm_realistic_realisable}
        A system $\mathcal{M}$ is realistic if and only if it is realisable.\footnote{This result was proved in~\cite{ciardo2023approximate} in the case of arbitrarily dimensional projections and hypermatrices that are not necessarily cubic (i.e., whose modes have possibly distinct lengths). Here we only need the restricted case described above.}
    \end{theorem}
Take now some global configuration $\cV\in (\S^N_A)^Y$, and let $\Mat(\cV)\in\R^{(Y\times A)^2}$ be the corresponding Gram matrix. Suppose first that $\Mat(\cV)$ has rational entries, and let $N$ be a common denominator. Denote by $M^{y,y'}$ the $(y,y')$-th block of $N\cdot\Mat(\cV)$ (of size $A\times A$), and observe that the quantity
\begin{align*}
    (M^{y,y'}\bone)_a=N\sum_{b\in A}\ang{\bv_{y,a}}{\bv_{y',b}}
    =
    N\ang{\bv_{y,a}}{\bv_{0}}
\end{align*}
is independent on $y'$ for each $a\in A$. Hence, the system $\mathcal{M}=\{M^{y,y'}:1\leq y<y'\leq |Y|\}$ is realistic. By~\Cref{thm_realistic_realisable}, it follows that $\mathcal{M}$ is realisable. Let $\Lambda$ be its realisation hypermatrix. We claim that $\frac{1}{N}\Lambda$ witnesses that $\Mat(\cV)$ lies in $K_{Y,A}$. I.e., we claim that $\Mat(\cV)=\sum_{\bz\in A^Y}\frac{1}{N}\Lambda_\bz\Mat(\bz)$. Indeed, for each $1\leq y<y'\leq |Y|$ and each $a,b\in A$ we have
\begin{align*}
    N\cdot\Mat(\cV)_{(y,a),(y',b)}
    =
    M^{y,y'}_{a,b}
    =
    \sum_{\substack{\bz\in A^{|Y|}\\z_y=a,\,z_{y'}=b}}\Lambda_{\bz}
    =
    \sum_{\bz\in A^{|Y|}}\Lambda_{\bz}\Mat(\bz)_{(y,a),(y',b)}
    =
    \left(\sum_{\bz\in A^{|Y|}}\Lambda_{\bz}\Mat(\bz)\right)_{(y,a),(y',b)}
\end{align*}
as needed. Since $M^{y',y}=(M^{y,y'})^\top$ and since the diagonal blocks of the Gram matrix of any global configuration are completely determined by its off-diagonal blocks, the result holds for each $y,y'\in Y$, thus proving the claim. Finally, the rational subset of Gram matrices of global configurations in $(\S^N_A)^Y$ is dense, so the same result must hold for non-rational $\Mat(\cV)$ as well, thus concluding the proof.
\end{proof}

\end{document}